\documentclass[11pt]{article}
\usepackage[utf8]{inputenc}
\usepackage[T1]{fontenc}
\usepackage[dvipsnames]{xcolor}
\usepackage{lmodern}
\usepackage{authblk}
\usepackage{amsmath}
\usepackage{amssymb}
\usepackage{amsthm}
\usepackage{bbm}
\usepackage{verbatim}
\usepackage{tikz}
\usepackage{subcaption}
\usepackage{nicematrix}
\usepackage{changepage}
\usepackage{xcolor}
\usepackage{pstricks,multido,pst-plot}
\usepackage{tcolorbox,pgfkeys}

\newtheorem{theorem}{Theorem}

\newtheorem{proposition}[theorem]{Proposition}
\newtheorem{lemma}[theorem]{Lemma}
\newtheorem{corollary}[theorem]{Corollary}

\usepackage[titletoc,title]{appendix}
\usepackage{graphicx}
\usepackage{hyperref} 
\hypersetup{colorlinks, linkcolor=blue, citecolor=red, urlcolor=blue}

\usepackage[titles]{tocloft}

\usepackage{titlesec}
\titlelabel{\thetitle. } 

\definecolor{lightgray}{gray}{0.9}
\definecolor{lightyellow}{rgb}{.95,.96,.7}
\definecolor{myblue}{rgb}{.39,.54,0.82}
\definecolor{midblue}{rgb}{.07,.40,.98}
\definecolor{lightblue}{rgb}{0.93,0.95,1.0}
\definecolor{mygrey}{rgb}{.75,.75,.75}
\definecolor{lightred}{rgb}{1.,0.33,0.09}
\definecolor{mygreen}{rgb}{0.40,0.75,0.16}
\definecolor{darkgreen}{rgb}{0.32,0.56,0.09}
\definecolor{mybrown}{rgb}{0.69,0.49,0.30}
\definecolor{mypink}{rgb}{1.00,0.22,0.57}
\definecolor{mymagenta}{rgb}{0.99,0.30,0.99}

\usepackage[top=2.5cm, bottom=2.5cm, left=2cm, right=2cm]{geometry}
\numberwithin{equation}{section} 
\numberwithin{figure}{section}
\numberwithin{theorem}{section}

\newcommand{\e}{ {\rm e} }
\newcommand{\leftArr}[1]{\overset{\text{\tiny$\leftarrow$}}{#1}}

\renewcommand{\ge}{\geqslant}
\renewcommand{\le}{\leqslant}
\newcommand{\chit}{\protect\raisebox{0.25ex}{$\chi$}}

\newcommand*\squeeze[2]{
	\mbox{$\displaystyle\squeezespaces{#1} #2 $}
}
\newcommand*\squeezespaces[1]{
	\thickmuskip=\scalemuskip{\thickmuskip}{#1}%
	\medmuskip=\scalemuskip{\medmuskip}{#1}%
	\thinmuskip=\scalemuskip{\thinmuskip}{#1}%
	\nulldelimiterspace=#1\nulldelimiterspace%
	\scriptspace=#1\scriptspace%
}
\newcommand*\scalemuskip[2]{%
	\muexpr #1*\numexpr\dimexpr#2pt\relax\relax/65536\relax
}

\begin{document}
	\Large
	\begin{center}
		\textbf{GUE-corners process in two-periodic Aztec diamonds}\\ 
		
		\vspace{10pt}
		
		\large
		Nicolas \textsc{Robert} and Philippe \textsc{Ruelle}
		\\
		
		\vspace{10pt}
		
		\small  
		\textit{Institut de Recherche en Math\'ematique et Physique \\
			Universit\'e catholique de Louvain, Louvain-la-Neuve, B-1348, Belgium}\\
		
		\vspace{10pt}
		
		\texttt{nicolas.robert@uclouvain.be \qquad philippe.ruelle@uclouvain.be}
		
	\end{center}

	\vspace{10pt}
	
	\small
	\begin{center}
		\textbf{Abstract}
	\end{center}

	\begin{adjustwidth}{1cm}{1cm}
		
		Links between uniform Aztec diamonds and random matrices are numerous in the literature. In particular, it has been established that, after appropriate rescaling, the probability density function of a certain subclass of dominos converges to the GUE-corners (GUE minor) process in the large size limit. We are interested to see whether this result holds when we modify the probability measure on the space of configurations. In the first part, we look at the case of biased Aztec diamonds, where different weights are associated to vertical and horizontal dominos. In the second part, we examine the case of two-periodic weightings. In both cases, we give exact expressions for the discrete probability distributions of the particles on the levels close to the contact point (turning point), valid for any finite Aztec diamond. In the scaling limit, we prove the convergence to GUE-corners with a rescaling that depends on the weighting. We include a discussion of our results in the light of recent results by Berggren and Bradinoff, who found, in a very close setting, a marked GUE-corners process rather than the classical GUE-corners process.

	\end{adjustwidth} 
	\normalsize
	
	\vspace{0.7cm}\hrule\vspace{0.1cm}
	
	\tableofcontents
	
	\section{Introduction}
	
	Aztec diamonds (ADs) were originally studied in 1992 \cite{elkies1992alternating,elkies1992alternating2}, and interest in the model rapidly grew in the following years. A main motivation for their study comes from the fact that it is one of the simplest cases where an arctic phenomenon (i.e.\ a spatial transition between statistically ordered and disordered regions, see Figure~\ref{fig:AD definition}) takes place \cite{jockusch1998random}. Moreover, ADs show connections with other statistical models such as the six-vertex model, Alternating Sign Matrices (ASMs), lozenge tilings of hexagons, etc.
	
	\begin{figure}[ht!]
		\centering
		\includegraphics[width=0.3\linewidth]{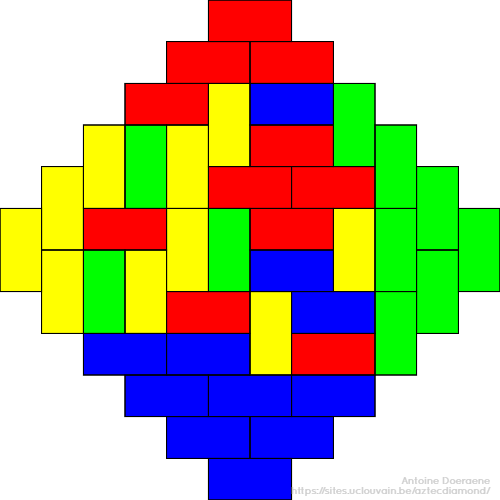}
		\hspace{3cm}
		\begin{tikzpicture}[>=stealth]
			\node[anchor=center,inner sep=0] at (0,0) {\includegraphics[width=0.3\linewidth]{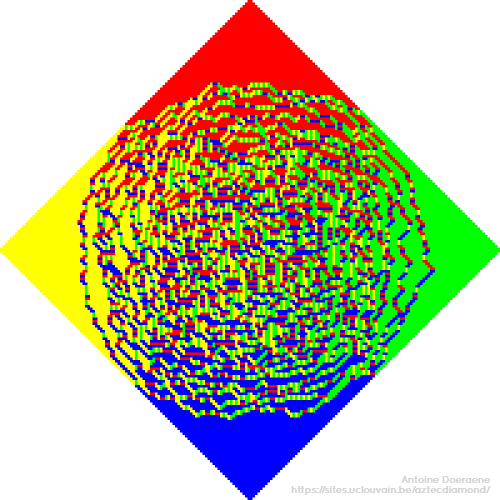}};
			\draw[line width=1.5pt](-1.3,1.3)  circle (10pt)
			(-1.5,1.5)node[anchor=south east] {GUE-corners};
		\end{tikzpicture}
		\caption{\textit{Left}: Aztec diamond of order $n=6$. \textit{Right}: Aztec diamond of order $n=100$ and arctic phenomenon. The order $n$ of an Aztec diamond corresponds to the number of corners along one edge of the domain.}
		\label{fig:AD definition}
	\end{figure}
	
	Another remarkable feature is the relation between point processes occurring in Aztec diamonds and random matrix theory. In particular, it is known that near the contact point between the arctic curve and the edge of an Aztec diamond, the distribution of a certain subclass of dominos converges to the GUE-corners process \cite{johansson2006eigenvalues, fleming2011interlaced} in the case of uniformly weighted configurations. A similar phenomenon has been observed in lozenge tilings \cite{johansson2006eigenvalues,Gorin, mkrtchyan_q-lozenge}, in plane partitions \cite{mkrtchyan2021turning, mkrtchyan2014plane}, in the six-vertex model \cite{dimitrov2020six, dimitrov2022gue}, and in ASMs \cite{gorin2014alternating} for various probability measures.
	
	In the present work, our main objective is to emphasize the universality of the GUE-corners process and prove that it continues to be the limiting process when the measure on Aztec diamond domino tilings is no longer uniform. More specifically, we start, as a way to introduce the method, with the simpler case of a one-periodic biased measure, which assigns different weights to horizontal and vertical dominos (or dimers). We then go on with the more complicated case of a two-periodic measure. We take advantage of the determinantal structure to obtain exact expressions for the distribution of the point process in the discrete setting (before taking the large size limit).
	
	Unlike the approach followed in \cite{johansson2006eigenvalues}, which uses the correlation kernel of the discrete point process and then computes its scaling limit, our approach is closer to that of \cite{fleming2011interlaced}, and is based directly on the full particle distribution induced by the measure on domino configurations. For any finite integer $m$, and for any finite Aztec diamond of order larger than $m$, the exact particle distribution on the first $m$ levels is computed by summing over all particle positions except those on the first $m$ levels. By keeping $m$ fixed as the order of the Aztec diamond goes to infinity, the scaling limit of that discrete distribution is shown to reproduce the well-known distributions of the GUE-corners process. This way of proceeding appears to us simpler than an approach based on the correlation kernel, in which the asymptotics may be harder to obtain. It also has the advantage to provide exact expressions for the discrete particle distributions (see Section~\ref{sec:marking} for plots), which otherwise would probably be hard to compute from the discrete kernel.
	
	The paper is structured as follows: we start by describing our method in the simple (yet interesting) biased case. The first step consists in properly identifying the determinantal point process of interest. After that, we deduce the probability measure associated with this point process from the biased measure on Aztec diamonds. We show how it is possible to integrate recursively the probability density function of this process to obtain the distribution on a subset of points. Finally, we prove the convergence in distribution of this probability density to the GUE-corners process by extracting the dominant order in $n$ of the distribution.
	
	In Section~\ref{sec:Two-periodic Aztec diamonds}, we repeat the same procedure for two-periodic Aztec diamonds. This case is technically more complicated, but the determinantal structure is preserved, and this is sufficient for our purposes.
	
	Section~\ref{sec:marking} is new, and was not included in the first version of our work. Indeed, a few months after the manuscript was finished, a paper by Berggren and Bradinoff \cite{berggren2026marked} has appeared. It examines essentially the same question as here, namely the presence of the GUE-corners process near a contact point of the rough region with a boundary of the Aztec diamond. A main difference is that the work \cite{berggren2026marked} is far more general since the measure defined on the domino configurations is $2 \times \ell$ periodic, the case discussed here corresponding to $\ell=2$. Another notable difference is that the process found in the scaling limit is not the classical GUE-corners process, as found here, but a marked GUE-corners process, the marking being the trace left in the continuum by the two different statistics attached to the particles according to the parity of their position.
	
	As the marking did not appear in our own work, we have added a new section to discuss the conditions under which a marking does or does not appear. It turns out that the presence or absence of marking appears to be conditioned by two different factors: the two-periodic weights and the specific particle system one considers. The weights considered in \cite{berggren2026marked} in the case $\ell=2$ are identical to those considered in the present work. However the particle system is slightly different, and this is what causes the marking. A discussion clarifying the effects of the above two factors is presented in Section~\ref{sec:marking}.

	\section{Biased Aztec diamonds}
	\label{sec:Biased Aztec diamonds}
	
	The name \textit{biased Aztec diamonds} refers to a specific choice of measure over the set of configurations where horizontal and vertical dominos do not necessarily carry the same weights. One can recover the \textit{uniform} Aztec diamond by setting the same weight on all dominos. Biased ADs have been largely studied in the literature and are used here as a toy example both for reviewing the method in a simple case and for introducing most of the notation. The spirit of our method is mainly inspired by \cite{fleming2011interlaced}, yet some steps have been modified to better fit to the two-periodic case we analyse right after.
	
	First of all, let us mention that it is not useful to consider weights on both horizontal and vertical dominos; since their total number is fixed, only the ratio between the two weights is relevant. We therefore associate a weight $1$ to horizontal dominos and $\sqrt{\lambda}$ (called the bias) to the vertical ones. The weight of a configuration $c$ is obtained by taking the product over all domino weights, $w(c)=\lambda^{N_v(c)/2}$ where $N_v(c)$ is the number of vertical dominos in the configuration (this number is always even). The probability of observing this configuration is defined as
	\begin{equation}
		\mathbb{P}(c) = \frac{w(c)}{Z_n},
		\qquad
		Z_n = \sum_{c} w(c),
	\end{equation}
	where $n$ refers to the order of the Aztec diamond, $Z_n$ is the partition function, and where the sum runs over all configurations.

	\subsection{Interlaced particles and adjacency}
	\label{sec:Interlaced particles}
	
	The main focus of this paper is the study of an underlying (interlaced) particle system shown in Figure~\ref{fig:IPS}. In this picture, the dashed lines give a coordinate system for the location of the cells with a given parity (e.g.\ if we think of the dominos as lying on a checkerboard-type domain, these would correspond, say, to the coordinates of the black cells). We choose to colour in white (resp.\ black) the dominos having their right (resp.\ left) part crossing the dashed lines. In the left part of each black domino, we also draw a red dot called \textit{particle} in this description. We use the notation $x_i^{(\ell)}=0,1,\dots,n$ to designate the vertical position of the $i$-th particle on the \textit{level}\footnote{We borrow this terminology from \cite{duits2020two}.} $\ell$. For instance, one can see in Figure~\ref{fig:IPS} that there is one particle on the level $\ell=1$ and $x_1^{(1)}=4$. There are two particles on the second level with $x_1^{(2)}=3$ and $x_2^{(2)}=6$, etc. In addition, one can remark that the particle system is interlaced in the sense that, for all $1\le i\le\ell\le n$, we have $x_i^{(\ell+1)} \le x_i^{(\ell)} \le x_{i+1}^{(\ell+1)}$. More details about this observation are given at the beginning of Section~\ref{sec:PDF 2P}. The notation $x^{(\ell)}$ has to be interpreted as the collection of all $x_i^{(\ell)}$ for $i=1,\dots,\ell$. In what follows, we will slightly abuse notation, and will use $x_i^{(\ell)}$ to denote the particle at position $(\ell,x_i^{(\ell)})$.

	\begin{figure}[!ht]
		\begin{subfigure}{0.64\textwidth}
			\centering
			\begin{tikzpicture}[scale=0.68, rotate=45]
				\draw[line width=2pt] (0,8) -- ++(1,0) -- ++(0,-1) -- ++(1,0) -- ++(0,-1) -- ++(1,0) -- ++(0,-1) -- ++(1,0) -- ++(0,-1) -- ++(1,0) -- ++(0,-1) -- ++(1,0) -- ++(0,-1) -- ++(1,0) -- ++(0,-1) -- ++(1,0) -- ++(0,-2) -- ++(-1,0) -- ++(0,-1) -- ++(-1,0) -- ++(0,-1) -- ++(-1,0) -- ++(0,-1) -- ++(-1,0) -- ++(0,-1) -- ++(-1,0) -- ++(0,-1) -- ++(-1,0) -- ++(0,-1) -- ++(-1,0) -- ++(0,-1) -- ++(-1,0) ;
				\draw[line width=2pt] (0,8) -- ++(-1,0) -- ++(0,-1) -- ++(-1,0) -- ++(0,-1) -- ++(-1,0) -- ++(0,-1) -- ++(-1,0) -- ++(0,-1) -- ++(-1,0) -- ++(0,-1) -- ++(-1,0) -- ++(0,-1) -- ++(-1,0) -- ++(0,-1) -- ++(-1,0) -- ++(0,-2) -- ++(1,0) -- ++(0,-1) -- ++(1,0) -- ++(0,-1) -- ++(1,0) -- ++(0,-1) -- ++(1,0) -- ++(0,-1) -- ++(1,0) -- ++(0,-1) -- ++(1,0) -- ++(0,-1) -- ++(1,0) -- ++(0,-1) -- ++(1,0) ;
				\begin{scope}[black!64,line width=2pt]
				\end{scope}
				
				\filldraw[fill=white, draw=black]
				(-1,8) rectangle +(2,-1)
				(-2,7) rectangle +(2,-1)
				(0,7) rectangle +(2,-1)
				(-3,6) rectangle +(2,-1)
				(-1,6) rectangle +(2,-1)
				(1,6) rectangle +(2,-1)
				(-4,5) rectangle +(2,-1)
				(0,5) rectangle +(2,-1)
				(2,5) rectangle +(2,-1)
				(-4,3) rectangle +(2,-1)
				(-2,3) rectangle +(2,-1)
				(1,2) rectangle +(2,-1)
				(-2,1) rectangle +(2,-1)
				(2,1) rectangle +(2,-1)
				(-5,0) rectangle +(2,-1)
				(-3,0) rectangle +(2,-1)
				(1,-4) rectangle +(2,-1)
				
				(-2,5) rectangle +(1,-2)
				(-5,4) rectangle +(1,-2)
				(-6,3) rectangle +(1,-2)
				(0,3) rectangle +(1,-2)
				(-7,2) rectangle +(1,-2)
				(-5,2) rectangle +(1,-2)
				(-3,2) rectangle +(1,-2)
				(-8,1) rectangle +(1,-2)
				(-6,1) rectangle +(1,-2)
				(0,1) rectangle +(1,-2)
				(-7,0) rectangle +(1,-2)
				(-1,0) rectangle +(1,-2)
				(-6,-1) rectangle +(1,-2)
				(-2,-1) rectangle +(1,-2)
				(0,-1) rectangle +(1,-2)
				(2,-1) rectangle +(1,-2)
				(-5,-2) rectangle +(1,-2)
				(-1,-2) rectangle +(1,-2)
				(0,-3) rectangle +(1,-2);
				
				\filldraw[fill=black!65!white, draw=black]
				(-4,4) rectangle +(2,-1)
				(0,4) rectangle +(2,-1)
				(2,4) rectangle +(2,-1)
				(1,3) rectangle +(2,-1)
				(-2,2) rectangle +(2,-1)
				(2,0) rectangle +(2,-1)
				(-5,-1) rectangle +(2,-1)
				(-3,-3) rectangle +(2,-1)
				(1,-3) rectangle +(2,-1)
				(-4,-4) rectangle +(2,-1)
				(-2,-4) rectangle +(2,-1)
				(-3,-5) rectangle +(2,-1)
				(-1,-5) rectangle +(2,-1)
				(1,-5) rectangle +(2,-1)
				(-2,-6) rectangle +(2,-1)
				(0,-6) rectangle +(2,-1)
				(-1,-7) rectangle +(2,-1)
				
				(-1,5) rectangle +(1,-2)
				(4,4) rectangle +(1,-2)
				(3,3) rectangle +(1,-2)
				(5,3) rectangle +(1,-2)
				(-4,2) rectangle +(1,-2)
				(4,2) rectangle +(1,-2)
				(6,2) rectangle +(1,-2)
				(1,1) rectangle +(1,-2)
				(5,1) rectangle +(1,-2)
				(7,1) rectangle +(1,-2)
				(4,0) rectangle +(1,-2)
				(6,0) rectangle +(1,-2)
				(-3,-1) rectangle +(1,-2)
				(1,-1) rectangle +(1,-2)
				(3,-1) rectangle +(1,-2)
				(5,-1) rectangle +(1,-2)
				(-4,-2) rectangle +(1,-2)
				(4,-2) rectangle +(1,-2)
				(3,-3) rectangle +(1,-2);
				
				\begin{scope}[violet,dashed,line width=0.3pt]
					\foreach \x in {1,2,3,4,5,6,7,8}
					\draw (\x,9-\x) -- +(-9,-9) node[anchor=north] {$\ell$=\x};
				\end{scope}
				
				\begin{scope}[violet,dashed,line width=0.3pt]
					\foreach \x in {0,1,2,3,4,5,6,7,8}
					\draw (-8+\x,\x) node[anchor=east]{$k$=\x \;\;\;} -- +(8,-8)  ;
				\end{scope}
				
				\filldraw[fill=red!85!white, draw=black!100!white]
				(-4,4) +(0.5,-0.5) circle (0.16)
				(0,4) +(0.5,-0.5) circle (0.16)
				(2,4) +(0.5,-0.5) circle (0.16)
				(1,3) +(0.5,-0.5) circle (0.16)
				(-2,2) +(0.5,-0.5) circle (0.16)
				(2,0) +(0.5,-0.5) circle (0.16)
				(-5,-1) +(0.5,-0.5) circle (0.16)
				(-3,-3) +(0.5,-0.5) circle (0.16)
				(1,-3) +(0.5,-0.5) circle (0.16)
				(-4,-4) +(0.5,-0.5) circle (0.16)
				(-2,-4) +(0.5,-0.5) circle (0.16)
				(-3,-5) +(0.5,-0.5) circle (0.16)
				(-1,-5) +(0.5,-0.5) circle (0.16)
				(1,-5) +(0.5,-0.5) circle (0.16)
				(-2,-6) +(0.5,-0.5) circle (0.16)
				(0,-6) +(0.5,-0.5) circle (0.16)
				(-1,-7) +(0.5,-0.5) circle (0.16)
				
				(-1,5) +(0.5,-0.5) circle (0.16)
				(4,4) +(0.5,-0.5) circle (0.16)
				(3,3) +(0.5,-0.5) circle (0.16)
				(5,3) +(0.5,-0.5) circle (0.16)
				(-4,2) +(0.5,-0.5) circle (0.16)
				(4,2) +(0.5,-0.5) circle (0.16)
				(6,2) +(0.5,-0.5) circle (0.16)
				(1,1) +(0.5,-0.5) circle (0.16)
				(5,1) +(0.5,-0.5) circle (0.16)
				(7,1) +(0.5,-0.5) circle (0.16)
				(4,0) +(0.5,-0.5) circle (0.16)
				(6,0) +(0.5,-0.5) circle (0.16)
				(-3,-1) +(0.5,-0.5) circle (0.16)
				(1,-1) +(0.5,-0.5) circle (0.16)
				(3,-1) +(0.5,-0.5) circle (0.16)
				(5,-1) +(0.5,-0.5) circle (0.16)
				(-4,-2) +(0.5,-0.5) circle (0.16)
				(4,-2) +(0.5,-0.5) circle (0.16)
				(3,-3) +(0.5,-0.5) circle (0.16);
				
			\end{tikzpicture}
		\end{subfigure}
		\begin{subfigure}{0.32\textwidth}
			\centering
			\begin{tikzpicture}[scale=0.65, rotate=45]
				\draw[line width=2pt] (0,3) -- ++(1,0) -- ++(0,-1) -- ++(1,0) -- ++(0,-1) -- ++(1,0) -- ++(0,-2) -- ++(-1,0) -- ++(0,-1) -- ++(-1,0) -- ++(0,-1) -- ++(-1,0) ;
				\draw[line width=2pt] (0,3) -- ++(-1,0) -- ++(0,-1) -- ++(-1,0) -- ++(0,-1) -- ++(-1,0) -- ++(0,-2) -- ++(1,0) -- ++(0,-1) -- ++(1,0) -- ++(0,-1) -- ++(1,0) ;
				\filldraw[fill=white, draw=black]
				(-1,3) rectangle +(2,-1)
				(-1,1) rectangle +(2,-1)
				(0,0) rectangle +(2,-1)
				
				(-3,1) rectangle +(1,-2)
				(-2,2) rectangle +(1,-2)
				(-2,0) rectangle +(1,-2)
				(-3,1) rectangle +(1,-2)
				(-3,1) rectangle +(1,-2);
				
				\filldraw[fill=black!65!white, draw=black]
				(-1,2) rectangle +(2,-1) 
				(-1,-2) rectangle +(2,-1)
				(0,-1) rectangle +(2,-1)
				
				(1,2) rectangle +(1,-2)
				(2,1) rectangle +(1,-2)
				(-1,0) rectangle +(1,-2);
				
				\filldraw[fill=red!85!white, draw=black!100!white]
				(-1,2)+(0.5,-0.5) circle (0.16)
				(-1,-2)+(0.5,-0.5) circle (0.16)
				(0,-1)+(0.5,-0.5) circle (0.16)
				
				(1,2)+(0.5,-0.5) circle (0.16)
				(2,1)+(0.5,-0.5) circle (0.16)
				(-1,0)+(0.5,-0.5) circle (0.16);

				\draw[SkyBlue, line width=2pt] (-1,2) rectangle +(2,-2);
				
			\end{tikzpicture}
			
			\vspace{0.7cm}
			
			\begin{tikzpicture}[scale=0.65, rotate=45]
				\draw[line width=2pt] (0,3) -- ++(1,0) -- ++(0,-1) -- ++(1,0) -- ++(0,-1) -- ++(1,0) -- ++(0,-2) -- ++(-1,0) -- ++(0,-1) -- ++(-1,0) -- ++(0,-1) -- ++(-1,0) ;
				\draw[line width=2pt] (0,3) -- ++(-1,0) -- ++(0,-1) -- ++(-1,0) -- ++(0,-1) -- ++(-1,0) -- ++(0,-2) -- ++(1,0) -- ++(0,-1) -- ++(1,0) -- ++(0,-1) -- ++(1,0) ;
				\filldraw[fill=white, draw=black]
				(-1,3) rectangle +(2,-1)
				(0,0) rectangle +(2,-1)
				
				(0,2) rectangle +(1,-2)
				(-3,1) rectangle +(1,-2)
				(-2,2) rectangle +(1,-2)
				(-2,0) rectangle +(1,-2)
				(-3,1) rectangle +(1,-2)
				(-3,1) rectangle +(1,-2);
				
				\filldraw[fill=black!65!white, draw=black]
				(-1,-2) rectangle +(2,-1)
				(0,-1) rectangle +(2,-1)
				
				(-1,2) rectangle +(1,-2)
				(1,2) rectangle +(1,-2)
				(2,1) rectangle +(1,-2)
				(-1,0) rectangle +(1,-2);
				
				\filldraw[fill=red!85!white, draw=black!100!white]
				(-1,2)+(0.5,-0.5) circle (0.16)
				(-1,-2)+(0.5,-0.5) circle (0.16)
				(0,-1)+(0.5,-0.5) circle (0.16)
				
				(1,2)+(0.5,-0.5) circle (0.16)
				(2,1)+(0.5,-0.5) circle (0.16)
				(-1,0)+(0.5,-0.5) circle (0.16);

				\draw[SkyBlue, line width=2pt] (-1,2) rectangle +(2,-2);
				
			\end{tikzpicture}
			\vspace{0.5cm}
		\end{subfigure}
		\caption{\textit{Left}: interlaced particle description of an Aztec diamond of order 8. \textit{Right}: example of two different AD configurations which lead to the same particle system. The local pattern in blue appears for each non-adjacent particle.}
		\label{fig:IPS}
	\end{figure}

	If all dominos have the same weight, all configurations are equiprobable. Then, the correspondence between the probability measure on the Aztec diamond and the interlaced particle system can be made solely by knowing how many AD configurations correspond to the same interlaced particle configuration. To answer this question, we use the notion of \textit{adjacency} \cite{fleming2011interlaced}. A particle $x_i^{(\ell)}$ is called \textit{adjacent} if $x_i^{(\ell)}=x_j^{(\ell+1)}$ for $j=i$ or $i+1$, namely the particle $x_i^{(\ell)}$ has another particle located directly to its right. 
	We introduce the characteristic function
	\begin{equation} 
		\alpha(x_i^{(\ell)})
		= 
		\delta_{x_i^{(\ell)},x_i^{(\ell+1)}} + \delta_{x_i^{(\ell)},x_{i+1}^{(\ell+1)}},
		\label{eq:ajdacency}
	\end{equation}
	which equals $1$ if the particle at $x_i^{(\ell)}$ is adjacent and $0$ otherwise. It allows us to write the number of adjacent particles lying from level $\ell_1$ to $\ell_2$ as
	\begin{equation} 
		\alpha(x^{(\ell_1)},\dots,x^{(\ell_2)})
		= \sum_{\ell=\ell_1}^{\ell_2} \alpha(x^{(\ell)})
		= \sum_{\ell=\ell_1}^{\ell_2} \sum_{i=1}^\ell \alpha(x_i^{(\ell)}).
		\label{eq:adjacency def}
	\end{equation}
	Since $n(n+1)/2$ is the total number of particles from level $1$ to $n$, the number of non-adjacent particles on the entire AD is given by
	\begin{equation}
		\frac{n(n+1)}{2} - \alpha(x^{(1)},\dots,x^{(n)}).
	\end{equation}
	Because of the remark made below about non-adjacent particles, we fix for consistency that all particles lying on the $n$-th level are adjacent, $\alpha(x^{(n)})=n$, as if there was an imaginary $(n+1)$-th level completely fulfilled.

	For further purposes, we split our definition of $\alpha$ into two distinct subclasses
	\begin{equation}
		\alpha_{\text{inf}}(x_i^{(\ell)})=     \delta_{x_i^{(\ell)},x_{i}^{(\ell+1)}},
		\qquad
		\alpha_{\text{sup}}(x_i^{(\ell)}) = \delta_{x_i^{(\ell)},x_{i+1}^{(\ell+1)}}.
	\end{equation}
	These functions encode the information of whether an adjacent particle on the level $\ell$ lies next to its superior or inferior neighbouring particle on the level $\ell+1$. When $\alpha$ functions are written with several arguments, it means that we take the sum over contributions of each individual particle, similarly to \eqref{eq:adjacency def}. We note that, for convenience, only $x_i^{(\ell)}$ appears in the arguments of the adjacency functions, although these functions formally depend on the particle on the next level also.

	As one can see in Figure~\ref{fig:IPS}, whenever a particle is non-adjacent, its domino is contained in a plaquette made of two parallel dominos (framed in blue on the picture) which can be flipped without modifying the particle positions. As a consequence, we have two AD configurations for the same particle configuration. This observation is crucial to make the link between the probability distribution of the Aztec diamond and the density of the particle system. It is also reminiscent of the relation between Aztec diamond tilings and ASMs or, equivalently, the six-vertex model. In fact, it can be shown that the freedom in those plaquette orientations is exactly the same one that appears when we relate AD configurations to ASMs, such that our particle system is in one-to-one correspondence with alternating sign matrices.

	\subsection{Probability density function}
	\label{sec:PDF}
	
	The probability distribution of the particle system is induced by the measure on AD configurations through the general relation
	\begin{equation}
		\rho(x^{(1)},\dots,x^{(n+1)}) = \left(\sum_{c\in\mathcal{C}} \mathbb{P}(c)\right)  \chit(x^{(1)}\prec\dots\prec x^{(n+1)}),
		\label{eq:pdf from AD to particles}
	\end{equation}
	where $\mathcal{C}=\mathcal{C}(x^{(1)},\dots,x^{(n+1)})$ is the subset of all AD configurations which produce the same interlaced $x$-particle system. The notation $x^{(\ell)}\prec x^{(\ell+1)}$ means that the particles on levels $\ell$ and $\ell+1$ are interlaced, namely, we have for all $i$,
	\begin{equation}\begin{aligned}
			& x_i^{(\ell)} \in \{0,1,\dots,n\},  
			\\[0.05cm]
			& x_i^{(\ell)} < x_{i+1}^{(\ell)}, 
			\\[0.05cm]
			& x_i^{(\ell+1)} \le x_i^{(\ell)} \le x_{i+1}^{(\ell+1)},
			\label{eq:particle conditions}
	\end{aligned}\end{equation}
	whereas $\chit (x^{(\ell)}\prec x^{(\ell+1)})=1$ if $x^{(\ell)}\prec x^{(\ell+1)}$ and $0$ otherwise.
	
	Based on the remark made in the previous section, it is not difficult to see that, in the uniform case \cite{fleming2011interlaced},
	\begin{equation}
		\sum_{c\in\mathcal{C}} \mathbb{P}(c)
		=
		\frac{2^{n(n+1)/2 - \alpha(x^{(1)},\dots,x^{(n)})}}{Z_n}.
	\end{equation}
	Here, $Z_n=2^{n(n+1)/2}$ is the partition function in the uniform case (i.e.~the number of domino tilings of an Aztec diamond of order $n$). 
	
	\begin{figure}[!ht]
		\centering
		\begin{tikzpicture}[scale=0.65, rotate=45]
			\draw[line width=2pt] (0,8) -- ++(1,0) -- ++(0,-1) -- ++(1,0) -- ++(0,-1) -- ++(1,0) -- ++(0,-1) -- ++(1,0) -- ++(0,-1) -- ++(1,0) -- ++(0,-1) -- ++(1,0) -- ++(0,-1) -- ++(1,0) -- ++(0,-1) -- ++(1,0) -- ++(0,-2) -- ++(-1,0) -- ++(0,-1) -- ++(-1,0) -- ++(0,-1) -- ++(-1,0) -- ++(0,-1) -- ++(-1,0) -- ++(0,-1) -- ++(-1,0) -- ++(0,-1) -- ++(-1,0) -- ++(0,-1) -- ++(-1,0) -- ++(0,-1) -- ++(-1,0) ;
			\draw[line width=2pt] (0,8) -- ++(-1,0) -- ++(0,-1) -- ++(-1,0) -- ++(0,-1) -- ++(-1,0) -- ++(0,-1) -- ++(-1,0) -- ++(0,-1) -- ++(-1,0) -- ++(0,-1) -- ++(-1,0) -- ++(0,-1) -- ++(-1,0) -- ++(0,-1) -- ++(-1,0) -- ++(0,-2) -- ++(1,0) -- ++(0,-1) -- ++(1,0) -- ++(0,-1) -- ++(1,0) -- ++(0,-1) -- ++(1,0) -- ++(0,-1) -- ++(1,0) -- ++(0,-1) -- ++(1,0) -- ++(0,-1) -- ++(1,0) -- ++(0,-1) -- ++(1,0) ;
			\begin{scope}[black!64,line width=2pt]
			\end{scope}
			
			\filldraw[fill=yellow!25!white, draw=black]
			(-1,8) rectangle +(2,-1)
			(-2,7) rectangle +(2,-1)
			(0,7) rectangle +(2,-1)
			(-3,6) rectangle +(2,-1)
			(-1,6) rectangle +(2,-1)
			(1,6) rectangle +(2,-1)
			(-4,5) rectangle +(2,-1)
			(0,5) rectangle +(2,-1)
			(2,5) rectangle +(2,-1)
			(-4,3) rectangle +(2,-1)
			(-2,3) rectangle +(2,-1)
			(1,2) rectangle +(2,-1)
			(-2,1) rectangle +(2,-1)
			(2,1) rectangle +(2,-1)
			(-5,0) rectangle +(2,-1)
			(-3,0) rectangle +(2,-1)
			(1,-4) rectangle +(2,-1);
			
			\filldraw[fill=green!25!white, draw=black]
			(-2,5) rectangle +(1,-2)
			(-5,4) rectangle +(1,-2)
			(-6,3) rectangle +(1,-2)
			(0,3) rectangle +(1,-2)
			(-7,2) rectangle +(1,-2)
			(-5,2) rectangle +(1,-2)
			(-3,2) rectangle +(1,-2)
			(-8,1) rectangle +(1,-2)
			(-6,1) rectangle +(1,-2)
			(0,1) rectangle +(1,-2)
			(-7,0) rectangle +(1,-2)
			(-1,0) rectangle +(1,-2)
			(-6,-1) rectangle +(1,-2)
			(-2,-1) rectangle +(1,-2)
			(0,-1) rectangle +(1,-2)
			(2,-1) rectangle +(1,-2)
			(-5,-2) rectangle +(1,-2)
			(-1,-2) rectangle +(1,-2)
			(0,-3) rectangle +(1,-2);
			
			\filldraw[fill=yellow!70!black!80!white, draw=black]
			(-4,4) rectangle +(2,-1)
			(0,4) rectangle +(2,-1)
			(2,4) rectangle +(2,-1)
			(1,3) rectangle +(2,-1)
			(-2,2) rectangle +(2,-1)
			(2,0) rectangle +(2,-1)
			(-5,-1) rectangle +(2,-1)
			(-3,-3) rectangle +(2,-1)
			(1,-3) rectangle +(2,-1)
			(-4,-4) rectangle +(2,-1)
			(-2,-4) rectangle +(2,-1)
			(-3,-5) rectangle +(2,-1)
			(-1,-5) rectangle +(2,-1)
			(1,-5) rectangle +(2,-1)
			(-2,-6) rectangle +(2,-1)
			(0,-6) rectangle +(2,-1)
			(-1,-7) rectangle +(2,-1);
			
			\filldraw[fill=green!50!black!80!white, draw=black]
			(-1,5) rectangle +(1,-2)
			(4,4) rectangle +(1,-2)
			(3,3) rectangle +(1,-2)
			(5,3) rectangle +(1,-2)
			(-4,2) rectangle +(1,-2)
			(4,2) rectangle +(1,-2)
			(6,2) rectangle +(1,-2)
			(1,1) rectangle +(1,-2)
			(5,1) rectangle +(1,-2)
			(7,1) rectangle +(1,-2)
			(4,0) rectangle +(1,-2)
			(6,0) rectangle +(1,-2)
			(-3,-1) rectangle +(1,-2)
			(1,-1) rectangle +(1,-2)
			(3,-1) rectangle +(1,-2)
			(5,-1) rectangle +(1,-2)
			(-4,-2) rectangle +(1,-2)
			(4,-2) rectangle +(1,-2)
			(3,-3) rectangle +(1,-2);
			
			\begin{scope}[violet,dashed,line width=0.3pt]
				\foreach \x in {1,2,3,4,5,6,7,8,9}
				\draw (\x,9-\x) -- +(-9,-9) node[anchor=north] {$\ell$=\x};
			\end{scope}
			
			\begin{scope}[violet,dashed,line width=0.3pt]
				\foreach \x in {0,1,2,3,4,5,6,7,8}
				\draw (-8+\x,\x) node[anchor=east]{$k$=\x \;\;\;} -- +(8.5,-8.5)  ;
			\end{scope}

			\foreach \x in {0,1,2,3,4,5,6,7,8}
			\foreach \y in {0,1,2,3,4,5,6,7}
			\filldraw[fill=white, draw=black]
			(-8+\x+\y,\x-\y) +(0.5,-0.5) circle (0.16);
			
			\foreach \x in {0,1,2,3,4,5,6,7,8}
			\foreach \y in {8}
			\filldraw[fill=red!85!white, draw=black]
			(-8+\x+\y,\x-\y) +(0.5,-0.5) circle (0.16);
			
			\filldraw[fill=red!85!white, draw=black]
			(-4,4) +(0.5,-0.5) circle (0.16)
			(0,4) +(0.5,-0.5) circle (0.16)
			(2,4) +(0.5,-0.5) circle (0.16)
			(1,3) +(0.5,-0.5) circle (0.16)
			(-2,2) +(0.5,-0.5) circle (0.16)
			(2,0) +(0.5,-0.5) circle (0.16)
			(-5,-1) +(0.5,-0.5) circle (0.16)
			(-3,-3) +(0.5,-0.5) circle (0.16)
			(1,-3) +(0.5,-0.5) circle (0.16)
			(-4,-4) +(0.5,-0.5) circle (0.16)
			(-2,-4) +(0.5,-0.5) circle (0.16)
			(-3,-5) +(0.5,-0.5) circle (0.16)
			(-1,-5) +(0.5,-0.5) circle (0.16)
			(1,-5) +(0.5,-0.5) circle (0.16)
			(-2,-6) +(0.5,-0.5) circle (0.16)
			(0,-6) +(0.5,-0.5) circle (0.16)
			(-1,-7) +(0.5,-0.5) circle (0.16)
			
			(-1,5) +(0.5,-0.5) circle (0.16)
			(4,4) +(0.5,-0.5) circle (0.16)
			(3,3) +(0.5,-0.5) circle (0.16)
			(5,3) +(0.5,-0.5) circle (0.16)
			(-4,2) +(0.5,-0.5) circle (0.16)
			(4,2) +(0.5,-0.5) circle (0.16)
			(6,2) +(0.5,-0.5) circle (0.16)
			(1,1) +(0.5,-0.5) circle (0.16)
			(5,1) +(0.5,-0.5) circle (0.16)
			(7,1) +(0.5,-0.5) circle (0.16)
			(4,0) +(0.5,-0.5) circle (0.16)
			(6,0) +(0.5,-0.5) circle (0.16)
			(-3,-1) +(0.5,-0.5) circle (0.16)
			(1,-1) +(0.5,-0.5) circle (0.16)
			(3,-1) +(0.5,-0.5) circle (0.16)
			(5,-1) +(0.5,-0.5) circle (0.16)
			(-4,-2) +(0.5,-0.5) circle (0.16)
			(4,-2) +(0.5,-0.5) circle (0.16)
			(3,-3) +(0.5,-0.5) circle (0.16);
			
			\draw[SkyBlue, line width=3.5pt] 
			(-4,4) rectangle +(2,-2) 
			(-4,2) rectangle +(2,-2) 
			(-2,2) rectangle +(2,-2) 
			(1,3) rectangle +(2,-2) 
			(1,-1) rectangle +(2,-2) 
			(1,-3) rectangle +(2,-2) 
			(-3,-1) rectangle +(2,-2);
			
		\end{tikzpicture}
		\caption{Weighted Aztec diamond of order $n=8$ described in terms of interlaced particles. Vertical dominos are drawn in green and horizontal ones in yellow. We kept dark shades for previously black dominos, and light ones for dominos previously in white. The blue squares around non-adjacent particles indicate that we can flip a pair of yellow dominos into a pair of green ones (and conversely) without modifying the particle system. All the other dominos are entirely fixed by the particle positions. As well as the particles in red, their complement, the holes, are drawn in white. The additional level $\ell=n+1$ is also represented.}
		\label{fig:IPS biased}
	\end{figure}

	In the biased case, one needs to distinguish the contribution from vertical and horizontal dominos in terms of the position of each particle. In Figure~\ref{fig:IPS biased}, we coloured the vertical and horizontal dominos in green and yellow respectively, and in two different intensities, the darker ones are those containing a particle. A striking remark is that the orientation of a domino lying on an adjacent particle is fully determined by the neighbouring particle to its right\footnote{We give a detailed explanation for this in Section~\ref{sec:PDF 2P} in the two-periodic case.}. If it is adjacent to its superior neighbour (i.e.\ $x_i^{(\ell)}=x_{i+1}^{(\ell+1)}$) then the domino has to be vertical, if it is adjacent to its inferior neighbour (i.e.\ $x_i^{(\ell)}=x_{i}^{(\ell+1)}$) then the domino is horizontal. In addition, we know that the number of horizontal/vertical light shaded dominos is the same as their dark shaded version\footnote{It is a direct consequence of the fact that any configuration can be obtained from another by a finite sequence of flips of two parallel dominos.\label{page:argument of flip}}. The number of vertical (resp.\ horizontal) dominos completely fixed by the particle configuration is then given by $2\alpha_{\text{sup}}(x^{(1)},\dots,x^{(n)})$ (resp.\ $2\alpha_{\text{inf}}(x^{(1)},\dots,x^{(n)})$). The remaining dominos are the ones in the blue plaquettes. Each plaquette contributes a factor $1$ if the dominos are both horizontal, or $\lambda$ if they are both vertical; their sum is raised to a power equal to the number of non-adjacent particles. From this observation, we finally obtain the total probability density
	\begin{equation}
		\rho(x^{(1)},\dots,x^{(n+1)}) = \frac{(1 + \lambda)^{\frac{n(n+1)}{2}}}{Z_n} \cdot \Omega(x^{(1)},\dots,x^{(n)}) \cdot \chit(x^{(1)}\prec\dots\prec x^{(n+1)}),
		\label{eq:tot pdf}
	\end{equation}
	where
	\begin{equation}
		\Omega(x^{(1)},\dots,x^{(n)})
		=
		\frac{\lambda^{\alpha_{\text{sup}}(x^{(1)},\dots,x^{(n)})}} {(1 + \lambda)^{\alpha(x^{(1)},\dots,x^{(n)})}}.
	\end{equation}
	In the distribution, we choose to include the $(n+1)$-th level, completely filled with particles. The variables $x^{(n+1)}$ are deterministic and have no effect on the distribution, but they turn out to be useful for computations. By summing over the positions of the particles in the $(m-1)$ first levels, we obtain the following result.

	\begin{proposition}
		The joint probability density function (JPDF) over the subset of particles on levels $m$ to $n+1$ is given by
		\begin{equation}
			\rho(x^{(m)},\dots,x^{(n+1)}) = \frac{(1+ \lambda)^{\frac{n(n+1)}{2}}}{Z_n} \cdot \frac{\Delta(x^{(m)})}{\prod_{i=1}^{m-1} i!}\cdot \Omega(x^{(m)},\dots,x^{(n)}) \cdot \chit(x^{(m)}\prec\dots\prec x^{(n+1)}) ,
			\label{eq:partially int pdf}
		\end{equation}
		where
		\begin{equation}
			\Delta(x^{(m)}) =\det_{1\le i,j \le m}\left[(x_{i}^{(m)})^{j-1} \right]= \prod_{1\le i<j \le m} (x_j^{(m)} - x_i^{(m)})
			\label{eq:Vandermonde}
		\end{equation}
		is a Vandermonde determinant.
		\label{prop:JPDF m to n}
	\end{proposition}
	
	In particular, we can recover the partition function from this formula by integrating over all particle positions (i.e.\ considering the above formula for $m=n+1$). Because $x_i^{(n+1)}=i-1$, we have $\Delta(x_i^{(n+1)})= \prod_{i=1}^{n}i! $ and since the probability density is normalized, we find 
	\begin{equation}
		Z_n = (1 + \lambda)^{\frac{n(n+1)}{2}},
	\end{equation}
	as expected \cite{elkies1992alternating}. We now prove Proposition~\ref{prop:JPDF m to n}.

	\begin{proof}
		From \eqref{eq:tot pdf}, the formula \eqref{eq:partially int pdf} is correct for $m=1$. Assuming it is also the case for $m$, we show that it holds for $m+1$ by verifying the following identity,
		\begin{equation}
			\sum_{x^{(m)}}
			\Delta(x^{(m)})
			\Omega(x^{(m)})
			\chit(x^{(m)}\prec x^{(m+1)})
			=
			\frac{1}{m!}\Delta(x^{(m+1)}).
			\label{eq:rec int}
		\end{equation}
		To lighten the notations, let us set $x_i\equiv x_i^{(m)}\quad (1\le i\le m)$ and $y_i\equiv x_i^{(m+1)}\quad (1\le i\le m+1)$.
		Since the Vandermonde determinant in the left-hand side (lhs) depends on each variable separately row by row, we can factorize $\Omega(x^{(m)})=\prod_i \Omega(x_i^{(m)})$ and insert it inside the determinant as well as the summations,
		\begin{equation}
			\sum_{x_1=y_1}^{y_2} \cdots\sum_{x_m=y_{m}}^{y_{m+1}}
			\det_{1\le i,j \le m}\left[
			x_i^{j-1} \Omega(x_i)\right]
			=
			\det_{1\le i,j \le m}\left[\sum_{x_i=y_i}^{y_{i+1}}
			x_i^{j-1} \Omega(x_i)\right].
		\end{equation}
		The factor $\Omega(x_i)$ only affect the boundary terms of the sum and can be written as 
		\begin{equation}
			\Omega(x_i) 
			=
			\frac{\lambda^{\alpha_{\text{sup}}(x_i)}} {(1 +\lambda)^{\alpha(x_i)}}
			=
			\left\{\begin{array}{ll}
				\Omega_{\text{sup}}\equiv\frac{\lambda}{(1 +\lambda)} & \text{if } x_i=y_{i+1} ,
				\\[0.2cm]
				\Omega_{\text{inf}}\equiv\frac{1}{(1 +\lambda)} & \text{if } x_i=y_{i} ,
				\\[0.2cm]
				1 & \text{otherwise. }
			\end{array}\right.
		\end{equation}
		Such sums of powers can easily be computed using Euler-Maclaurin or Faulhaber's formula, for $j\ge 1$,
		\begin{equation}\begin{aligned}
				\sum_{x_i=y_i}^{y_{i+1}} x_i^{j-1}\Omega(x_i) 
				&= 
				\sum_{x_i=y_i+1}^{y_{i+1}} x_i^{j-1} + y_{i+1}^{j-1} (\Omega_{\text{sup}}-1) + y_i^{j-1} \,\Omega_{\text{inf}}
				\\ &=
				\sum_{x_i=0}^{y_{i+1}} x_i^{j-1} - \sum_{x_i=0}^{y_{i}} x_i^{j-1} - \frac{y_{i+1}^{j-1}}{(1 +\lambda)}  + \frac{ y_i^{j-1} }{(1 +\lambda)} 
				= p_j(y_{i+1}) - p_j(y_i),
		\end{aligned}\end{equation}
		where
		\begin{equation}
			p_j(y) = \frac{y^j}{j} + \frac{\lambda-1}{\lambda+1}\frac{y^{j-1}}{2}+ \sum_{k=2}^{j-1} \frac{B_k}{j}\binom{j}{k} y^{j-k},
		\end{equation}
		with $B_k$ the Bernoulli numbers. As all entries of the determinant are proportional to the difference of polynomials in each variable, we can extend this determinant of size $m$ into a determinant of size $(m+1)$ using the relation
		\begin{equation}
			\det_{1\le i,j \le m}\big[
			p_j(y_{i+1}) - p_j(y_i)\big] 
			= 
			\det_{1\le i,j \le m+1}\big[p_{j-1}(y_i)\big]
			=
			\det_{1\le i,j \le m+1}\left[\frac{y_i^{j-1}}{j-1}
			\right]
			=
			\frac{1}{m!}\Delta(x^{(m+1)}),
			\label{eq:extension of antisym matrix}
		\end{equation}
		where we set $p_0(y)=1$. The first equality is obtained by elementary row operations and actually holds for more general functions than polynomials. The second is obtained by column operations and is a standard property of Vandermonde determinants. The last relation finally yields the desired result and concludes the proof.
	\end{proof}

	\subsection{Distribution on first levels}
	\label{sec:Integration from the right}
	
	In this section, we obtain the distribution $\rho(x^{(1)},\dots,x^{(m)})$ by integrating from the right side of the particle system. 
	Another possible approach to obtain this result, similar to \cite{fleming2011interlaced}, consists in the introduction of the holes (the complement of the particle set, see Figure~\ref{fig:IPS biased}) and the use of the symmetries of the Aztec diamond. It is how we initially obtained the distribution, however, the method described here after is slightly more direct.
	
	\begin{proposition}
		The probability density of the particle process on the first $m$ levels is given by
		\begin{equation}
			\rho(x^{(1)},\dots,x^{(m)}) = A_{n,m} \det_{1\le i,j\le m} \left[L_{n-m+i,x_j^{(m)}} \right]
			\cdot
			\Omega(x^{(1)},\dots,x^{(m-1)})
			\cdot
			\chit(x^{(1)}\prec\dots\prec x^{(m)}),
			\label{eq:pdf 1 to m, x with determinant}
		\end{equation}
		where the $L_{i,j}= \lambda^j \binom{i}{j}$ can be defined by the initial condition $L_{0,j}=\delta_{0,j}$, and the recurrence, for $i\ge0$,
		\begin{equation}
			L_{i+1,j} = L_{i,j} + \lambda\, L_{i,j-1} , 
			\qquad\textit{while}\qquad 
			A_{n,m} = \frac{\lambda^{-m(m-1)/2}}{(1+\lambda)^{m(n+1-m)}}.
		\end{equation}
		\label{prop:int from the right}
	\end{proposition}

	\begin{proof}
		We proceed by induction starting from $m=n+1$ (or indeed $m=n$ as both values yield the same distribution as it should be). In this case, we recover directly \eqref{eq:tot pdf} due to $\det_{1\le i,j\le n+1} L_{i-1,j-1} = \lambda^{n(n+1)/2}$. For the induction step, we need to prove that the distribution integrated over its $(m+1)$-th level yields the correct density for the $m$ first levels. This is equivalent to show that
		\begin{equation}
			\sum_{x^{(m+1)}} \det_{1\le i,j\le m+1} \left[ L_{n-m+i-1,x_j^{(m+1)}} \right]
			\cdot
			\Omega(x^{(m)})
			\cdot
			\chit(x^{(m)}\prec x^{(m+1)})
			=
			\frac{A_{n,m}}{A_{n,m+1}}\det_{1\le i,j\le m} \left[ L_{n-m+i,x_j^{(m)}}\right].
			\label{eq:induction step, biased}
		\end{equation}
		As before, let us set $x_j\equiv x_j^{(m)}\quad (1\le j\le m)$ and $y_j\equiv x_j^{(m+1)}\quad (1\le j\le m+1)$.
		The adjacency factor $\Omega(x^{(m)})$ weighs the distribution when particles of levels $m$ and $m+1$ have the same position. Until now, we wrote it as a function of the level $m$, but here, it is convenient to interpret this factor as a function $\leftArr\Omega(x^{(m+1)})=\Omega(x^{(m)})$ of the level $m+1$ instead, such that 
		\begin{equation}
			\leftArr\Omega(y_j) 
			=
			\left\{\begin{array}{cl}
				\frac{1}{(1 +\lambda)} & \text{if } y_j=x_j ,
				\\[0.2cm]
				\frac{\lambda}{(1 +\lambda)} & \text{if } y_j=x_{j-1} ,
				\\[0.2cm]
				1 & \text{otherwise. }
			\end{array}\right.
		\end{equation}
		We rewrite the lhs of \eqref{eq:induction step, biased} with the sums in the determinant,
		\begin{equation}
			\det_{1\le i,j\le m+1}  \left[\sum_{y_j=x_{j-1}}^{x_j}  L_{n-m+i-1,y_j} \;
			\leftArr\Omega(y_j)\right],
		\end{equation}
		where we set $x_0=0$ and $x_{m+1}=n$. (Note that $\leftArr\Omega$ takes the value $1$ in these special cases since it corresponds to the boundary of the AD and not to the presence of a particle.) We simplify this determinant with the column operations: $C_j\mapsto \sum_{k=1}^j C_k$. Since we have $\leftArr\Omega(y_i=x_i)+\leftArr\Omega(y_{i+1}=x_i)=1$, the entries become
		\begin{equation}
			\sum_{y_j=0}^{x_j}  L_{n-m+i-1,y_j} \;\leftArr\Omega(y_j),
		\end{equation}
		with $\leftArr\Omega$ affecting only the upper bound of the summation. In the last column, the expression simplifies further and we find $\sum_{y=0}^{n}  L_{n-m+i-1,y} = (1+\lambda)^{n-m+i-1}$. Since we want to reduce the size of the determinant, we apply the row operations $R_i\mapsto R_i - (1+\lambda)R_{i-1}$ in order to obtain $0$'s in all the elements of the last column except for the top-right corner. The rest of the entries, for $i=2,\dots,m+1$ and $j=1,\dots,m$, are given by
		\begin{equation}
			\sum_{y_j=0}^{x_j} \big( L_{n-m+i-1,y_j} -(1+\lambda)L_{n-m+i-2,y_j} \big)\leftArr\Omega(y_j)
			=
			\sum_{y_j=0}^{x_j} \lambda \big( L_{n-m+i-2,y_j-1} - L_{n-m+i-2,y_j} \big)\leftArr\Omega(y_j).
		\end{equation}
		The last summation, being telescopic, can be carried explicitly and, after some simplifications, using the recurrence relation of the $L_{i,j}$'s, one finds that these entries are equal to $\frac{-\lambda}{1+\lambda} L_{n-m+i-1,x_j}$. In summary, the determinant is now given by
		\begin{equation}
			\begin{vNiceArray}{cw{c}{4cm}c|c}[margin]
				\Block{1-3}{\sum_{x=0}^{x_j}  L_{n-m,x} \;\leftArr\Omega(x)} & & & (1+\lambda)^{n-m} \\[0.2cm]
				\hline
				\Block{3-3}{\frac{-\lambda}{1+\lambda} L_{n-m+i-1,x_j}} & & & 0  \\
				& & & \vdots \\
				& & & 0 \\
			\end{vNiceArray}_{(m+1)\times (m+1)}
			\hspace{0cm}=
			\frac{\lambda^{m}}{(1+\lambda)^{2m-n}} \det_{1\le i,j\le m}  \Big[L_{n-m+i,x_j}\Big],
		\end{equation}
		and the prefactor exactly corresponds to $A_{n,m}/A_{n,m+1}$.
	\end{proof}
	
	The determinant in \eqref{eq:pdf 1 to m, x with determinant} can be explicitly computed, providing the following result.
	
	\begin{corollary}
		The probability density of the particle process on the first $m$ levels is given by
		\begin{equation}\begin{aligned}
				\rho(x^{(1)},\dots,x^{(m)}) 
				=
				A_{n,m}
				\Delta(x^{(m)})
				\prod_{i=1}^m\frac{(n+1-i)! \;\lambda^{x_i^{(m)}}}{x^{(m)}_{i}!(n-x^{(m)}_{i})!}
				\cdot
				\Omega(x^{(1)},\dots,x^{(m-1)})
				\cdot
				\chit(x^{(1)}\prec\dots\prec x^{(m)}).
			\end{aligned}
			\label{eq:pdf 1 to m, x}
		\end{equation}
	\end{corollary}
	
	To see that this distribution is indeed equal to the one in Proposition~\ref{prop:int from the right}, we only have to show that
	\begin{equation}
		\begin{vmatrix}
			\binom{n+1-m}{x^{(m)}_1}  & \cdots & \binom{n+1-m}{x^{(m)}_m}
			\\[0.1cm]
			\vdots  & \ddots & \vdots
			\\[0.1cm]
			\binom{n}{x^{(m)}_1}  & \cdots & \binom{n}{x^{(m)}_m}
		\end{vmatrix}
		=
		\prod_{i=1}^m\frac{(n+1-i)!}{x^{(m)}_{i}!(n-x^{(m)}_{i})!} \prod_{1\le i<j\le m} (x^{(m)}_{j}-x^{(m)}_{i}),
	\end{equation}
	since the factor $\lambda^{\sum_i x_i^{(m)}}$ accounts for the dependence on $\lambda$ of the $L_{i,j}$'s. To compute this determinant, one can use techniques from \cite{krattenthaler2001advanced}, for instance. However, this computation is not very instructive and will not be detailed here. This expression also makes the link with the result in the uniform case as obtained in \cite{fleming2011interlaced}.

	From this distribution, we can also obtain the JPDF of the particle positions on a single level by integrating over the first $m-1$ levels. Since all the dependence on these levels is contained solely inside the characteristic function $\chit$ and the adjacency functions $\Omega$, the computation is identical to the proof of Proposition~\ref{prop:JPDF m to n}. Hence, we obtain\footnote{Similarly, it is not difficult to obtain the distribution $\rho(x^{(m)},\dots,x^{(m+k)})$ for $k+1$ consecutive levels.}
	\begin{equation}
		\rho(x^{(m)}) =
		\frac{ \lambda^{\sum_i x_i^{(m)} - \frac{m(m-1)}{2}}} {(1+\lambda)^{m(n+1-m)}}
		\frac{\left[\Delta(x^{(m)})\right]^2}
		{\prod_{i=1}^{m} x_i^{(m)}!(n-x_i^{(m)})!} 
		\;
		\prod_{i=0}^{m-1} \frac{(n-i)!}{i!}.
		\label{eq:one level pdf}
	\end{equation}
	This last distribution is particularly interesting since it describes how the positions of particles evolve when we look deeper inside the Aztec diamond. As no approximation has been made so far, it contains all the contributions coming from the edges of the domain and the finite size effects. For instance, it was shown in \cite{fleming2011interlaced} that for $\lambda=1$, it is possible to recover the shape of the arctic curve from it. In this case, it is well known that this curve is a circle \cite{jockusch1998random}. For $\lambda\neq 1$, the curve is an ellipse such that horizontal or vertical dominos are favoured according to $\lambda<1$ or $\lambda>1$ respectively. We could probably recover the elliptic arctic curve using the expression above, however, it is an already well-documented subject, and other techniques such as the tangent method \cite{colomo2016arctic} are more appropriate for that.
	
	The distribution \eqref{eq:one level pdf} is also known to correspond to the probability measure of the Krawtchouk ensemble \cite{cohn2000local,johansson2002non},
	\begin{equation}
		m! \; \mathbb{P}^{Kr}_{m,n,p}(x^{(m)}) = 
		\frac{[\Delta(x^{(m)})]^2}{(n!)^m (pq)^{\frac{m(m-1)}{2}}} 
		\prod_{i=1}^m \binom{n}{x_i^{(m)}} p^{x_i^{(m)}} q^{n-x_i^{(m)}}
		\prod_{i=0}^{m-1} \frac{(n-i)!}{i!},
	\end{equation}
	for $p=\frac{\lambda}{1+\lambda}$ and $q=1-p$, and has been studied in detail by Johansson \cite{johansson2002non,johansson2001discrete,johansson2005arctic,johansson2006random}. However, his formalism and approach to the problem differ from ours. More precisely, he studied the statistics of a certain class of "zig-zag paths", which can be shown to correspond exactly to the particle system we are considering, the difference is that he does not compute the distribution $\rho(x^{(1)},\dots,x^{(m)})$ (which is our main interest) but directly $\rho(x^{(m)})$ via the Lindström-Gessel-Viennot method.

	\subsection{Convergence to the GUE-corners process}
	\label{sec:conv to GUE biased}
	
	In this section, our main purpose is to compute the scaling limit of the distribution for the $x_i^{(\ell)}$'s, and to show that they behave like the eigenvalues of GUE matrices. This correspondence is an impressive result in the study of Aztec diamonds and other tiling models \cite{adler2015tacnode,aggarwal2022gaussian}, obtained first at the level of correlation functions, by Johansson and Nordenstam \cite{johansson2006eigenvalues}, and later at the level of the probability distribution itself by Fleming and Forrester \cite{fleming2011interlaced}.

	The GUE-corners process (sometimes also called under its original name: GUE minor) is a well-known \textit{determinantal point process} in random matrix theory \cite{johansson2006eigenvalues,gorin2015asymptotics}. GUE refers to the Gaussian Unitary Ensemble which is constituted by random complex $m\times m$ hermitian matrices with entries following Gaussian laws in both their real and imaginary part, while the GUE-corners process more precisely describes the eigenvalues of matrices from GUE as well as all the eigenvalues $z_i^{(\ell)}$ of the size $\ell$ principal submatrices located in the upper-left corner ($1 \le i \le \ell$ and $1 \le \ell \le m$). As a property of hermitian matrices, the eigenvalues between two consecutive submatrices are interlaced. In the present work, we are mainly interested in the probability distribution of this process, 
	\begin{equation}
		\rho_{\text{GUE-corners}}(z^{(1)},\dots,z^{(m)})
		=
		\frac{1}{(2\pi)^{m/2}} \Delta(z^{(m)}) \prod_{i=1}^{m} {\rm e}^{-\frac{1}{2}(z_i^{(m)})^2}
		\; \chit(z^{(1)}\prec\dots\prec z^{(m)}).
	\end{equation}
	
	Over the last 20 years, the GUE-corners process has exhibited many universal features and has been observed to be the correct limiting process in many models showing an interlaced particle system. In some cases, the initial (discrete) particle process can be shown to be already \textit{determinantal} (as in the present case), which is very convenient to perform explicit calculations at finite size. In some other cases, however (as for the \textit{Alternating Sign Matrices} uniformly weighted), the initial process is not determinantal, although it still converges to GUE-corners in a suitable limit \cite{gorin2014alternating}.

	The main goal of this work is to emphasize this universality aspect by showing that this limiting process does not depend directly on the probability measure that we consider, but only the rescaling on the particle system does. Meaning that for some specific classes of probability measures on Aztec diamonds, in the scaling limit, GUE-corners always gives the correct description in a certain neighbourhood of the contact point between the edge and the arctic curves.

	\begin{proposition}
		If we define $z_i^{(\ell)} = \frac{x_i^{(\ell)} - \mu}{\sigma}$ to be the rescaling of the $x$-particle system, with $\mu=\frac{n\lambda}{1+\lambda}$ and $\sigma^2=\frac{n\lambda}{(1+\lambda)^2}$, then the JPDF $\rho(z^{(1)},\dots,z^{(m)})$ converges in distribution to $\rho_{\rm{GUE-corners}}(z^{(1)},\dots,z^{(m)})$ in the limit where $n\to\infty$ and $m$ remains finite.
	\end{proposition}
	
	What may seem surprising is the simplicity of the rescaling and the fact there is no dependence on the level of the particles. In Section~\ref{sec:conv to GUE 2P}, we give more details concerning this rescaling and get a better understanding on this apparent simplicity. 
	
	A remark about the limit that we consider: keeping $m$ finite while $n$ is growing is equivalent to zoom in a microscopic region around the contact point between the arctic curve and the edge of the domain. We emphasize the fact that this limit disregards the global structure of Aztec diamonds and tells nothing about the shape of the arctic curve for instance. Let us now present the proof of the result.

	\begin{proof}
		The Jacobian determinant of the transformation $x_i^{(\ell)}\mapsto z_i^{(\ell)}= \frac{x_i^{(\ell)}-\mu}{\sigma}$ is equal to $\sigma^{m(m+1)/2}$ and the Vandermonde determinant transforms like $\Delta(x^{(m)})=\sigma^{m(m-1)/2}\Delta(z^{(m)})$. Thus, we find the probability distribution of the $z$-particles to be
		\begin{equation}
			\rho(z^{(1)},\dots,z^{(m)}) 
			=
			\sigma^{m^2}
			\cdot
			\frac{\lambda^{\sigma\sum_i z_i^{(m)} + m \mu - \frac{m(m-1)}{2}} } {(1+\lambda)^{m(n-m+1)}} \Delta(z^{(m)})
			\prod_{i=1}^{m}\frac{ (n+1-i)!}
			{(z_i^{(m)}\sigma+\mu)!(n-z_i^{(m)}\sigma-\mu)!} \Omega \;\chit.
		\end{equation}
		Since we keep the number of particles finite in this approximation while the number of possible positions grows with $n$, it is expected that the probability to find adjacent particles goes to zero as $n$ increases. It implies that we can neglect the exponent $\alpha$ (and thus the effect of $\Omega)$ at the leading order in $n$. Finally, we use the Stirling approximation,
		\begin{equation}
			(an+b\sqrt{n}+c)! = \sqrt{2\pi}
			(an)^{an+b\sqrt{n}+c+\frac{1}{2}}
			\; e^{\frac{b^2}{2a}-an}
			\big(1+\mathcal{O}(n^{-1/2})\big),
		\end{equation}
		to extract the dominant order of the distribution, and after some computations, we find
		\begin{equation}
			\rho(z^{(1)},\dots,z^{(m)}) =
			\frac{1}{(2\pi)^{m/2}} \Delta(z^{(m)}) \prod_{i=1}^{m} {\rm e}^{-\frac{1}{2}(z_i^{(m)})^2}
			\; \chit(z^{(1)}\prec\dots\prec z^{(m)}) 
			\; \big(1+\mathcal{O}(n^{-1/2})\big),
		\end{equation}
		which ensures the convergence in distribution to the GUE-corners process in the limit $n\to\infty$.
	\end{proof}

	\section{Two-periodic Aztec diamonds}
	\label{sec:Two-periodic Aztec diamonds}

	In this section, we repeat the same procedure to obtain the distribution of particles for the two-periodic measure on Aztec diamonds. This measure can be defined on the dual graph as follows: we attribute weights $a$ and $b$ on the faces of the graph as represented in Figure~\ref{fig:two-periodic config}. The weight is then assigned to dimers according to whether they are adjacent to a face $a$ or $b$. Following this rule, the weight of a configuration $c$ is given by $w(c)=\prod_{{\rm faces}\; x_{i,j}} x_{i,j}^{N_{i,j}}$, where $N_{i,j}=0,1,2$ is the number of dimers adjacent to the face $x_{i,j}$ in the configuration. We take the convention that the face in the bottom-left corner of the graph is always of type $a$.

	\begin{figure}[!ht]
		\centering
		\begin{tikzpicture}[scale=0.8,rotate=45]
			\def\n{8} 
			
			\foreach \x in {1,...,\n}
			\draw[gray!70!white, very thin] (\x-1,\n-\x) rectangle +(2*\n-2*\x+1,2*\x -1)
			;
			
			\filldraw[SkyBlue!70!white, opacity=0.25]
			(4,11) rectangle +(1,-1)
			(4,9) rectangle +(1,-1)
			(6,9) rectangle +(1,-1)
			(5,6) rectangle +(1,-1)
			(9,10) rectangle +(1,-1)
			(9,6) rectangle +(1,-1)
			(9,4) rectangle +(1,-1);

			\begin{scope}[violet,dashed,line width=0.1pt, opacity=0.8]
				\foreach \x in {1,...,\n}
				\draw (\x-1.2,\n-\x-0.2) node[anchor=north] {$\ell$=\x} 
				;
			\end{scope}

			\begin{scope}[violet,dashed,line width=0.1pt, opacity=0.8]
				\foreach \x in {0,...,\n}
				\draw (\x-0.5,\n+\x-0.5) node[anchor=east]{$k$=\x \;\;\;} 
				;
			\end{scope}

			\foreach \x in {2,4,...,\n}{
				\foreach \y in {2,4,...,\n}{
					\filldraw[draw=white,fill=white]
					(-0.5+\n-\x+\y,0.5+\x+\y-4) circle(7pt)
					(-0.5+\n-\x+\y,2.5+\x+\y-4) circle(7pt);
					\draw[draw=white,fill=white]
					(-1.5+\n-\x+\y,1.5+\x+\y-4) circle(7pt)
					( 0.5+\n-\x+\y,1.5+\x+\y-4) circle(7pt);
			}}
			
			\foreach \x in {2,4,...,\n}{
				\foreach \y in {2,4,...,\n}{
					\draw[blue!90!black]
					(-0.5+\n-\x+\y,0.5+\x+\y-4) node {$b$}
					(-0.5+\n-\x+\y,2.5+\x+\y-4) node {$b$};
					\draw[red!90!black]
					(-1.5+\n-\x+\y,1.5+\x+\y-4) node {$a$}
					( 0.5+\n-\x+\y,1.5+\x+\y-4) node {$a$};
			}}
			
			\filldraw[red!90!black,line width=5pt]
			(4,11)--+(1,0)
			(7,12)--+(0,-1)
			(6,9)--+(1,0)
			(8,11)--+(1,0)
			(5,6)--+(0,-1)
			(10,11)--+(1,0)
			(4,3)--+(1,0)
			(9,8)--+(0,-1)
			(11,10)--+(0,-1)
			(6,3)--+(1,0)
			(9,6)--+(0,-1)
			(10,7)--+(1,0)
			(13,10)--+(0,-1)
			(6,1)--+(1,0)
			(11,6)--+(0,-1)
			(13,8)--+(0,-1)
			(8,1)--+(1,0)
			(11,4)--+(0,-1)
			(13,6)--+(0,-1)
			(15,8)--+(0,-1)
			;
			
			\filldraw[blue!90!black,line width=5pt]
			(4,9)--+(0,-1)
			(3,6)--+(1,0)
			(9,10)--+(1,0)
			(4,5)--+(0,-1)
			(5,4)--+(1,0)
			(12,11)--+(0,-1)
			(5,2)--+(1,0)
			(12,9)--+(0,-1)
			(7,2)--+(1,0)
			(9,4)--+(1,0)
			(12,7)--+(0,-1)
			(14,9)--+(0,-1)
			(7,0)--+(1,0)
			(9,2)--+(1,0)
			(12,5)--+(0,-1)
			(14,7)--+(0,-1)
			;
			
			\filldraw[red!25!white,line width=5pt]
			(0,7)--+(0,1)
			(2,9)--+(0,1)
			(5,12)--+(-1,0)
			(7,14)--+(-1,0)
			(2,7)--+(0,1)
			(6,11)--+(0,1)
			(5,10)--+(-1,0)
			(9,14)--+(-1,0)
			(2,5)--+(0,1)
			(7,10)--+(-1,0)
			(9,12)--+(-1,0)
			(7,8)--+(-1,0)
			(8,9)--+(0,1)
			(11,12)--+(-1,0)
			(6,5)--+(0,1)
			(8,7)--+(0,1)
			(8,5)--+(0,1)
			(11,8)--+(-1,0)
			(8,3)--+(0,1)
			(10,5)--+(0,1)
			;
			
			\filldraw[blue!25!white,line width=5pt]
			(1,8)--+(0,1)
			(3,10)--+(0,1)
			(6,13)--+(-1,0)
			(8,15)--+(-1,0)
			(1,6)--+(0,1)
			(3,8)--+(0,1)
			(8,13)--+(-1,0)
			(5,8)--+(0,1)
			(4,7)--+(-1,0)
			(10,13)--+(-1,0)
			(3,4)--+(0,1)
			(6,7)--+(-1,0)
			(7,6)--+(0,1)
			(10,9)--+(-1,0)
			(7,4)--+(0,1)
			(10,3)--+(-1,0)
			;

			\filldraw[fill=yellow!70!white, draw=black]
			(7,12) circle(3.5pt)
			(3,6) circle(3.5pt)
			(9,10) circle(3.5pt)
			(5,6) circle(3.5pt)
			(5,4) circle(3.5pt)
			(9,8) circle(3.5pt)
			(11,10) circle(3.5pt)
			(5,2) circle(3.5pt)
			(9,6) circle(3.5pt)
			(13,10) circle(3.5pt)
			(7,2) circle(3.5pt)
			(9,4) circle(3.5pt)
			(11,6) circle(3.5pt)
			(13,8) circle(3.5pt)
			(7,0) circle(3.5pt)
			(9,2) circle(3.5pt)
			(11,4) circle(3.5pt)
			(13,6) circle(3.5pt)
			(15,8) circle(3.5pt)
			;
			
			\filldraw[fill=green!70!white, draw=black]
			(4,11) circle(3.5pt)
			(4,9) circle(3.5pt)
			(6,9) circle(3.5pt)
			(8,11) circle(3.5pt)
			(10,11) circle(3.5pt)
			(4,5) circle(3.5pt)
			(4,3) circle(3.5pt)
			(12,11) circle(3.5pt)
			(6,3) circle(3.5pt)
			(10,7) circle(3.5pt)
			(12,9) circle(3.5pt)
			(6,1) circle(3.5pt)
			(12,7) circle(3.5pt)
			(14,9) circle(3.5pt)
			(8,1) circle(3.5pt)
			(12,5) circle(3.5pt)
			(14,7) circle(3.5pt) 
			;
			
		\end{tikzpicture}
		\caption{Same configuration as displayed in Section~\ref{sec:PDF} shown on the dual lattice with the two-periodic weighting. We kept the light and dark shade with the same convention as previously and dimers are coloured according to the weight, $a$ or $b$, they carry. Sky blue squares stand for the non-adjacent particles and the corresponding pair of dimers that can be flipped without modifying the particle configuration. Even particles are coloured in yellow and odd ones in green.}
		\label{fig:two-periodic config}
	\end{figure}

	This probability measure on Aztec diamonds has been thoroughly studied in the literature recently \cite{di2014octahedron,chhita2016domino,duits2020two,beffara2022local,ruelle2022double,borodin2023biased}, partly due to the interesting phenomenology it brings to the model. Indeed, the disordered region in this case is separated into two distinct zones with different behaviours. The external one is the usual rough/liquid region presenting the same characteristics as in the uniform case. While, in the very center, the disordered tiles are more likely to group themselves by pairs of parallel dominos around faces $a$ or $b$ depending on the favoured weight, see Figure~\ref{fig:two-periodic ADs and arctic phenomenon}. This region is usually called smooth or gaseous.

	\begin{figure}[ht!]
		\centering
		\includegraphics[width=0.3\linewidth]{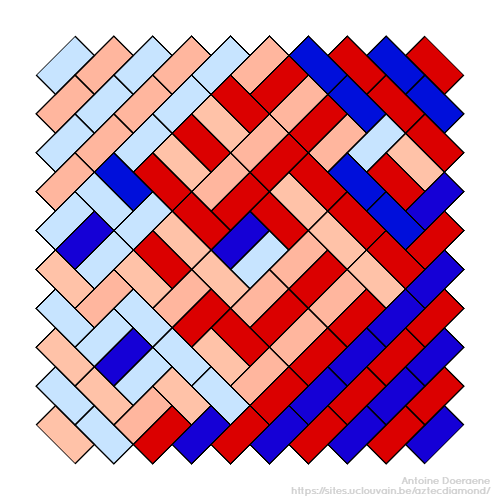}
		\hspace{0.5cm}
		\includegraphics[width=0.3\linewidth]{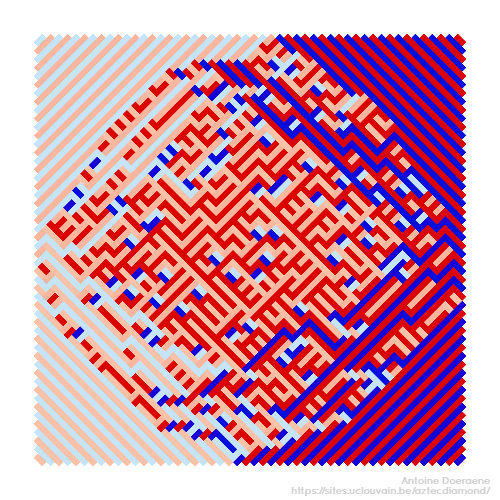}
		\hspace{0.5cm}
		\includegraphics[width=0.3\linewidth]{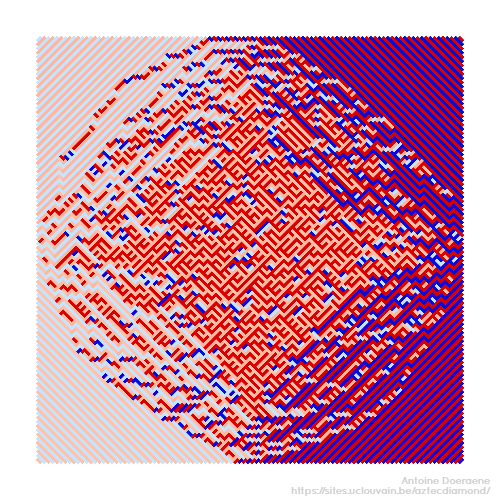}
		
		\includegraphics[width=0.3\linewidth]{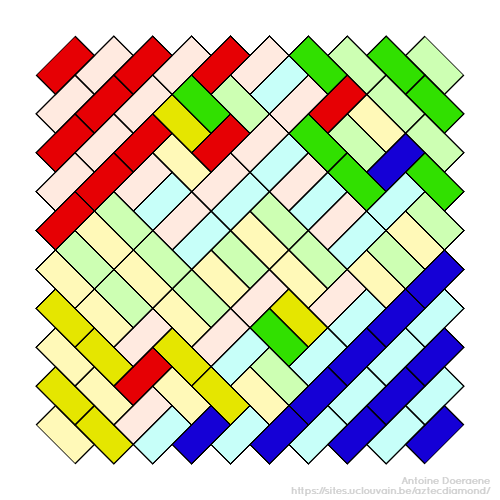}
		\hspace{0.5cm}
		\includegraphics[width=0.3\linewidth]{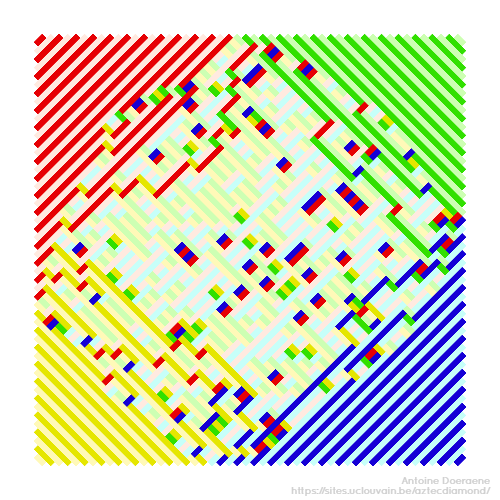}
		\hspace{0.5cm}
		\includegraphics[width=0.3\linewidth]{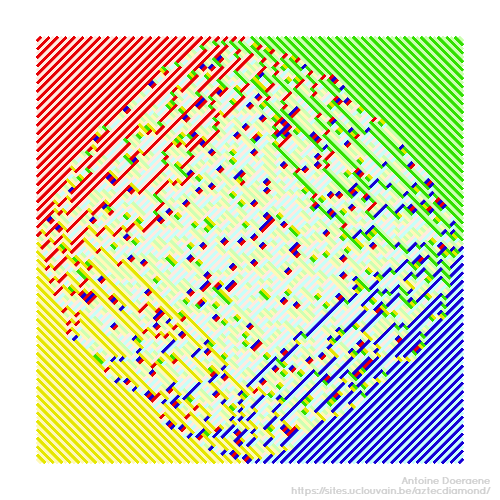}
		\caption{Two-periodic AD configurations with $a/b=2$ and $n=10,\,50,\,100$ from left to right. The first row uses the same colours as the previous figure while in the second eight different colours are used to highlight the gaseous region in the center. These eight colours are chosen as follows: first, we distinguish the four usual colours as in Figure~\ref{fig:AD definition}. After that, we lighten the colours carried by dominos with weight $a$, and darken the ones carried by dominos with weight $b$.}
		\label{fig:two-periodic ADs and arctic phenomenon}
	\end{figure}

	We define the particle system in the same way as before. The main difference with the previous case is that we have now to account for the parity of the particles. For this purpose, we define the parity of a position $(\ell,x_i^{(\ell)})$ on the grid by 
	\begin{equation}
		\varepsilon_\ell(x_i^{(\ell)}) 
		= (-1)^{\ell+x_i^{(\ell)}},
		\label{eq:parity}
	\end{equation}
	and assign a particle the parity of its position. It follows that the particles lying at the top of an $a$-face are even and those at the top of a $b$-face are odd.
	
	According to this definition, we further split the adjacency functions into\footnote{In practice, these adjacency functions need both particle coordinates $(\ell,x_i^{(\ell)})$ as arguments to be well-defined, but we decided to omit the level $\ell$ since it will be specified as upperscript of $x_i^{(\ell)}$ most of the time.}:
	\begin{align}
		\alpha_{\text{inf}}^\text{e}(x_i^{(\ell)}) &= 
		\delta_{x_i^{(\ell)},x_i^{(\ell+1)}} \delta_{\ell+x_i^{(\ell)},\,\text{even}} ,
		&
		\alpha_{\text{inf}}^\text{o}(x_i^{(\ell)}) &= 
		\delta_{x_i^{(\ell)},x_i^{(\ell+1)}} \delta_{\ell+x_i^{(\ell)},\,\text{odd}} ,
		\nonumber\\[0.1cm]
		\alpha_{\text{sup}}^\text{e}(x_i^{(\ell)}) &= 
		\delta_{x_i^{(\ell)},x_{i+1}^{(\ell+1)}} \delta_{\ell+x_i^{(\ell)},\,\text{even}} ,
		&
		\alpha_{\text{sup}}^\text{o}(x_i^{(\ell)}) &= 
		\delta_{x_i^{(\ell)},x_{i+1}^{(\ell+1)}} \delta_{\ell+x_i^{(\ell)},\,\text{odd}}.
		\label{eq:def of adjacency 2P}
	\end{align}
	We also keep the convention that for $\alpha$ functions of several arguments, we take the sum over the contributions of each particle similarly to \eqref{eq:adjacency def}.

	\subsection{Probability density function}
	\label{sec:PDF 2P}
	
	Referring to the equation \eqref{eq:pdf from AD to particles}, we need to determine the number of $a$- and $b$-dimers in terms of the particle positions. As promised in Section~\ref{sec:PDF}, we provide an explanation for the interlacement as well as the equivalence in the number of dark shaded dimers of each weight, and the adjacency classes of their corresponding particles. In Figure~\ref{fig:Construction by level}, we construct iteratively a configuration of order $n=4$ from level $1$ to $n$. In the first of these pictures, we see there are $n$ vertices to cover on the left edge of the graph, and $n+1$ in the next column, corresponding to the level $\ell=1$. Hence, one vertex of this column will not be covered by a dimer touching the left edge, and this is where the first particle will be located. For the three as yet uncovered vertices lying in between the levels $\ell=1$ and $\ell=2$, we see two groups of vertices separated by the red dimer: a group of one vertex above the red dimer, a group of two vertices below the red dimer. Therefore, the dimers covering these three vertices will leave two uncovered vertices, where the two particles at $\ell=2$ will be located. Continuing this way explains the interlacing property. Moreover, by looking at the first two pictures of Figure~\ref{fig:Construction by level}, one sees that if two particles are sitting next to each other in two neighbouring levels, then we have $x_i^{(\ell)}=x_{i}^{(\ell+1)}$ if the dimer covering the vertex at $x_i^{(\ell)}$ is horizontal, and otherwise $x_i^{(\ell)}=x_{i+1}^{(\ell+1)}$ if that dimer is vertical. These observations repeat in the other pictures. The conclusion is that the weight of a dimer covering an adjacent particle located at $x_i^{(\ell)}$ gets a factor $a$ if $\alpha_{\text{inf}}^\text{o}(x_i^{(\ell)})=1$ or $\alpha_{\text{sup}}^\text{e}(x_i^{(\ell)})=1$, and a factor $b$ if $\alpha_{\text{inf}}^\text{e}(x_i^{(\ell)})=1$ or $\alpha_{\text{sup}}^\text{o}(x_i^{(\ell)})=1$.

	\begin{figure}[ht]
		\centering
		\begin{tikzpicture}[scale=0.6,rotate=45]
			\def\n{4} 
			\foreach \x in {1,...,\n}
			\draw[gray!70!white, very thin] (\x-1,\n-\x) rectangle +(2*\n-2*\x+1,2*\x -1)
			;
			
			\foreach \x in {2,4,...,\n}{
				\foreach \y in {2,4,...,\n}{
					\draw[blue!90!black]
					(-0.5+\n-\x+\y,0.5+\x+\y-4) node {\scalebox{.7}{$b$}}
					(-0.5+\n-\x+\y,2.5+\x+\y-4) node {\scalebox{.7}{$b$}};
					\draw[red!90!black]
					(-1.5+\n-\x+\y,1.5+\x+\y-4) node {\scalebox{.7}{$a$}}
					( 0.5+\n-\x+\y,1.5+\x+\y-4) node {\scalebox{.7}{$a$}};
			}}
			\draw[red!90!black,line width=2.4pt]
			(2,5)--+(1,0)
			;
			
			\draw[blue!90!black,line width=2.4pt]
			;
			
			\draw[red!25!white,line width=2.4pt]
			(0,4)--+(0,-1)
			(2,6)--+(1,0)
			;
			
			\draw[blue!25!white,line width=2.4pt]
			(1,5)--+(0,-1)
			(3,7)--+(1,0)
			;
			
			\foreach \x in {1,...,2}
			\foreach \y in {0,...,4}
			\filldraw[fill=black,draw=black]
			(\x-1+\y,\n-\x+\y) circle (1.7pt);
			\foreach \x in {1,...,2}
			\foreach \y in {0,...,3}
			\filldraw[fill=black,draw=black]
			(\x-1+\y,\n+1-\x+\y) circle (1.7pt);
			
			\filldraw[fill=yellow!70!white, draw=black]
			;
			
			\filldraw[fill=green!70!white, draw=black]
			(2,5) circle(3.5pt)
			;
			
		\end{tikzpicture}
		\hspace{0.5cm}
		\begin{tikzpicture}[scale=0.6,rotate=45]
			\def\n{4} 
			\foreach \x in {1,...,\n}
			\draw[gray!70!white, very thin] (\x-1,\n-\x) rectangle +(2*\n-2*\x+1,2*\x -1)
			;
			
			\foreach \x in {2,4,...,\n}{
				\foreach \y in {2,4,...,\n}{
					\draw[blue!90!black]
					(-0.5+\n-\x+\y,0.5+\x+\y-4) node {\scalebox{.7}{$b$}}
					(-0.5+\n-\x+\y,2.5+\x+\y-4) node {\scalebox{.7}{$b$}};
					\draw[red!90!black]
					(-1.5+\n-\x+\y,1.5+\x+\y-4) node {\scalebox{.7}{$a$}}
					( 0.5+\n-\x+\y,1.5+\x+\y-4) node {\scalebox{.7}{$a$}};
			}}
			\draw[red!90!black,line width=2.4pt]
			(2,5)--+(1,0)
			(4,5)--+(1,0)
			(3,4)--+(0,-1)
			;
			
			\draw[blue!90!black,line width=2.4pt]
			;
			
			\draw[red!25!white,line width=2.4pt]
			(0,4)--+(0,-1)
			(2,6)--+(1,0)
			(2,4)--+(0,-1)
			(4,6)--+(1,0)
			;
			
			\draw[blue!25!white,line width=2.4pt]
			(1,5)--+(0,-1)
			(3,7)--+(1,0)
			(1,3)--+(0,-1)
			;
			
			\foreach \x in {1,...,3}
			\foreach \y in {0,...,4}
			\filldraw[fill=black,draw=black]
			(\x-1+\y,\n-\x+\y) circle (1.7pt);
			\foreach \x in {1,...,3}
			\foreach \y in {0,...,3}
			\filldraw[fill=black,draw=black]
			(\x-1+\y,\n+1-\x+\y) circle (1.7pt);
			
			\filldraw[fill=yellow!70!white, draw=black]
			(3,4) circle(3.5pt)
			;
			
			\filldraw[fill=green!70!white, draw=black]
			(2,5) circle(3.5pt)
			(4,5) circle(3.5pt)
			;
			
		\end{tikzpicture}
		\hspace{0.5cm}
		\begin{tikzpicture}[scale=0.6,rotate=45]
			\def\n{4} 
			\foreach \x in {1,...,\n}
			\draw[gray!70!white, very thin] (\x-1,\n-\x) rectangle +(2*\n-2*\x+1,2*\x -1)
			;
			
			\foreach \x in {2,4,...,\n}{
				\foreach \y in {2,4,...,\n}{
					\draw[blue!90!black]
					(-0.5+\n-\x+\y,0.5+\x+\y-4) node {\scalebox{.7}{$b$}}
					(-0.5+\n-\x+\y,2.5+\x+\y-4) node {\scalebox{.7}{$b$}};
					\draw[red!90!black]
					(-1.5+\n-\x+\y,1.5+\x+\y-4) node {\scalebox{.7}{$a$}}
					( 0.5+\n-\x+\y,1.5+\x+\y-4) node {\scalebox{.7}{$a$}};
			}}
			\draw[red!90!black,line width=2.4pt]
			(2,5)--+(1,0)
			(4,5)--+(1,0)
			(3,4)--+(0,-1)
			(2,1)--+(1,0)
			(4,3)--+(1,0)
			;
			
			\draw[blue!90!black,line width=2.4pt]
			(6,5)--+(0,-1)
			;
			
			\draw[red!25!white,line width=2.4pt]
			(0,4)--+(0,-1)
			(2,6)--+(1,0)
			(2,4)--+(0,-1)
			(4,6)--+(1,0)
			(2,2)--+(1,0)
			(4,4)--+(1,0)
			;
			
			\draw[blue!25!white,line width=2.4pt]
			(1,5)--+(0,-1)
			(3,7)--+(1,0)
			(1,3)--+(0,-1)
			;
			
			\foreach \x in {1,...,4}
			\foreach \y in {0,...,4}
			\filldraw[fill=black,draw=black]
			(\x-1+\y,\n-\x+\y) circle (1.7pt);
			\foreach \x in {1,...,4}
			\foreach \y in {0,...,3}
			\filldraw[fill=black,draw=black]
			(\x-1+\y,\n+1-\x+\y) circle (1.7pt);
			
			\filldraw[fill=yellow!70!white, draw=black]
			(3,4) circle(3.5pt)
			;
			
			\filldraw[fill=green!70!white, draw=black]
			(2,5) circle(3.5pt)
			(4,5) circle(3.5pt)
			(6,5) circle(3.5pt)
			(2,1) circle(3.5pt)
			(4,3) circle(3.5pt)
			;
			
		\end{tikzpicture}
		\hspace{0.5cm}
		\begin{tikzpicture}[scale=0.6,rotate=45]
			\def\n{4} 
			\foreach \x in {1,...,\n}
			\draw[gray!70!white, very thin] (\x-1,\n-\x) rectangle +(2*\n-2*\x+1,2*\x -1)
			;
			
			\foreach \x in {2,4,...,\n}{
				\foreach \y in {2,4,...,\n}{
					\draw[blue!90!black]
					(-0.5+\n-\x+\y,0.5+\x+\y-4) node {\scalebox{.7}{$b$}}
					(-0.5+\n-\x+\y,2.5+\x+\y-4) node {\scalebox{.7}{$b$}};
					\draw[red!90!black]
					(-1.5+\n-\x+\y,1.5+\x+\y-4) node {\scalebox{.7}{$a$}}
					( 0.5+\n-\x+\y,1.5+\x+\y-4) node {\scalebox{.7}{$a$}};
			}}
			\draw[red!90!black,line width=2.4pt]
			(2,5)--+(1,0)
			(4,5)--+(1,0)
			(3,4)--+(0,-1)
			(2,1)--+(1,0)
			(4,3)--+(1,0)
			(5,2)--+(0,-1)
			(7,4)--+(0,-1)
			;
			
			\draw[blue!90!black,line width=2.4pt]
			(6,5)--+(0,-1)
			(3,0)--+(1,0)
			(6,3)--+(0,-1)
			;
			
			\draw[red!25!white,line width=2.4pt]
			(0,4)--+(0,-1)
			(2,6)--+(1,0)
			(2,4)--+(0,-1)
			(4,6)--+(1,0)
			(2,2)--+(1,0)
			(4,4)--+(1,0)
			(4,2)--+(0,-1)
			;
			
			\draw[blue!25!white,line width=2.4pt]
			(1,5)--+(0,-1)
			(3,7)--+(1,0)
			(1,3)--+(0,-1)
			;
			
			\foreach \x in {1,...,\n}
			\foreach \y in {1,...,\x}
			\filldraw[fill=black,draw=black]
			(\x-1,\n-\y) circle (1.7pt)
			(\x-1,\n-1+\y) circle (1.7pt)
			(2*\n-\x,\n-\y) circle (1.7pt)
			(2*\n-\x,\n-1+\y) circle (1.7pt);
			
			\filldraw[fill=yellow!70!white, draw=black]
			(3,4) circle(3.5pt)
			(5,2) circle(3.5pt)
			(7,4) circle(3.5pt)
			;
			
			\filldraw[fill=green!70!white, draw=black]
			(2,5) circle(3.5pt)
			(4,5) circle(3.5pt)
			(6,5) circle(3.5pt)
			(2,1) circle(3.5pt)
			(4,3) circle(3.5pt)
			(3,0) circle(3.5pt)
			(6,3) circle(3.5pt)
			;
			
		\end{tikzpicture}
		\caption{Construction level by level of a configuration of order $4$.}
		\label{fig:Construction by level}
	\end{figure}

	In addition to that, the light and dark red dimers come in equal number and with the same weight by the same argument stated on page~\pageref{page:argument of flip} (the same goes for the blue dimers), so that the weights of the light dimers can be included into those of the dark ones, by replacing the factors $a$ and $b$ by $a^2$ and $b^2$. Finally, each non-adjacent particle lies around a flippable plaquette and contributes a weight $a^2+b^2$. Altogether, this leads to the probability distribution for the total system
	\begin{equation}
		\rho(x^{(1)},\dots,x^{(n+1)}) = 
		\frac{(a^2 + b^2)^{\frac{n(n+1)}{2}}}{Z_n} 
		\cdot 
		\Omega(x^{(1)},\dots,x^{(n)}) \cdot \chit(x^{(1)}\prec\dots\prec x^{(n+1)}),
		\label{eq:tot pdf, 2P}
	\end{equation}
	where we define
	\begin{equation}
		\Omega(x^{(1)},\dots,x^{(n)})
		=
		\frac{(a^2)^{\alpha_{\text{inf}}^\text{o}(x^{(1)},\dots,x^{(n)})+\alpha_{\text{sup}}^\text{e}(x^{(1)},\dots,x^{(n)})} (b^2)^{\alpha_{\text{inf}}^\text{e}(x^{(1)},\dots,x^{(n)})+\alpha_{\text{sup}}^\text{o}(x^{(1)},\dots,x^{(n)})}} {(a^2 +b^2)^{\alpha(x^{(1)},\dots,x^{(n)})}}.
		\label{eq:Adjacency factor}
	\end{equation}
	This result is valid for $n$ even or odd despite the fact that two-periodic ADs for $n$ odd break the symmetry between the weight $a$ and $b$. The total density found here looks very similar to the one for the biased case; only the factor $\Omega$ is different, reflecting the fact that particles are weighted differently when they are adjacent. In some sense, it is as if all the complexity of the probability measure is reduced into this notion of adjacency. We will see in the following that although the change in the distribution seems innocuous at this stage, it has a huge impact on the integration of the distribution and further computations.

	\begin{proposition}
		\label{prop:JPDF m to n, 2P}
		The joint probability distribution over the subset of particles on levels $m$ to $n+1$ is given by
		\begin{equation}
			\rho(x^{(m)},\dots,x^{(n+1)}) 
			= 
			\frac{(a^2 +b^2)^{\frac{n(n+1)}{2}}}
			{Z_n} \cdot \tilde\Delta(x^{(m)})\cdot \Omega(x^{(m)},\dots,x^{(n)}) \cdot \chit(x^{(m)}\prec\dots\prec x^{(n+1)}) 
			\label{eq:partially int pdf 2P}
		\end{equation}
		with
		\begin{equation}
			\tilde\Delta(x^{(m)}) =\det_{1\le i,j \le m}\left[\frac{(x_{i}^{(m)})^{j-1}}{(j-1)!} + \varepsilon_m(x_{i}^{(m)}) \sum_{k=0}^{j-2} \frac{(x_{i}^{(m)})^{k}}{k!} \alpha_{j-2-k} \right] .
			\label{eq:Vandermonde 2P}
		\end{equation}
		The coefficients $\alpha_{k}$ for $k\in\mathbb{N}$ are given by
		\begin{gather}
			\alpha_0 = v,
			\qquad
			\alpha_k = \sum_{\ell=2}^{k} \left(v \,\alpha_{\ell-2}\, \alpha_{k-\ell} - \frac{2^\ell B_\ell}{\ell!}  \, \alpha_{k-\ell}\right),
			\\
			f(z) = \sum_{k=0}^\infty \alpha_k z^k =
			\frac{1}{2v z \tanh(z)}\left(1-\sqrt{1-4 v^2 \tanh^2(z)} \right),
			\label{eq:alpha cst}
		\end{gather}
		with $B_\ell$ the Bernoulli numbers and $v=\frac{1}{2}\frac{b^2-a^2}{a^2+b^2} \in[-\frac{1}{2},\frac{1}{2}]$.
	\end{proposition}

	One can remark that only odd powers of $v$ appear in the determinant. As a consequence, the inversion of the weight $v\mapsto -v$ has exactly the same effect on the distribution as the inversion of the parity (even and odd) of the particles, as expected from \eqref{eq:tot pdf, 2P}. We also note that our definition of $\tilde\Delta$ differs slightly from the previous section; for convenience, we absorbed the factorials inside the determinant. 
	
	To lighten the notations, let us set $x_i\equiv x_i^{(m)}\quad (1\le i\le m)$ and $y_i\equiv x_i^{(m+1)}\quad (1\le i\le m+1)$.
	In anticipation of the proof, it is useful to rewrite the function $\Omega(x_i^{(m)})$ in terms of the level $m+1$,
	\begin{equation}
		\Omega(x_i) 
		=
		\left\{\begin{array}{ll}
			\Omega_{\text{sup}} = \frac{1}{2} +v \varepsilon_{m+1}(y_{i+1}) & \text{if } x_i=y_{i+1} ,
			\\[0.15cm]
			\Omega_{\text{inf}} = \frac{1}{2} -v \varepsilon_{m+1}(y_{i}) & \text{if } x_i=y_{i} ,
			\\[0.15cm]
			1 & \text{otherwise. }
		\end{array}\right.
	\end{equation}
	We also need the value of a sum of powers with modified boundary terms,
	\begin{equation}
		\sum_{x_i=y_i}^{y_{i+1}} x_i^{j-1}\Omega(x_i) 
		=
		\frac{y_{i+1}^j-y_i^j}{j} + v\big(\varepsilon_{m+1}(y_{i+1}) y_{i+1}^{j-1}-\varepsilon_{m+1}(y_i) y_i^{j-1}\big) + \sum_{k=2}^{j-1} \frac{B_k}{j} \binom{j}{k}  (y_{i+1}^{j-k}-y_i^{j-k}),
	\end{equation}
	as well as its alternating version,
	\begin{equation}
		\sum_{x_i=y_i}^{y_{i+1}} \varepsilon_{m}(x_i) x_i^{j-1}\Omega(x_i) 
		=
		-v( y_{i+1}^{j-1}- y_i^{j-1}) - \sum_{k=2}^{j} \frac{B_k}{j}\binom{j}{k} (2^k-1) \big(\varepsilon_{m+1}(y_{i+1}) y_{i+1}^{j-k}- \varepsilon_{m+1}(y_i)y_i^{j-k}\big),
	\end{equation}
	for which a proof is given in Appendix~\ref{apx:Appendix sum of powers}. We note that, due to the presence of $\varepsilon$ in the result of these summations, we no longer find polynomials. However, the result is still antisymmetric in the bounds of the summations, which is the important ingredient to extend the size of the determinant. We now present the proof for Proposition~\ref{prop:JPDF m to n, 2P}.

	\begin{proof}
		In the following, we proceed by induction. The case $m=1$ is verified since $\tilde\Delta(x^{(1)})=1$. Supposing that the result holds for $m$, we have to show that
		\begin{equation}
			\sum_{x^{(m)}} \tilde\Delta(x^{(m)}) \Omega(x^{(m)})
			\chit(x^{(m)}\prec x^{(m+1)})
			= 
			\tilde\Delta(x^{(m+1)}).
			\label{eq:induction step}
		\end{equation}
		If we pull all functions and sums inside the determinant, it remains to evaluate the quantity
		\begin{equation}
			\det_{1\le i,j\le m}\left[ \sum_{x_i=y_i}^{y_{i+1}} 
			\left(\frac{x_i^{j-1}}{(j-1)!} + \varepsilon_m(x) \sum_{k=0}^{j-2} \frac{x_i^{k}}{k!} \alpha_{j-2-k} \right)\Omega(x_i) \right].
		\end{equation}
		We compute the summations in the determinant by using the previous formulas and, by the argument of antisymmetry \eqref{eq:extension of antisym matrix}, extend the $m\times m$ determinant into a determinant of size $(m+1)$,
		\begin{equation}\begin{aligned}
				\det_{1\le i,j\le m+1}\Bigg[&
				\frac{y_{i}^{j-1}}{(j-1)!} + v \varepsilon_{m+1}(y_{i}) \frac{y_{i}^{j-2}}{(j-2)!}
				+
				\sum_{k=2}^{j-2} \frac{B_k}{k!} \frac{y_{i}^{j-1-k}}{(j-1-k)!}
				\\[0.2cm]&
				-
				\sum_{k=1}^{j-3} \alpha_{j-3-k}\left(v \frac{y_{i}^k}{k!} + \varepsilon_{m+1}(y_{i})\sum_{\ell=2}^{k+1} \frac{B_\ell}{\ell!}\frac{y_{i}^{k+1-\ell}}{(k+1-\ell)!}(2^\ell-1) \right)\Bigg],
				\label{eq:heavy determinant, 2P}
		\end{aligned}\end{equation}
		where we used the convention that $(-1)!=\infty$. The shift in the $j$ index is induced by the enlargement of the determinant. We want to prove that the previous determinant is equal to $\tilde\Delta(x^{(m+1)})$ by showing, by doing column operations, that its entries can be simplified to the following form
		\begin{equation}
			C_{j} =
			\frac{y_i^{j-1}}{(j-1)!} + \varepsilon_{m+1}(y_i) \sum_{k=0}^{j-2} \frac{y_i^{k}}{k!} \alpha_{j-2-k}.
			\label{eq:Column j}
		\end{equation}
		It is indeed the case for $j=1$, so we may assume it is true for the first $p$ columns and show that it is also correct for $j=p+1$. For this, we start by rearranging the summations in \eqref{eq:heavy determinant, 2P}, and express the entries of the $(p+1)$-th column as
		\begin{equation}
			\frac{y_i^{p}}{p!} + v \varepsilon_{p+1}(y_i) \frac{y_i^{p-1}}{(p-1)!}
			+
			\sum_{k=1}^{p-2} \frac{y_i^{k}}{k!} \left(\frac{B_{p-k}}{(p-k)!} - v \alpha_{p-2-k}\right) 
			-
			\varepsilon_{p+1}(y_i) \sum_{k=0}^{p-3} \frac{y_i^{k}}{k!} \sum_{\ell=2}^{p-1-k}  \frac{B_\ell}{\ell!}(2^\ell-1) \alpha_{p-1-k-\ell}.
			\label{eq:Column m+1}
		\end{equation}
		We apply the column operation
		\begin{equation}
			C_{p+1}\mapsto C_{p+1}'= C_{p+1}-\sum_{j=1}^{p-2} \left( \frac{B_{p-j}}{(p-j)!}-v \alpha_{p-2-j}\right)\,C_{j+1}
		\end{equation}
		in order to cancel the terms which do not contain $\varepsilon$ functions as prefactor (except the very first term). Let us note that, as $B_{j}$ and $\alpha_j$ are equal to zero for $j$ odd, only the columns of the same parity as $C_{p+1}$ really appear in this transformation, although it is implicit in our notations. One can verify that, after this operation, the $(p+1)$-th column of the determinant becomes
		\begin{equation}\begin{aligned}
				C_{p+1}'=
				\frac{y_i^{p}}{p!} + v \varepsilon_{p+1}(y_i) \frac{y_i^{p-1}}{(p-1)!}
				-
				\varepsilon_{p+1}(y_i) \sum_{k=2}^{p-1} \frac{y_i^{p-1-k}}{(p-1-k)!} \sum_{\ell=2}^{k}  \left(\frac{2^\ell B_\ell}{\ell!}  - v\,\alpha_{\ell-2}\right)\alpha_{k-\ell}.
		\end{aligned}\end{equation}
		Comparing this expression with the $(p+1)$-th column of $\tilde\Delta(x^{(m+1)})$, we see that the two are equal if and only if the coefficients $\alpha_k$ obey the following recurrence relation, for all $k\ge 1$,
		\begin{equation}
			\alpha_0 = v,
			\qquad
			\alpha_k = \sum_{\ell=2}^{k} \left(v \,\alpha_{\ell-2}\, \alpha_{k-\ell} - \frac{2^\ell B_\ell}{\ell!}  \, \alpha_{k-\ell}\right),
		\end{equation}
		which is indeed the case. This shows that all columns can be given the form appearing in $\tilde\Delta(x^{(m+1)})$ and concludes the proof of the identity \eqref{eq:induction step}.
		
		For the sake of completeness, we compute the generating function $f(z)=\sum_{k=0}^\infty \alpha_k z^k$ of these numbers. Multiplying each side of the recurrence relation by $z^k$ and summing over $k$ gives a quadratic equation for $f$,
		\begin{equation}
			v \,z^2 f(z)^2 -
			z\coth(z) f(z) + v = 0.
			\label{eq:quadra eq on generating function}
		\end{equation}
		Among the two solutions, we select the one for which $f(0)=\alpha_0$, namely
		\begin{equation}
			f(z) = 
			\frac{1}{2v z \tanh(z)}\left(1-\sqrt{1-4 v^2 \tanh^2(z)} \right).
		\end{equation}
	\end{proof}
	
	One may notice that, despite the complexity of the weighting system, the determinantal structure of the problem still allows us to perform explicit computations. However, the drawback is that the determinant can no longer be evaluated in a closed form.

	\medskip\noindent
	\underline{\it Remark 1.}\; The similarity between the generating function of Catalan numbers $c(z) = \frac{1}{2z}\left(1-\sqrt{1-4 z} \right) = \sum_{k=0}^\infty C_k z^k$ and the function $f$ allows us to rewrite the $\alpha_k$ in term of the $C_k=\binom{2k}{k}\frac{1}{k+1}$ by using the expansion
	\begin{equation}
		f(z) = \frac{v \tanh(z)}{z}\, c(v^2 \tanh^2(z)) =\sum_{k=0}^\infty z^k \sum_{\ell=0}^\infty \frac{C_\ell v^{2\ell+1}}{(k+1)!}\partial_w^{k+1}\tanh^{2\ell+1}(w)\Big|_{w=0}.
	\end{equation}

	\medskip\noindent
	\underline{\it Remark 2.}\; By taking $v=0$ (equivalently $a=b$), one directly recovers the result on the uniform AD since in this case $\alpha_m=0$ for all $m$, and the determinant in the density becomes Vandermonde. Another interesting limit of these numbers is when $v\to \pm1/2$. This is equivalent to taking extreme values for the weights; $a/b\to 0$ or $a/b\to \infty$. In this case, the $\alpha_m$ numbers also simplify. At the level of the generating function, we have
	\begin{equation}
		\frac{1}{2v z \tanh(z)}\left(1-\sqrt{1-4 v^2 \tanh^2(z)} \right)\Big|_{v=\pm1/2}
		=
		\pm\frac{\tanh(z/2)}{z} = \mp g(z).
	\end{equation}
	The function $g$ is the exponential generating function for the Genocchi numbers (shifted by two units).

	\medskip\noindent
	\underline{\it Remark 3.}\; Similarly to the biased case, it is possible to find the partition function by evaluating the distribution \eqref{eq:partially int pdf 2P} at $m=n+1$, which gives the equality 
	\begin{equation}
		Z_n
		= 
		(a^2 +b^2)^{\frac{n(n+1)}{2}} \tilde\Delta(x_i^{(n+1)}).
	\end{equation}
	It requires however to compute the determinant $\tilde\Delta(x_i^{(n+1)}=i-1)$ on the last particle level. The computation can be performed recursively where the induction step consists of a triangularization of the first two columns. The remaining determinant after this operation can be shown to be proportional to $\tilde\Delta(x_i^{(n-1)}=i)=\tilde\Delta(x_i^{(n-1)}=i-1; -v)$ which concludes the induction argument. Although seemingly not so complicated, the proof involves tedious computations and is not so much relevant for our purposes. We merely give the final result:
	\begin{proposition}
		The determinant $\tilde\Delta_{n+1}=\tilde\Delta(x_i^{(n+1)}=i-1)$ can be computed exactly, and is given by 
		\begin{equation}
			\tilde\Delta_{n+1}
			= 
			(1-4v^2)^{\frac{1}{2}\lfloor\frac{(n+1)^2}{4}\rfloor}\times
			\left\{\begin{array}{cl}
				1 & \text{if } n\ne 1 \mod 4, \\
				\sqrt{\frac{1-2v}{1+2v}} & \text{if } n= 1 \mod 4.
			\end{array}\right.
			\label{eq:det last level}
		\end{equation}
		\label{prop:Delta tilde}
	\end{proposition}
	The partition function follows directly from this proposition,
	\begin{equation}
		Z_n
		= 
		(2ab)^{\lfloor\frac{(n+1)^2}{4}\rfloor}
		(a^2+b^2)^{\lfloor\frac{n^2}{4}\rfloor}
		\times
		\left\{\begin{array}{cl}
			1 & \text{if } n\ne 1 \mod 4, \\
			a/b & \text{if } n= 1 \mod 4.
		\end{array}\right.
	\end{equation}
	This quantity has been computed in \cite{di2014octahedron} by using the relation between Aztec diamonds and the octahedron recurrence and in \cite{ruelle2023lambda} by using $\lambda$-determinants. 
	
	We note that the expression \eqref{eq:det last level} is relatively simple, in the sense that it depends only on a certain power of factors $(1-2v)$ and $(1+2v)$. It means in particular that all the combinatorics hidden in the $\alpha_k$ coefficients seems to vanish when we compute the determinant on the last particle level, which is rather surprising.

	\subsection{Distribution on first levels}
	\label{sec:Integration from the right, 2P}
	
	Similarly to the biased case, it is possible to guess the distribution for the first $m$ levels and prove it to be correct by integrating from the higher to the lower levels. The strategy of the proof is the same as in Proposition~\ref{prop:int from the right}, only the complexity of the expressions involved increases. The distribution may seem hard to guess in this case, but a part of the intuition on how we initially derived it is given at the end of the section.
	
	\begin{proposition}
		The probability density of the particle process on the first $m$ levels is given by
		\begin{equation}
			\rho(x^{(1)},\dots,x^{(m)}) = A_{n,m} \det_{1\le i,j\le m} \left[ L_{n-m+i,x_j^{(m)}} \big((-1)^{m}v\big)\right]
			\cdot
			\Omega(x^{(1)},\dots,x^{(m-1)})
			\cdot
			\chit(x^{(1)}\prec\dots\prec x^{(m)}),
			\label{eq:pdf 1 to m, x 2P}
		\end{equation}
		where, for $i\ge 0$, the $L_{i,j}$'s are defined by the initial condition $L_{0,j}=\delta_{0,j}$, and the recurrence
		\begin{equation}
			L_{i+1,j}(v) = \big(L_{i,j}(v) +  L_{i,j-1}(-v)\big)\beta_{i}(v),
			\qquad
			\beta_{i}(v)=
			\left\{\begin{array}{cl}
				(1-2v)^{-1} & \text{if } i=0 \mod 4, \\
				(1+2v)^{-1} & \text{if } i=2 \mod 4, \\
				1 & \text{if } i=1 \mod 2,
			\end{array}\right.
			\label{eq:def L_i,j}
		\end{equation}
		while 
		\begin{equation}
			A_{n,m}(v)
			=
			\left( \frac{1-4v^2}{4}\right)^{m(n+1-m)/2} 
			\left\{\begin{array}{cl}
				\sqrt{\frac{1+2v}{1-2v}}    & \text{if } n=1 \pmod{4} 
				\text{ and } m=1 \pmod{2}, \\
				\sqrt{\frac{1-2v}{1+2v}}    & \text{if } n=3 \pmod{4} 
				\text{ and } m=1 \pmod{2},
				\\
				1      & \text{otherwise} .
			\end{array}\right.
			\label{eq:Anm}
		\end{equation}
		\label{prop:int from the right, 2P}
	\end{proposition}

	Observe that, due to the initial conditions, it is immediate that $L_{i,j}=0$ for $j<0$ and $j>i$. Similarly to the biased case, the recurrence generalizes Pascal's triangular relation but here with a coefficient $\beta_i$, that depends on $i$ mod $4$, and a replacement $v\mapsto-v$ in $L_{i,j-1}$. An analogous recurrence for related quantities $S_{i,j}$ has been derived in \cite{ruelle2022double} where refined partition functions of two-periodic Aztec diamonds were considered. The quantities introduced in \cite{ruelle2022double} are in fact closely related to ours by the transformation $S_{i,j}\prod_{k=0}^{i-1}\beta_k = L_{i,j}$. Before proving the proposition, we need the following results on the $L_{i,j}$ coefficients:
	\begin{lemma}
		For $n\ge 0$, we have
		\begin{equation}
			\sum_{k=0}^{n} L_{n,k} = A_{n,1}^{-1}(-v) = 2^n \prod_{k=1}^n \beta_{k} \beta_{k-1}.
			\label{eq:sum of Lij}
		\end{equation}
		\label{lemma:relation on the L_ij, 1}
	\end{lemma}
	
	\begin{lemma}
		For $i\ge 1$ and $j\in\mathbb{N}$, we have
		\begin{equation}
			-2\overline\beta_i \overline\beta_{i-1} \sum_{k=0}^{j}\big( L_{i,k}-2\beta_i \beta_{i-1}  L_{i-1,k}\big)\left[ \frac{1}{2}-v(-1)^{j}\right]^{\delta_{k,j}}= \overline L_{i,j}.
		\end{equation}
		\label{lemma:relation on the L_ij, 2}
	\end{lemma}
	The bar above some quantities is a shortcut notation to mean that they are evaluated at $-v$ rather than~$v$, for instance, $\overline\beta_i = \beta_i(-v)$. In this section, we will also frequently need some properties on the $\beta_k$ coefficients that we gather here:
	\begin{equation}
		\beta_k = \overline\beta_{k+2} = \beta_{k+4},
		\qquad
		\beta_k + \overline\beta_{k} = 2 \beta_k \overline\beta_{k},
		\label{eq:relation on beta}
	\end{equation}
	for all $k\in\mathbb{Z}$.

	\begin{proof}[\underline{Proof of Lemma~\ref{lemma:relation on the L_ij, 1}}]
		In some sense, this lemma may appear trivial; if we assume the proposition to be true, it just ensures that the distribution at $m=1$ is normalized,
		\begin{equation}
			A_{n,1}(v)\sum_{k=0}^{n} L_{n,k}(-v) = \sum_{k=0}^{n}\rho(x_1^{(1)}=k)=1.
			\label{eq:normalized distrib}
		\end{equation}
		The lemma can nevertheless be obtained independently. Hereafter, we show that the left- and the right-hand side of \eqref{eq:sum of Lij} obey the same recurrence relation. For $n=0$, it is clear that both sides equal $1$, and, denoting the rhs by $R_n$, that it satisfies $2\beta_{n+1}\beta_n R_n = R_{n+1}$. Therefore, it only remains to show that the lhs respect the same relation, namely
		\begin{equation}
			2\beta_{n+1}\beta_{n}\sum_{k=0}^{n} L_{n,k} = \sum_{k=0}^{n+1} L_{n+1,k}.
		\end{equation}
		By using the recurrence of $L_{i,j}$ on the rhs and playing with the properties of $\beta_k$, this is equivalent to
		\begin{equation}
			\sum_{k=0}^{n+1} \frac{L_{n,k}}{\beta_{n-1}} - \frac{\overline L_{n,k-1}}{\overline\beta_{n-1}} = 0.
		\end{equation}
		The expression inside the sum is equal to $L_{n-1,k}-L_{n-1,k-2}$ so that the summation on $k$ indeed gives zero. 
	\end{proof}

	\begin{proof}[\underline{Proof of Lemma~\ref{lemma:relation on the L_ij, 2}}]
		In this proof, we follow the same strategy as in the previous one, and show that the two sides of the equality obey the same recurrence relation. The case $i=1$ is easy to prove and only relies on the values of $L_{0,k}=\delta_{0,k}$ and $L_{1,k}=(\delta_{0,k}+\delta_{1,k})\beta_0$.
		For the recurrence, it is clear that the right-hand side obeys $(\overline L_{i,j}+ L_{i,j-1})\overline\beta_{i} =  \overline L_{i+1,j}$. Hence, it suffices to show the lhs satisfies the same relation. By inserting the lhs expression into $(\overline L_{i,j}+ L_{i,j-1}) \overline\beta_{i}$ instead of $\overline L_{i,j}$, we find it is equal to
		\begin{equation}
			-2\overline\beta_{i}\sum_{k=0}^{j} \left[\overline\beta_{i}\overline\beta_{i-1}L_{i,k} + \beta_{i}\beta_{i-1}\overline L_{i,k-1}-2\overline\beta_{i}\overline\beta_{i-1}\beta_{i}\beta_{i-1}(L_{i-1,k}+\overline L_{i-1,k-1})\right]\left[ \frac{1}{2}-v(-1)^{j}\right]^{\delta_{k,j}},
		\end{equation}
		where we used that $L_{i,j}=0$ for $j<0$ and $j>i$. Using the recurrence on the $L_{i,j}$'s and \eqref{eq:relation on beta}, we rewrite this expression as
		\begin{multline}
			-2\overline\beta_{i}\beta_i\sum_{k=0}^{j}\left[-\overline\beta_{i-1}L_{i,k} + \beta_{i-1}\overline L_{i,k-1}\right]\left[ \frac{1}{2}-v(-1)^{j}\right]^{\delta_{k,j}}
			\\=
			-2\overline\beta_{i+1}\overline\beta_{i} \sum_{k=0}^{j}\left[L_{i+1,k} -2 \beta_{i+1}\beta_{i} L_{i,k}\right]\left[ \frac{1}{2}-v(-1)^{j}\right]^{\delta_{k,j}}.
		\end{multline}
		This last equality brings the expression back to its original form evaluated at $i+1$, which concludes the proof.
	\end{proof}

	\begin{proof}[\underline{Proof of Proposition~\ref{prop:int from the right, 2P}}]
		We proceed by induction starting from $m=n+1$. In this case, the determinant is lower triangular, and, in order to recover \eqref{eq:tot pdf, 2P}, we need to show that $\prod_{k=0}^n L_{k,k}\big((-1)^{n+1}v\big)=(a^2+b^2)^{n(n+1)/2}/Z_n = [\tilde\Delta_{n+1}(v)]^{-1}$. The product can be computed by using
		\begin{equation}
			L_{k,k}(v)= (1-4v^2)^{-\lfloor\frac{k}{4}\rfloor}\times 
			\left\{\begin{array}{cl}
				1 & \text{if } k=0 \mod 4, \\
				(1-2v)^{-1} & \text{if } k=1 \mod 4, \\
				(1+2v)^{-1} & \text{if } k=2 \mod 4, \\
				(1-4v^2)^{-1} & \text{if } k=3 \mod 4,
			\end{array}\right.
		\end{equation}
		and by separating each value of $n$ mod $4$. One can verify that it gives the claimed result by comparison with \eqref{eq:det last level}.
		For the induction step, we have to show that the distribution integrated over its level $m+1$ gives back the density for levels $1$ to $m$. It is equivalent to verify that
		\begin{multline}
			\sum_{x^{(m+1)}} \det_{1\le i,j\le m+1} \left[ L_{n-m+i-1,x_j^{(m+1)}}\big((-1)^{m+1}v\big)\right]
			\cdot
			\Omega(x^{(m)})
			\cdot
			\chit(x^{(m)}\prec x^{(m+1)})
			\\=
			\frac{A_{n,m}}{A_{n,m+1}}\det_{1\le i,j\le m} \left[ L_{n-m+i,x_j^{(m)}}\big((-1)^{m}v\big)\right].
			\label{eq:induction step, 2P}
		\end{multline}
		As before, we set $x_j\equiv x_j^{(m)}\quad (1\le j\le m)$ and $y_j\equiv x_j^{(m+1)}\quad (1\le j\le m+1)$.
		Similarly to what we did in Section~\ref{sec:Integration from the right}, we write $\Omega(x^{(m)}) = \leftArr\Omega(x^{(m+1)})$ as a function of $x^{(m+1)}$ such that
		\begin{equation}
			\leftArr\Omega(y_j) 
			=
			\left\{\begin{array}{ll}
				\frac{1}{2}+v\varepsilon_{m}(x_j)=\frac{1}{2}-v(-1)^{m+1+x_{j}} & \text{if } y_j=x_j ,
				\\[0.15cm]
				\frac{1}{2}-v\varepsilon_{m}(x_{j-1})=\frac{1}{2}+v(-1)^{m+1+x_{j-1}} & \text{if } y_j=x_{j-1} ,
				\\[0.15cm]
				1 & \text{otherwise. }
			\end{array}\right.
		\end{equation}
		We rewrite the lhs in \eqref{eq:induction step, 2P} with the sums inside the determinant,
		\begin{equation}
			\det_{1\le i,j\le m+1}  \left[\sum_{y_{j}=x_{j-1}}^{x_j}  L_{n-m+i-1,y_{j}}\big((-1)^{m+1}v\big) \;
			\leftArr\Omega(y_{j})\right],
		\end{equation}
		where we defined $x_0=0$ and $x_m=n$. (Note that $\leftArr\Omega$ takes the value $1$ in these special cases since it corresponds to the boundary of the AD and not to the presence of a particle.) Without loss of generality, we absorb the sign $(-1)^{m+1}$ in a redefinition of $v\mapsto(-1)^{m+1} v$. We simplify this determinant with the following column operations: $C_j\mapsto \sum_{k=1}^j C_k$. Since we have $\leftArr\Omega(y_i=x_i)+\leftArr\Omega(y_{i+1}=x_i)=1$, the entries become
		\begin{equation}
			\sum_{y_{j}=0}^{x_j}  L_{n-m+i-1,y_{j}} \left[\frac{1}{2}-v(-1)^{x_{j}}\right]^{\delta_{y_{j},x_{j}}}.
		\end{equation}
		In the last column, the expression simplifies further: $\sum_{y=0}^{n}  L_{n-m+i-1,y} = 2^{n-m+i-1} \prod_{k=1}^{n-m+i-1} \beta_{k} \beta_{k-1}$ due to Lemma~\ref{lemma:relation on the L_ij, 1}. Since we want to reduce the size of the determinant, we choose to apply the row operations $R_i\mapsto R_i - 2\beta_{n-m+i-1} \beta_{n-m+i-2}R_{i-1}$ in order to obtain $0$'s in the entries of the last column, except for the top-right corner. The rest of the entries, for $i=2,\dots,m+1$ and $j=1,\dots,m$, are given by
		\begin{equation}
			\sum_{y_{j}=0}^{x_j} \big( L_{n-m+i-1,y_{j}}-2\beta_{n-m+i-1} \beta_{n-m+i-2}  L_{n-m+i-2,y_{j}}\big)\left[ \frac{1}{2}-v(-1)^{x_j}\right]^{\delta_{y_{j},x_j}}.
		\end{equation}
		We can use Lemma~\ref{lemma:relation on the L_ij, 2} to simplify the expression and find something proportional to the desired entries. In summary, the determinant is now given by
		\begin{equation}
			\begin{vNiceArray}{cw{c}{5cm}c|c}[margin]
				\Block{1-3}{\sum_{y=0}^{x_j}  L_{n-m,y} \;\leftArr\Omega(y ; (-1)^{m+1} v)} & & &  2^{n-m} \prod_{k=1}^{n-m} \beta_{k} \beta_{k-1} \\[0.2cm]
				\hline
				\Block{3-3}{-\dfrac{L_{n-m+i-1,x_j}}{2\overline\beta_{n-m+i-1}\overline\beta_{n-m+i-2}} } & & & 0  \\
				& & & \vdots \\
				& & & 0 \\
			\end{vNiceArray}_{(m+1)\times (m+1)}
			\simeq
			\det_{1\le i,j\le m} \Big[ \overline L_{n-m+i,x_j}\Big].
		\end{equation}
		After a tedious computation, the proportionality factor between the two determinants can be shown to correspond exactly to $A_{n,m}\big((-1)^{m+1}v\big)/A_{n,m+1}\big((-1)^{m+1}v\big)$.
	\end{proof}

	\medskip\noindent
	\underline{\it Remark.}\; In this proof, we used the knowledge of the partition function to slightly simplify the expression of $A_{n,m}$. However, it is not difficult to rewrite it while keeping $Z_n$ implicit, by defining $\tilde A_{n,m} = A_{n,m} \frac{Z_n}{(a^2+b^2)^{n(n+1)/2}}$. The induction step would not change since $A_{n,m-1}/A_{n,m} = \tilde A_{n,m-1}/\tilde A_{n,m}$. In the end, it would be easy to deduce the partition function from the proposition, just by integrating the case $m=1$. 
	\medskip

	By the same integration performed in Proposition~\ref{prop:JPDF m to n, 2P}, we can write down the one level density distribution\footnote{Similarly, it is not difficult to obtain the distribution $\rho(x^{(m)},\dots,x^{(m+k)})$ for $k+1$ consecutive levels.} at finite size
	\begin{equation}
		\rho(x^{(m)})= A_{n,m} \,
		\det_{1\le i,j\le m}  \left[L_{n-m+i,x_j^{(m)}} \big((-1)^{m}v\big) \right]\cdot \tilde\Delta(x^{(m)}) ,
		\label{eq:proba x one line 2P}
	\end{equation}
	where $\tilde\Delta(x^{(m)})$ is still given by \eqref{eq:Vandermonde 2P}. It is a relatively simple formula given by the product of two determinants of identical sizes. We have not been able to evaluate these determinants in closed form and the product of the two does not seem to simplify the expression. The question of whether this last expression can be related to orthogonal or biorthogonal polynomials (similarly to the biased case with the Krawtchouk ensemble) remains open.

	Another question is whether it is still possible to obtain the shape of the (outer) arctic curve by looking at the support of this distribution in the limit where $n$ becomes large (method proposed by \cite{fleming2011interlaced} in the uniform case). In practice, the above distribution contains all the necessary information; the problem rather lies in the feasibility of the computation from the determinantal expressions.

	In the remainder of this section, we attempt to give a better understanding of the distribution \eqref{eq:pdf 1 to m, x 2P}. Originally, following the method proposed in \cite{fleming2011interlaced}, we deduced it from \eqref{eq:partially int pdf 2P} by using the Jacobi identity, sometimes also referred to as Jacobi complementary minor theorem. For a square matrix $M$, the formula reads
	\begin{equation}
		\det M_{J,K} = (-1)^{\sum_i j_i + k_i} \det\left[ (M^{-1})_{K^c,J^c}\right] \det M ,
		\label{eq:Jacobi formula}
	\end{equation}
	where $J=\{j_1,j_2,\dots\},K=\{k_1,k_2,\dots\}$ are subsets of rows and columns of $M$, and $J^c, K^c$ their complements. The above formula makes the bridge between the minors of a given matrix and the complementary minors of its inverse. Interestingly, this formula can be used to connect the distribution of particles with the distribution of holes (the absence of particles). For instance, we can take $M$ to be $\tilde\Delta_{n+1}(\pm v)$. In this case, the rows selected in $J$ correspond to the particle positions, and consequently, $J^c$ gives the hole positions. 
	
	The difficulty then lies in the inverse matrix, which can be hard to determine in principle, but since the equality holds at the determinant level, the exact inverse of $M$ is not needed. 
	Suppose that we have a lower triangular matrix $L$ such that $L M = U$ with $U$ an upper triangular matrix with ones on its diagonal. The inverse of $M$ is given by $M^{-1} = U^{-1}  L $, however, due to its triangular shape, $U^{-1}$ do not contribute to the minors containing the last $m$ rows of $M^{-1}$ which are simply given by $|M^{-1}_{K^c,J^c} |= |L_{K^c,J^c}|$.
	
	Our claim is that such a matrix $\big(L_{i,j}\big)_{i,j=0}^n$ is actually given by the coefficients $L_{i,j}$ in \eqref{eq:def L_i,j} where the $-1$ prefactor in the Jacobi's formula has been absorbed inside the determinant. The exact sign of $v$ in the different expressions depends on the parity of $n$ and $m$. Once we have the distribution of holes, one can use the symmetries of the Aztec diamonds to exchange the roles of particles and holes and recover the distribution \eqref{eq:pdf 1 to m, x 2P}. It is also needed to know how to relate the adjacency of particles and holes, but this can be done by explicitly counting the different kinds of dimers on a given level according to the particle positions.

	\subsection{Generating functions, asymptotics and convergence to the GUE-corners process}
	\label{sec:conv to GUE 2P}
	
	In this section, our first objective is to find an asymptotic formula for the $L_{i,j}$ coefficients, allowing us to understand the behaviour of the determinant near the edge in the large $n$ limit. Such an analysis for a similar recurrence has been carefully realized in \cite{ruelle2022double} so that we will not detail all the computations in the present paper. Since the $L_{i,j}$'s obey a linear recurrence relation, their generating function is rational, and the tools developed in \cite{pemantle2008twenty} are quite adapted. The only difficulty lies in the functions $\beta_i(v)$ that depends on the value of $i$ modulo $4$. It leads us to split the generating function into four functions, according to each subsequence. In the regime where $i$ and $j$ become large, the large deviation function leading to the asymptotics of these subsequences is the same; however, the prefactor and subleading terms may slightly differ for each of them. 
	
	The theorem 3.19 of \cite{pemantle2008twenty} gives the asymptotics for large $n$ of the coefficients
	\begin{equation}
		L_{sn,rn}(v)=\frac{1}{\sqrt{2\pi sn}} \text{e}^{sn F(\frac{r}{s})} 
		\sum_{k=0}^\infty \frac{N_k^{[sn]}(\frac{r}{s})}{(sn)^{k/2}},
		\label{eq:asympt formula pemantle}
	\end{equation}
	for $s,r\in [0,1]$. $F(t)$ is the large deviation function and depends, in practice, only on the denominator of the generating function. The $N_k^{[sn]}(t)$ give the corrections at all orders in $n$ and depend on derivatives of numerator and denominator of the generating function, they are therefore supposed to be smooth in $t$; these correction terms also depend implicitly on $v$. The exponent $[sn]$ reminds that the function may depend on $sn$ but only on its value modulo 4 so that we have here four different functions $N_k^{[sn]}$ for each $k\in\mathbb{N}$. An explicit form of these corrections is not needed for our purposes\footnote{A formula for $N_0^{[sn]}$ is given in \cite{pemantle2008twenty}.}; however, since several cancellations occur in the determinant, we need to verify they do not interfere with the leading order. Once we have a clear knowledge of these cancellations, one can always truncate this series to only keep the necessary terms. 
	
	Concerning the function $F$, it turns out that the discussion made in Section~4.1 of \cite{ruelle2022double} applies almost verbatim in our case: in particular, it can be shown that $F(t)=F_1(t)-\frac{1}{2}\log\sqrt{1-4v^2}$, where $F_1$ is defined in the equation (4.12) of \cite{ruelle2022double}. The function $F$ can then be proven to be strictly concave with a unique maximum at $r_*=1/2$ where $F(r_*)=\log\frac{2}{\sqrt{1-4v^2}}$ and $F''(r_*)=\frac{-4}{1-4v^2}$.
	We can now formulate the main result of this work.
	
	\begin{theorem}
		If we define $z_i^{(\ell)} = \frac{x_i^{(\ell)} - \mu}{\sigma}$ to be the rescaling of the $x$-particle system, with $\mu=n r_*$ and $\sigma^2=\frac{n}{-F''(r_*)}$, then the JPDF $\rho(z^{(1)},\dots,z^{(m)})$ converges in distribution to $\rho_{\rm{GUE-corners}}(z^{(1)},\dots,z^{(m)})$ in the limit where $n\to\infty$ and $m$ remains finite.
	\end{theorem}

	\begin{proof}
		We start by recalling the expression of the particle distribution at finite $n$, given in (\ref{eq:pdf 1 to m, x 2P}),
		\begin{equation}
			\rho(x^{(1)},\dots,x^{(m)}) = A_{n,m} \det_{1\le i,j\le m} \left[ L_{n-m+i,x_j^{(m)}} \big((-1)^{m}v\big)\right]
			\cdot
			\Omega(x^{(1)},\dots,x^{(m-1)})
			\cdot
			\chit(x^{(1)}\prec\dots\prec x^{(m)}).
		\end{equation}
		
		As in Section \ref{sec:conv to GUE biased}, we neglect the factor $\Omega(x^{(1)}, \ldots,x^{(m-1)})$ since the probability that some particles are adjacent decays to 0 in the large $n$ limit. On level $m \ll n$, the particles are close to the contact point of the arctic curve with the boundary, located at $\mu = nr_* = \frac n2$. We set\footnote{The superscript $m$ is omitted in the definition of the $r_j$ to keep the notation lighter.} $r_j=x_j^{(m)}/n$ and $s_i=1-\frac{m-i}{n}$ in order to make contact with the asymptotic form \eqref{eq:asympt formula pemantle} for the $L_{sn,rn}$ coefficients. Thus the $s_i$ are all close to 1, and the $r_j$ are close to $r_*$ and related to the $z_j$ by the relation $r_j = r_* + z_j/\sqrt{-nF''(r_*)} + \ldots$\,, where the dots refer to the fluctuations of $r_j$ around $r_*$ of scales larger than $\sqrt{n}$.
		
		From \eqref{eq:asympt formula pemantle}, the coefficient $L_{s_in,r_jn}$ contains the function $s_i n F(r_j/s_i)$ which we can expand around $s_i=1$ and $r_j=r_*$,
		\begin{align}
			s_i n F(r_j/s_i) &= n F(r_j) - (m-i) F(r_j) + (m-i) r_j F'(r_j) + \ldots \nonumber\\
			&= n F(r_*) - \frac 12 z_j^2 - (m-i) \big[ F(r_j) - r_j F'(r_j)\big] + \ldots
			\label{eq:sin}
		\end{align}
		where the dots are terms which go to 0 when $n$ goes to infinity.
		
		Using \eqref{eq:asympt formula pemantle}, we can now write the determinant as 
		\begin{equation}
			\det_{1\le i,j\le m} \left[ L_{s_i n,r_j n} \right] = (2\pi n)^{-m/2} \, {\rm e}^{mnF(r_*)} \, \prod_{j=1}^m\text{e}^{-\frac{1}{2}z^2_j} \times \det_{1\le i,j\le m} \left[G_i(r_j)\right],
			\label{eq:det}
		\end{equation}
		where the functions $G_i(r_j)$ are given by
		\begin{equation}
			G_i(r_j) = \frac 1{\sqrt{s_i}} \, {\rm e}^{s_i n F(r_j/s_i) - n F(r_*) + z_j^2/2} \, \sum_{k=0}^\infty \frac{N_k^{[s_in]}(r_j/s_i)}{(s_in)^{k/2}}.
		\end{equation}
		From \eqref{eq:sin}, each function $G_i(r)$ has an expansion in inverse powers of $n$, whose dominant, order 0 term is 
		\begin{equation}
			\tilde{N}_0^{[n-m+i]}(r) \equiv G_i(r)\Big|_{n^0} = N_{0}^{[n-m+i]}(r) \, \text{e}^{(i-m) [F(r)-r F'(r)]}.
			\label{eq:zeroth}
		\end{equation}
		
		Before evaluating the asymptotic value of this determinant, we examine the other various factors depending on $n$ that will enter the probability distribution in the scaling limit. There are four such terms: the $A_{n,m}$ factor, given in \eqref{eq:Anm}, the Jacobian of the change of variables $x_j \mapsto z_j$, equal to $\sigma^{m(m+1)/2}$, and the terms $(2\pi n)^{-m/2}$ and ${\rm e}^{mnF(r_*)}$ in \eqref{eq:det}. Collecting their explicit values, we find
		\begin{equation}
			(2\pi n)^{-m/2} A_{n,m} \, \sigma^{m(m+1)/2} \, {\rm e}^{mnF(r_*)} = \frac{C_{n,m}(v)}{(2\pi)^{m/2}} \, \Big(\frac 4{1-4v^2}\Big)^{m(m-3)/4} \, n^{m(m-1)/4},
		\end{equation}
		where $C_{n,m}(v)$ is the function on $v$, given in \eqref{eq:Anm}, which depends on $n \bmod 4$ and on $m \bmod 2$.
		
		Collecting all terms, the scaling limit of the discrete distribution is given in terms of the remaining determinant as
		\begin{equation}
			\squeeze{0.6}{
				\rho(z^{(1)},\ldots,z^{(m)}) = \lim_{n \to \infty} \frac{C_{n,m}(v)}{(2\pi)^{m/2}} \, \Big(\frac 4{1-4v^2}\Big)^{m(m-3)/4} \, n^{m(m-1)/4} \, \prod_{j=1}^m\text{e}^{-\frac{1}{2}z^2_j} \, \det_{1\le i,j\le m} \left[G_i(r_j)\right] \cdot \, \chit(z^{(1)}\prec\dots\prec z^{(m)}). 
			}
			\label{eq:rho}
		\end{equation}
		
		To evaluate the determinant, we expand the functions $G_i(r_j)$ around $r_j = r_* + \varepsilon_j$, to obtain 
		\begin{equation}
			G_i(r_j) = G_i(r_*) + \varepsilon_j \, G'_i(r_*) + \frac {\varepsilon_j^2}2\,G''_i(r_*) + \ldots + \frac {\varepsilon_j^{m-1}}{(m-1)!}\,G^{(\ge m-1)}_i(r_*),
		\end{equation}
		where $G^{(\ge m-1)}_i(r_*)$ contains the tail of the Taylor expansion, namely,
		\begin{equation}
			G^{(\ge m-1)}_i(r_*) = G^{(m-1)}_i(r_*) + \frac{\varepsilon_j}{m} \, G^{(m)}_i(r_*) + \frac{\varepsilon_j^2}{m(m+1)} \, G^{(m+1)}_i(r_*)+ \ldots
		\end{equation}
		The matrix inside the determinant factorizes so that, with $G^{(m-1)}_i(r_*)$ standing for $G^{(\ge m-1)}_i(r_*)$,
		\begin{equation}
			\det_{1\le i,j\le m} \left[G_i(r_j)\right] = \det_{1\le i,j\le m} \left[\sum_{k=1}^m \frac 1{(k-1)!}\,G_i^{(k-1)}(r_*) \, \varepsilon_j^{k-1} \right] = \det_{1\le i,j\le m} \Big[\frac 1{(j-1)!}\,G_i^{(j-1)}(r_*) \Big] \: \det_{1\le i,j\le m} \big[\varepsilon_j^{i-1}\big].
		\end{equation}
		The Vandermonde determinant yields, upon using $\varepsilon_j = z_j/\sqrt{-nF''(r_*)} + \ldots$\,,
		\begin{equation}
			\det_{1\le i,j\le m} \big[\varepsilon_j^{i-1} \big] = \Delta\Big(\frac{z_j^{(m)}}{\sqrt{-n F''(r_*)}} + \ldots\Big) = \big(\!\!-\!nF''(r_*)\big)^{-m(m-1)/4} \, \Delta(z_j^{(m)}) + \ldots
		\end{equation}
		Thus its dominant order in $n$ precisely cancels the power $n^{m(m-1)/4}$ in the expression \eqref{eq:rho}, implying, in the large $n$ limit, that the corrections in the previous equation will decay to zero, and that the contribution of the derivatives $G_i^{(j-1)}(r_*)$ will be given by the derivatives of their zeroth order in $n$, that is,
		\begin{equation}
			B_m^{[n]} \equiv \lim_{n \to \infty} \, \det_{1\le i,j\le m} \Big[\frac 1{(j-1)!}\,G_i^{(j-1)}(r_*) \Big] = \det_{1\le i,j\le m} \Big[\frac 1{(j-1)!}\,\frac{{\rm d}^{j-1}\tilde{N}_0^{[n-m+i]}(r)}{{\rm d}r^{j-1}}\Big|_{r_*} \Big].
		\end{equation}
		
		We obtain the following final result for the particle distribution in the scaling limit,
		\begin{equation}
			\rho(z^{(1)},\ldots,z^{(m)}) = \frac{1}{(2\pi)^{m/2}} \, \left\{C_{n,m}(v) \, B_m^{[n]} \, \Big(\frac {1-4v^2}4\Big)^{m/2}\right\} \, \Delta\big(z_j^{(m)}\big) \, \prod_{j=1}^m\text{e}^{-\frac{1}{2}(z_j^{(m)})^2} \, 
			\chit(z^{(1)}\prec\dots\prec z^{(m)}), 
			\label{eq:final}
		\end{equation}
		precisely the GUE-corners distribution except for the factor within the curly brackets. The proper normalizations of both the discrete distribution we started from and of the GUE-corners process ensure that this factor must be equal to 1.
	\end{proof}

	\medskip\noindent
	\underline{\it Remark.}\; It would be clearly desirable to check explicitly that the factor within the brackets is actually equal to 1. Although the final result is surprisingly simple, the explicit calculation of $B_m^{[n]}\big((-1)^m v\big)$ looks awkward: it requires multiple derivatives of the functions $F(r)$ and $N_0^{[0]}(r),\ldots,N_0^{[3]}(r)$, which are all rather complicated for $v\neq 0$. We have nevertheless checked that the factor is indeed equal to 1 for $m \le 4$. For $v=0$ however the situation simplifies drastically. For any $v$, the arguments in the above proof indeed show that, in the large $n$ limit, the determinant of the matrix $\big(G_i(r_j)\big)_{i,j}$ can be reduced to the determinant of its zeroth order $\big(\tilde{N}_0^{[n-m+i]}(r_j)\big)_{i,j}$. However, for $v=0$, the functions $N_0^{[n-m+i]}(r)=N_0(r)$ do not depend on $i$, and we obtain
	\begin{align}
		&\det_{1\le i,j\le m}  \Big[N_{0}^{[n-m+i]}(r_j) \, \text{e}^{(i-m) [F(r_j)-r_j F'(r_j)]}\Big] 
		=
		\Big[ \prod_{j=1}^m N_0(r_j) \Big] \, (-1)^{m(m-1)/2} \, \Delta\Big(\text{e}^{-[F(r_j)-r_j F'(r_j)]}\Big) 
		\nonumber\\ & \qquad=
		\big[N_0(r_*)\big]^m \, \Big[\big(\!\!-\!r_* F''(r_*)\big) {\rm e}^{-F(r_*)}\Big]^{m(m-1)/2} \, \Delta\Big(\frac{z_j^{(m)}}{\sqrt{-n F''(r_*)}}\Big).
	\end{align}
	As stated in Proposition \ref{prop:int from the right}, $L_{i,j}$ is given by a binomial coefficient, $L_{i,j} = \binom{i}{j}$. The asymptotic value relevant here is 
	\begin{equation}
		L_{n,rn} \simeq \frac{r^{-rn} (1-r)^{-(1-r)n}}{\sqrt{2\pi n \,r(1-r)}},
	\end{equation}
	from which one obtains $F(r) = -r \log r - (1-r) \log{(1-r)}$ (implying, as announced above, $r_*=\frac 12$, $F(r_*) = \log 2$ and $F''(r_*)=-4$) and $N_0(r) = [r(1-r)]^{-1/2}$, with $N_0(r_*)=2$. Then, the coefficient in front of the Vandermonde determinant, to be identified with $B_m^{[n]}$, yields $B_m^{[n]}(0)=2^m$. With $C_{n,m}(0)=1$, we obtain that the factor in the curly brackets in \eqref{eq:final} is indeed equal to 1.

	\section{Marked GUE-corners processes}
	\label{sec:marking}
	
	In this section, we discuss the results obtained in the previous section in the light of recent results by Berggren and Bradinoff \cite{berggren2026marked}. As mentioned earlier, they considered the more general $2 \times \ell$ periodic measure, but for comparison purposes, we concentrate here on the $\ell=2$ case. The general method these authors have used is completely different from the way we proceeded, since they started from the correlation kernel of the discrete process and then evaluated its asymptotics in the scaling limit. 
	
	In that limit, they recovered the correlation kernel of the GUE-corners process modulated by an extra factor, which is the hallmark of a marking (more precise information on markings can be found in \cite{claeys2023marking}). The marking affects the correlations in a non-trivial way, and the particle density distributions in particular, by multiplicative factors which may depend on the weights defining the measure on the domino tilings and on the positions of the particles. The marking appears to be the remnant of a distinction made in the discrete process among the particles according to the parity of their position on a given level.
	
	As our results show no trace of a marking, a natural question is to clarify the reasons of why it is so. In what follows, we will not present proofs but fairly simple arguments for why and under which conditions a marking should be expected. 
	
	\subsection{Marking vs no marking}
	\label{sec:markingvsnomarking}
	
	For that, it is crucial to understand in details the setting used in \cite{berggren2026marked} for the case $\ell=2$, especially regarding the two-periodic weights and the particle system considered. There are in addition a few differences with respect to our setting, like the fact that we focus, after rotation of the Aztec diamond, on the left contact point while the analysis in \cite{berggren2026marked} examines the right contact point (this will however have a subtle consequence, see below). In this regard, and in order to avoid further confusion, we translate the conventions of \cite{berggren2026marked} into the setting used in the previous sections.
	
	\begin{figure}[h]
\begin{center}
\psset{unit=0.75cm}
\pspicture(-2,-0.5)(3,4.5)
\psline[linewidth=1pt,linecolor=black](-1,3)(0,4)
\psline[linewidth=1pt,linecolor=black](-1,1)(2,4)
\psline[linewidth=1pt,linecolor=black](0,0)(3,3)
\psline[linewidth=1pt,linecolor=black](2,0)(3,1)
\psline[linewidth=1pt,linecolor=black](-1,3)(2,0)
\psline[linewidth=1pt,linecolor=black](-1,1)(0,0)
\psline[linewidth=1pt,linecolor=black](0,4)(3,1)
\psline[linewidth=1pt,linecolor=black](2,4)(3,3)
\rput(-1,0.4){\footnotesize $\alpha_1^{-1}$}
\rput(0.2,0.6){\footnotesize $\beta_1^{-1}$}
\rput(1.85,0.7){\footnotesize $\alpha_2^{-1}$}
\rput(3,0.4){\footnotesize $\beta_2^{-1}$}
\rput(-0.6,1.7){\footnotesize $1$}
\rput(0.6,1.7){\footnotesize $1$}
\rput(1.4,1.7){\footnotesize $1$}
\rput(2.6,1.7){\footnotesize $1$}
\rput(0,2){
\rput(-0.8,0.4){\footnotesize $\alpha_1$}
\rput(0.2,0.6){\footnotesize $\beta_1$}
\rput(1.75,0.7){\footnotesize $\alpha_2$}
\rput(2.8,0.4){\footnotesize $\beta_2$}
\rput(-0.6,1.7){\footnotesize $1$}
\rput(0.6,1.7){\footnotesize $1$}
\rput(1.4,1.7){\footnotesize $1$}
\rput(2.6,1.7){\footnotesize $1$}
}
\endpspicture
\hspace{1cm}
\pspicture(-2,-0.5)(3,4.)
\psline[linewidth=1pt,linecolor=black](-1,3)(0,4)
\psline[linewidth=1pt,linecolor=black](-1,1)(2,4)
\psline[linewidth=1pt,linecolor=black](0,0)(3,3)
\psline[linewidth=1pt,linecolor=black](2,0)(3,1)
\psline[linewidth=1pt,linecolor=black](-1,3)(2,0)
\psline[linewidth=1pt,linecolor=black](-1,1)(0,0)
\psline[linewidth=1pt,linecolor=black](0,4)(3,1)
\psline[linewidth=1pt,linecolor=black](2,4)(3,3)
\rput(-1,0.4){\footnotesize $\alpha_1^{-1}$}
\rput(0.2,0.6){\footnotesize $\alpha_1^{-1}$}
\rput(1.75,0.7){\footnotesize $\alpha_1$}
\rput(2.9,0.4){\footnotesize $\alpha_1$}
\rput(-0.6,1.7){\footnotesize $1$}
\rput(0.6,1.7){\footnotesize $1$}
\rput(1.4,1.7){\footnotesize $1$}
\rput(2.6,1.7){\footnotesize $1$}
\rput(0,2){
\rput(-0.8,0.4){\footnotesize $\alpha_1$}
\rput(0.2,0.6){\footnotesize $\alpha_1$}
\rput(1.85,0.7){\footnotesize $\alpha_1^{-1}$}
\rput(2.9,0.4){\footnotesize $\alpha_1^{-1}$}
\rput(-0.6,1.7){\footnotesize $1$}
\rput(0.6,1.7){\footnotesize $1$}
\rput(1.4,1.7){\footnotesize $1$}
\rput(2.6,1.7){\footnotesize $1$}
}
\endpspicture
\hspace{1cm}
\pspicture(-2,-0.5)(3,4.)
\psline[linewidth=1pt,linecolor=black](-1,3)(0,4)
\psline[linewidth=1pt,linecolor=black](-1,1)(2,4)
\psline[linewidth=1pt,linecolor=black](0,0)(3,3)
\psline[linewidth=1pt,linecolor=black](2,0)(3,1)
\psline[linewidth=1pt,linecolor=black](-1,3)(2,0)
\psline[linewidth=1pt,linecolor=black](-1,1)(0,0)
\psline[linewidth=1pt,linecolor=black](0,4)(3,1)
\psline[linewidth=1pt,linecolor=black](2,4)(3,3)
\rput(0,1){$a$}
\rput(2,1){$b$}
\rput(0,3){$b$}
\rput(2,3){$a$}
\rput(1,2){$1$}
\endpspicture
\caption{The left panel shows the edge weights used in \cite{berggren2026marked}. The middle panel shows the weights upon the specialization $\alpha_1^{-1} = \beta_1^{-1} = \alpha_2 = \beta_2$, which are themselves gauge equivalent to the face weights shown on the right panel. The face weight systems $(a,b),(\frac ab,1)$ and $(1,\frac ba)$ are gauge equivalent.}
\label{fig:edgeweights}
\end{center}
\end{figure}

\vskip -.5truecm
	
			The authors of \cite{berggren2026marked} use the periodic edge weights on the fundamental domain as shown on the left panel of Figure~\ref{fig:edgeweights}, on which they impose the further conditions $\alpha_1^{-1} = \beta_1^{-1} = \alpha_2 = \beta_2$. This choice of weights is gauge equivalent to the face weights shown on the right panel, which are precisely the weights we have used in the preceding sections, with the same convention that the lowest leftmost face has weight $a$. The relation between the edge and the face weights reads $\alpha_1^{-1} = \frac ab \equiv \sqrt\tau$.
			
			The next issue concerns the particle system. In the preceding sections, the particles were located on the left ends of the S and E dimers, and aligned on vertical levels labelled by an integer $\ell$, ranging from 1 to $n$. On each level, there are $n+1$ possible positions, taking values from 0 to $n$, counting from bottom to top. The choice in \cite{berggren2026marked} is different, as the particles are located on the right ends of the S and E dimers, and are therefore aligned on vertical lines, labelled by an integer $t$ between 1 and $n$, which lie in-between our level lines. The possible vertical positions then take $n$ values, from 0 to $n-1$, again counting from bottom to top. For convenience, we will refer to red particles for the particle system chosen here, and to blue particles for the particle system of \cite{berggren2026marked}. The Figure~\ref{fig:redblue} illustrates the two particle systems associated with the same configuration (the face weights $c,d$, set equal to 1 for the moment, will be discussed shortly).

\begin{figure}[h]
\begin{center}
\includegraphics[width=0.7\linewidth]{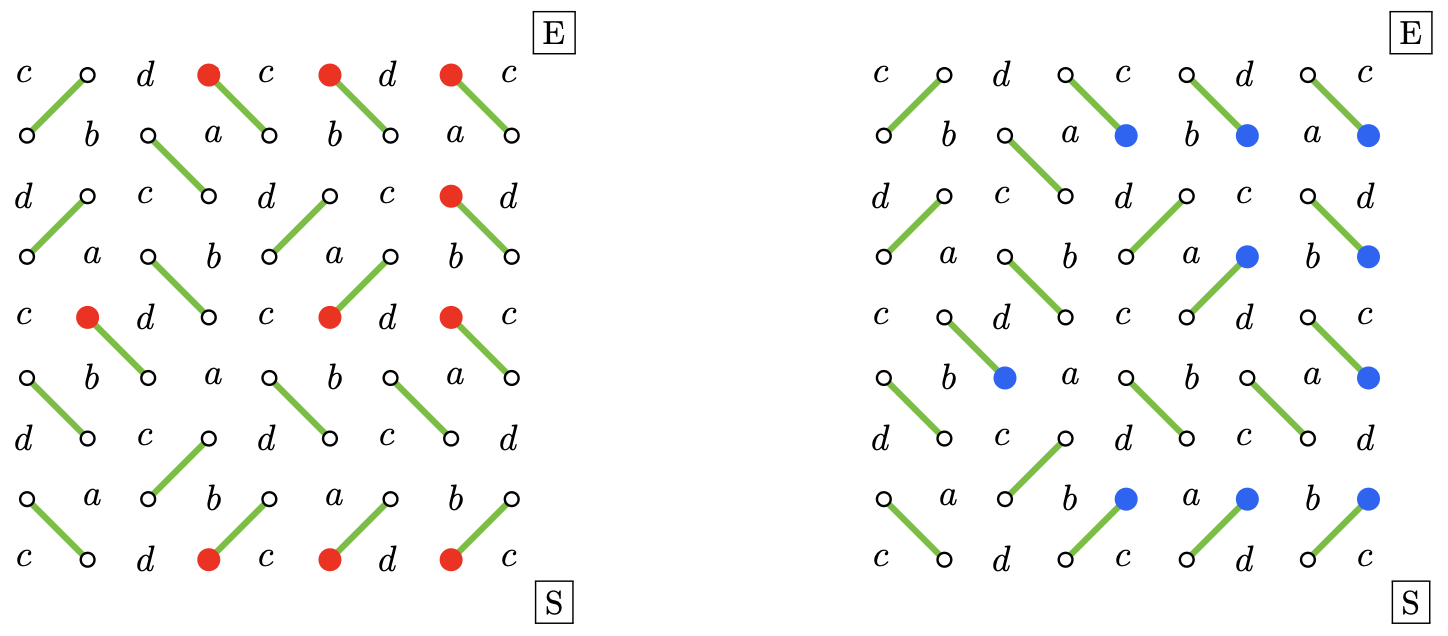}
\caption{Two particle systems: on the left, the red particles as defined here, on the right, the blue particles as defined in \cite{berggren2026marked}.} 
\label{fig:redblue}
\end{center}
\end{figure}

	In \cite{berggren2026marked}, the (blue) particles are marked according to the parity of their vertical position, getting a mark $j=0$ if the position is even, $j=1$ if it is odd. In the scaling limit, the marking leaves a trace in the form of a factor affecting the correlation kernel. This factor, denoted by $\nu(t,j)$ in \cite{berggren2026marked}, is a function which depends solely on the combination $j + t \bmod 2$ of the mark of the particle and the parity of the level $t$ where it is located, and is therefore a function of the parity we have defined in (\ref{eq:parity}). It is then natural to distinguish the mark $j=0,1$ of a particle, defined as the parity of its vertical position, from the parity $\varepsilon$ of the particle itself, defined as $\varepsilon = (-1)^{\ell+j}$ or $(-1)^{t+j}$ for the red resp.~blue particles. A particle is even resp.~odd if its parity is $+1$ resp.~$-1$.	
	
	For $\ell=2$, the factors $\nu(t,j)$, written now as $\nu(\varepsilon)$, explicitly read (see \cite{berggren2026marked}, Example 3.2\footnote{Compared to \cite{berggren2026marked}, the two cases in (\ref{nu}) are swapped, equivalently $a$ and $b$ are interchanged. This is because the authors of \cite{berggren2026marked} have rotated clockwise the Aztec diamond by 45 degrees, while we have rotated it anticlockwise by 45 degrees. As a consequence, from the perspective of \cite{berggren2026marked}, the lowest leftmost face has weight $b$, and not $a$ as we have assumed here.}),
	
	\begin{equation}
	\nu(\varepsilon) = \begin{cases}
	\displaystyle \frac {\alpha_1^2}{1+\alpha_1^2} = \frac 1{1+\tau} & \text{if $\varepsilon = +1$},\\
	\noalign{\smallskip}
	\displaystyle \frac 1{1+\alpha_1^2} = \frac \tau{1+\tau} & \text{if $\varepsilon = -1$}.
	\end{cases}
	\label{nu}
	\end{equation}
	
	A direct consequence of the marking is that the factor $\nu(\varepsilon)$ will affect the correlation functions, and in particular, the particle density distribution on any level $t$. In fact, on any level, there will be two distributions, one for the even sites, the other for the odd sites. If $\rho(t,z)$ denotes the particle density distribution on level $t$ for the unmarked GUE-corners process, the distributions arising in the scaling limit from the even and odd particles are given respectively by $\rho_{\rm ev}(t,z) = \nu(+1) \, \rho(t,z)$ and $\rho_{\rm odd}(t,z) = \nu(-1) \, \rho(t,z)$. Thus each distribution takes a certain fraction, given by $\nu(\varepsilon)$, of the total mass of $\rho(t,z)$. For the uniform measure ($\alpha_1 = \frac ba = 1$), the factors $\nu(\pm 1)$ are both equal to $\frac 12$, so that the two distributions match and the distinction induced by the marking disappears.
	
	In order to test for a possible marking of the red particles, similar factors, which we will denote by $\mu(\varepsilon)$, can be introduced for the distributions of the red particles. These factors would possibly distinguish the two distributions relative to the even and odd particles, namely $\rho_{\rm ev}(t,z) = \mu(+1) \, \rho(t,z)$ and $\rho_{\rm odd}(t,z) = \mu(-1) \, \rho(t,z)$.

	The factors $\mu(\pm 1)$ and $\nu(\pm 1)$ can be computed in the simplest case, namely the first level, on which there is a single particle (at level 1, the particle is odd if its vertical position is even, and vice-versa). To do this, one numerically computes the position probability distribution for both the red and the blue particles, and for a large enough Aztec diamond. Since the level 1 particle density of the unmarked GUE-corners process is a simple Gaussian, one should obtain, for the large but finite Aztec diamond, a duplicated Gaussian-like curve or a single one, depending on whether there is a marking or not. By summing the values of the distribution over even or odd sites, the factors $\mu(\pm 1)$ and $\nu(\pm 1)$ themselves may be evaluated.
	 		
	Starting with the red particle, the distribution is given in Proposition~\ref{sec:Integration from the right, 2P}, where the numbers $L_{n,x}$ are computed from the recurrence (\ref{eq:def L_i,j}). A similar recurrence can be established for the blue particle, by using for instance the octahedron equation \cite{ruelle2022double} (see the Appendix~\ref{apx:appB} for more details). We have carried out the calculations for an Aztec diamond of order $n=300$, for two different values of $\sqrt\tau = \frac ab$, namely 1 and 2. The results are shown in Figure~\ref{fig:bluemarking}.
	
	\begin{figure}[t]
\begin{center}
\includegraphics[width=6.5cm]{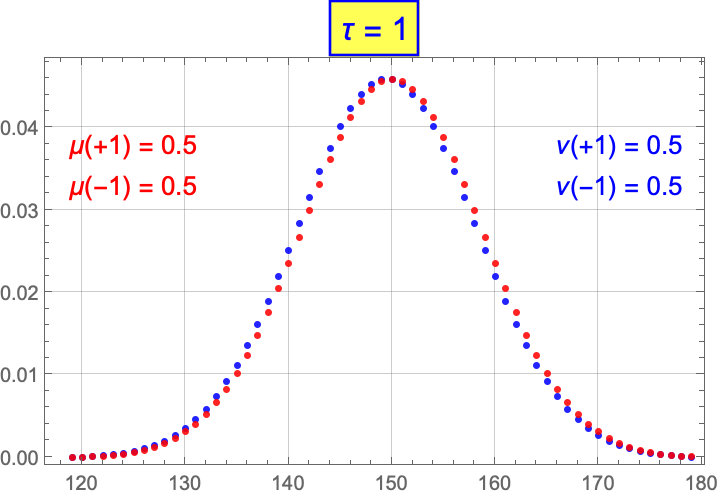}
\hspace{2cm}
\includegraphics[width=6.5cm]{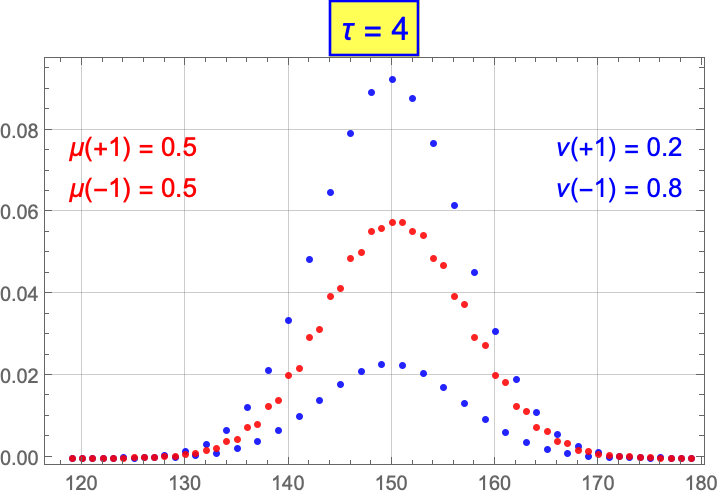}
\caption{Plots of the probability distributions of the position of the particle on the first level, for the red particle (in red) and for the blue particle (in blue), and for two values of $\tau = \big(\frac ab\big)^2 = 1,4$. The values of $\mu(\varepsilon)$ and $\nu(\varepsilon)$ displayed in the plots have been numerically computed by summing the corresponding distribution over the odd sites ($\varepsilon=+1$) or over the even sites ($\varepsilon=-1$).} 
\label{fig:bluemarking}
\end{center}
\end{figure}

	The plots on the left show no essential difference between the red and the blue particle, as expected. In both cases, the odd and the even positions follow the same Gaussian distribution, each one taking a mass equal to $\frac 12$. On the right plots, the distribution for the red particle again follows a single Gaussian and shows no marking, in agreement with the Proposition~\ref{sec:Integration from the right, 2P} (the distribution itself is nevertheless affected by the change of weights). However the blue plot clearly shows two Gaussian-like curves, the one corresponding to the even positions (odd particles) having a higher amplitude. Moreover the summation of the discrete distribution over the even resp. odd positions yield masses respectively equal to $\nu(-1)=\frac 45$ and $\nu(+1)=\frac 15$, exactly the values computed in \cite{berggren2026marked}, see Eq.~(\ref{nu}).
	
	Thus for the first level, the numerical results are perfectly consistent with the results of Section~\ref{sec:Two-periodic Aztec diamonds} for the red particle, and with the analytical results of \cite{berggren2026marked} for the blue particle.  The following heuristic argument provides an explanation for what brings the marking in one case, and no marking in the other. 
	
	The blue particle on the first level is always the right vertex of a face with two dimers around it, see the Figure~\ref{fig:redblue}. The dimer supporting the particle can have two different orientations, but in both cases, related by a simple flip, that face contributes a factor $b^2$ if the particle is even, $a^2$ if the particle is odd, whereas all other faces contribute the same factor since the flip does not affect them. It follows that there is a systematic bias by a factor $b^2$ for the even particles, and by a factor $a^2$ for the odd particles. Normalizing the full distribution to 1, one obtains that the even distribution is affected by a factor $\nu(+1) = \frac{b^2}{a^2+b^2} = \frac 1{1+\tau}$ while the odd one gets a factor $\nu(-1) = \frac{a^2}{a^2+b^2} = \frac \tau{1+\tau}$, reproducing the values in (\ref{nu}).
	
	The same argument applied to the red particle brings a different result. Indeed the red particle in the first level is always\footnote{This is indeed the case except if the particle is adjacent, but this situation happens with a probability going to zero when $n$ gets large, and may be neglected.} the left vertex of a face with two dimers around it. The flip of that face preserves the position of the particle, but the face contribution is $a^2$ in one case and $b^2$ in the other case, so that each pair of configurations related by the flip contributes a factor $a^2 + b^2$ no matter what the parity of the particle is. There is therefore no marking for the red particle.
	
	Building on this heuristics, we observe that the results for the blue and red particles would be interchanged if we set the weights $a,b$ to 1, and instead use the weights $c,d$, see Figure~\ref{fig:redblue}. This new weight system is gauge equivalent to the edge weights on the left panel of Figure~\ref{fig:edgeweights} by imposing a constraint different from that used above, namely $\alpha_1^{-1} = \beta_1 = \alpha_2 = \beta_2^{-1}$. The edge and face weights are related by $\alpha_1^{-1} = \frac cd \equiv \sqrt\sigma$.

	The faces which previously had a weight $a$ or $b$ get a weight 1, and those which had a weight 1 acquire a weight $c$ or $d$. It follows that the flippable face around which the blue particle is, no longer has weight $a$ or $b$, but has now weight 1. Then the contributions of the pairs of configurations related by a flip are proportional to $c^2 + d^2$, irrespective of its position. It is the converse situation for the red particle, which now lies around a flippable face with weight $c$ or $d$, depending of the parity of its position. One therefore expects a marking for the red particle, with a corresponding factor affecting the distributions equal to $\mu(+1) = \frac{\sigma}{1 + \sigma}$ for the  even distribution, and $\mu(-1) = \frac 1{1 + \sigma}$ for the odd one. In contrast we expect no marking for the blue particle, in agreement with the value of the coefficients $\nu(\varepsilon)$ quoted in Theorem 1.1 of \cite{berggren2026marked}. These expectations are borne out by numerical computations, whose plots are shown in Figure~\ref{fig:redmarking}.
	
	\bigskip
	
		\begin{figure}[h]
\begin{center}
\includegraphics[width=6.5cm]{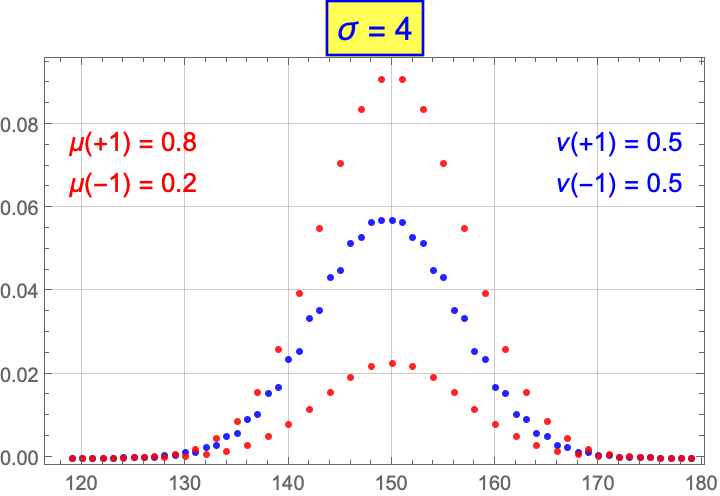}
\hspace{2cm}
\includegraphics[width=6.5cm]{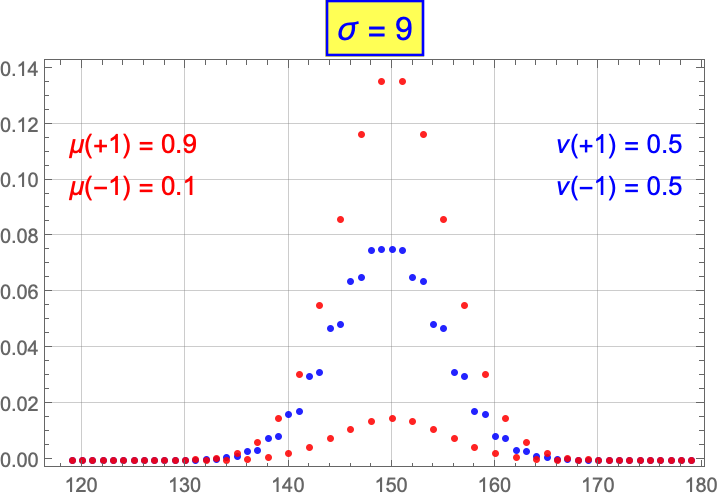}
\caption{Plots of the probability distributions of the position of the red and blue particles on the first level, for  two values of $\sqrt\sigma = \frac cd = 2,3$ and $n=300$.} 
\label{fig:redmarking}
\end{center}
\end{figure}

	Finally, in view of the previous results, one may expect that the red and the blue particles are both subjected to a marking in case we keep all weights $a,b,c,d$ generic (the probability distributions remain functions of $\tau$ and $\sigma$ only). One may in addition introduce a vertical bias by assigning an extra weight $\sqrt\lambda$ to all vertical edges (oriented along the NS axis), similarly to what was done in Section~\ref{sec:Biased Aztec diamonds}.
	
	Repeating the argument used above, we find, (1) for the red particle, that the systematic bias factor is $c^2(\lambda a^2 + b^2)$ if the particle is even, or $d^2(a^2 + \lambda b^2)$ if it is odd, and (2) for the blue particle, that there is a systematic bias by a factor $b^2(c^2 + \lambda d^2)$ if the particle is even, or by a factor $a^2(\lambda c^2 + d^2)$ if it is odd. Normalizing the probability distribution to 1, we expect that the marking indeed affects the red and the blue particles, with the following mass values taken by the even and odd particles,
	\begin{equation}
	\mu(\varepsilon) = \frac{\sigma(1 + \lambda \tau) \,\delta_{\varepsilon,+1} + (\tau + \lambda ) \, \delta_{\varepsilon,-1}}{\sigma(1 + \lambda \tau) + (\tau + \lambda )}, \qquad
\nu(\varepsilon) = \frac{(\sigma + \lambda ) \,\delta_{\varepsilon,+1} + \tau(1 + \lambda \sigma) \, \delta_{\varepsilon,-1}}{(\sigma + \lambda) + \tau(1 + \lambda \sigma)}.
\label{fractions}
\end{equation}

	These values are confirmed by the numerical calculations shown in Figure~\ref{fig:redbluemarking}, for two different triplets, $(\sqrt\tau,\sqrt\sigma,\lambda) = (3,2,2)$ and $(\frac 43,\frac 52,3)$, and $n=300$. From (\ref{fractions}), the marking factors $\mu(\varepsilon)$ and $\nu(\varepsilon)$ in the two cases are given respectively by,
	\begin{alignat}{4}
    \mu(+1) &= \tfrac{76}{87},\qquad &  \mu(-1) &= \tfrac{11}{87},\qquad & \nu(+1) &= \tfrac{2}{29},\qquad  &\nu(-1) &= \tfrac{27}{29}, \nonumber
    \\[0.2cm]
    \mu(+1) &= \tfrac{1425}{1597},\qquad & \mu(-1) &= \tfrac{172}{1597},\qquad & \nu(+1) &= \tfrac{333}{1597},\qquad & \nu(-1) &= \tfrac{1264}{1597},
\end{alignat}
and perfectly match the numerical values quoted in the plots of Figure~\ref{fig:redbluemarking}. We also note that the maximum of the Gaussian-like curves, corresponding to the contact point between the arctic curve and the boundary of the Aztec diamond, gets shifted to higher values as $\lambda$ increases.

	The above results only concern the first level, but there is little doubt that the same marking is present on all higher levels. The next subsection gives further evidence that it is indeed the case (for the red particle system).

	\bigskip
	
		\begin{figure}[t]
\begin{center}
\includegraphics[width=6.5cm]{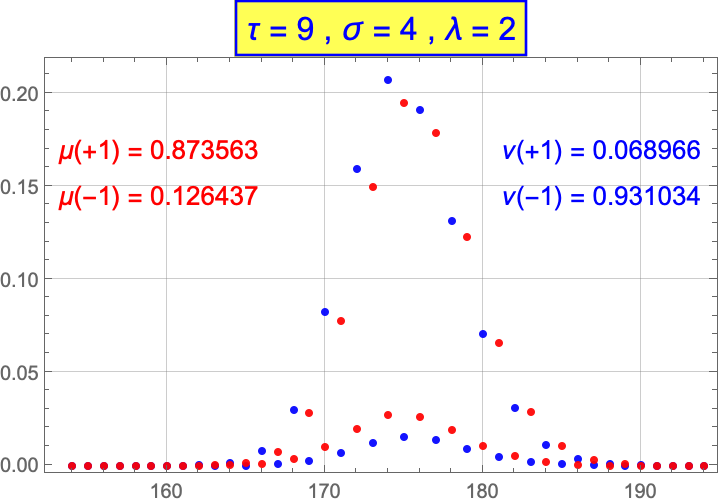}
\hspace{2cm}
\includegraphics[width=6.5cm]{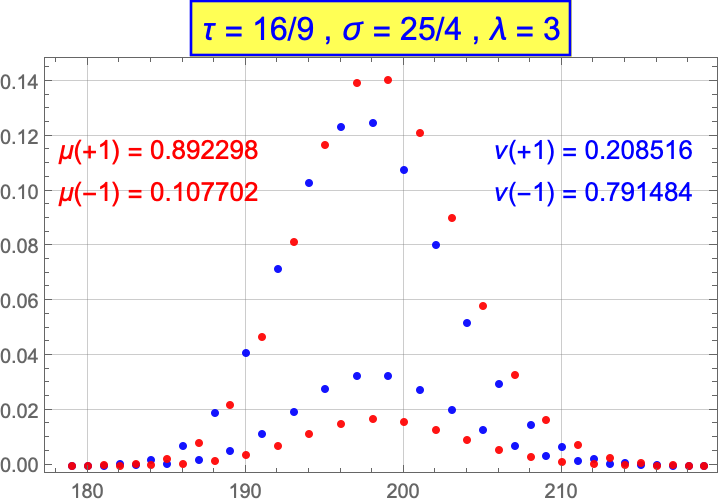}
\caption{Plots of the probability distributions of the position of the red and blue particles on the first level, for  the parameters $(\sqrt\tau,\sqrt\sigma,\lambda) = (3,2,2),(\frac 43,\frac 52,3)$ and $n=300$.} 
\label{fig:redbluemarking}
\end{center}
\end{figure}

	\subsection{General two-periodic weights with bias}
	\label{sec:two-periodicwith bias}
	
	
	In this section, we adapt the approach we followed in Section~\ref{sec:Two-periodic Aztec diamonds} to the more general measure of the previous subsection, with four alternating face weights $a,b,c,d$ and a vertical bias $\lambda$. Although we cannot push the analysis as far as to compute the distributions over the first $m$ levels (neither discrete nor continuous), we clearly see that the marking $\mu(\varepsilon)$, as written in (\ref{fractions}), appears on all levels. To do this, we compute the measure on the complete set of particles induced by the periodic and biased measure, much in the same way as we did in Section~\ref{sec:Two-periodic Aztec diamonds}.
			
	\begin{figure}[!ht]
		\centering
		\begin{tikzpicture}[scale=0.8,rotate=45]
			\def\n{8} 
			
			\foreach \x in {1,...,\n}
			\draw[gray!70!white, very thin] (\x-1,\n-\x) rectangle +(2*\n-2*\x+1,2*\x -1)
			;
			
			\filldraw[SkyBlue!70!white, opacity=0.25]
			(4,11) rectangle +(1,-1)
			(4,9) rectangle +(1,-1)
			(6,9) rectangle +(1,-1)
			(5,6) rectangle +(1,-1)
			(9,10) rectangle +(1,-1)
			(9,6) rectangle +(1,-1)
			(9,4) rectangle +(1,-1);

			\begin{scope}[violet,dashed,line width=0.1pt, opacity=0.8]
				\foreach \x in {1,...,\n}
				\draw (\x-1.2,\n-\x-0.2) node[anchor=north] {$\ell$=\x} 
				;
			\end{scope}

			\begin{scope}[violet,dashed,line width=0.1pt, opacity=0.8]
				\foreach \x in {0,...,\n}
				\draw (\x-0.5,\n+\x-0.5) node[anchor=east]{$k$=\x \;\;\;} 
				;
			\end{scope}

			\foreach \x in {2,4,...,\n}{
				\foreach \y in {2,4,...,\n}{
					\filldraw[draw=white,fill=white]
					(-0.5+\n-\x+\y,0.5+\x+\y-4) circle(7pt)
					(-0.5+\n-\x+\y,2.5+\x+\y-4) circle(7pt);
					\draw[draw=white,fill=white]
					(-1.5+\n-\x+\y,1.5+\x+\y-4) circle(7pt)
					( 0.5+\n-\x+\y,1.5+\x+\y-4) circle(7pt);
			}}
			
			\foreach \x in {2,4,...,\n}{
				\foreach \y in {2,4,...,\n}{
					\draw
					(-0.5+\n-\x+\y,0.5+\x+\y-4) node {$b$}
					(-0.5+\n-\x+\y,2.5+\x+\y-4) node {$b$};
					\draw
					(-1.5+\n-\x+\y,1.5+\x+\y-4) node {$a$}
					( 0.5+\n-\x+\y,1.5+\x+\y-4) node {$a$};     
			}}
			\foreach \x in {0,2,...,\n}{
				\foreach \y in {0,2,...,\n}{
					\draw
					(-0.5+\n-\x+\y,-0.5+\x+\y) node {$c$};
			}}
			\foreach \x in {2,4,...,\n}{
				\foreach \y in {2,4,...,\n}{
					\draw
					(-0.5+\n-\x+\y,-2.5+\x+\y) node {$c$};
			}}
			\foreach \x in {0,2,...,\n}{
				\foreach \y in {2,4,...,\n}{
					\draw
					(-1.5+\n-\x+\y,-1.5+\x+\y) node {$d$}
					(0.5+\n-\y+\x,-1.5+\x+\y) node {$d$};
			}}
			
			\filldraw[red!90!black,line width=5pt]
			(4,11)--+(1,0)
			(6,9)--+(1,0)
			(8,11)--+(1,0)
			(10,11)--+(1,0)
			(4,3)--+(1,0)
			(6,3)--+(1,0)
			(10,7)--+(1,0)
			(6,1)--+(1,0)
			(8,1)--+(1,0)
			;
			
			\filldraw[orange!90!black,line width=5pt]
			(7,12)--+(0,-1)
			(5,6)--+(0,-1)
			(9,8)--+(0,-1)
			(11,10)--+(0,-1)
			(9,6)--+(0,-1)
			(13,10)--+(0,-1)
			(11,6)--+(0,-1)
			(13,8)--+(0,-1)
			(11,4)--+(0,-1)
			(13,6)--+(0,-1)
			(15,8)--+(0,-1)
			;
			
			\filldraw[blue!90!black,line width=5pt]
			(3,6)--+(1,0)
			(9,10)--+(1,0)
			(5,4)--+(1,0)
			(5,2)--+(1,0)
			(7,2)--+(1,0)
			(9,4)--+(1,0)
			(7,0)--+(1,0)
			(9,2)--+(1,0)
			;
			
			\filldraw[violet!90!black,line width=5pt]
			(4,9)--+(0,-1)
			(4,5)--+(0,-1)
			(12,11)--+(0,-1)
			(12,9)--+(0,-1)
			(12,7)--+(0,-1)
			(14,9)--+(0,-1)
			(12,5)--+(0,-1)
			(14,7)--+(0,-1)
			;
			
			\filldraw[red!25!white,line width=5pt]
			(5,12)--+(-1,0)
			(7,14)--+(-1,0)
			(5,10)--+(-1,0)
			(9,14)--+(-1,0)
			(7,10)--+(-1,0)
			(9,12)--+(-1,0)
			(7,8)--+(-1,0)
			(11,12)--+(-1,0)
			(11,8)--+(-1,0)
			;
			
			\filldraw[orange!25!white,line width=5pt]
			(0,7)--+(0,1)
			(2,9)--+(0,1)
			(2,7)--+(0,1)
			(6,11)--+(0,1)
			(2,5)--+(0,1)
			(8,9)--+(0,1)
			(6,5)--+(0,1)
			(8,7)--+(0,1)
			(8,5)--+(0,1)
			(8,3)--+(0,1)
			(10,5)--+(0,1)
			;
			
			\filldraw[blue!25!white,line width=5pt]
			(6,13)--+(-1,0)
			(8,15)--+(-1,0)
			(8,13)--+(-1,0)
			(4,7)--+(-1,0)
			(10,13)--+(-1,0)
			(6,7)--+(-1,0)
			(10,9)--+(-1,0)
			(10,3)--+(-1,0)
			;
			
			\filldraw[violet!25!white,line width=5pt]
			(1,8)--+(0,1)
			(3,10)--+(0,1)
			(1,6)--+(0,1)
			(3,8)--+(0,1)
			(5,8)--+(0,1)
			(3,4)--+(0,1)
			(7,6)--+(0,1)
			(7,4)--+(0,1)
			;
			
			\filldraw[fill=yellow!70!white, draw=black]
			(7,12) circle(3.5pt)
			(3,6) circle(3.5pt)
			(9,10) circle(3.5pt)
			(5,6) circle(3.5pt)
			(5,4) circle(3.5pt)
			(9,8) circle(3.5pt)
			(11,10) circle(3.5pt)
			(5,2) circle(3.5pt)
			(9,6) circle(3.5pt)
			(13,10) circle(3.5pt)
			(7,2) circle(3.5pt)
			(9,4) circle(3.5pt)
			(11,6) circle(3.5pt)
			(13,8) circle(3.5pt)
			(7,0) circle(3.5pt)
			(9,2) circle(3.5pt)
			(11,4) circle(3.5pt)
			(13,6) circle(3.5pt)
			(15,8) circle(3.5pt)
			;
			
			\filldraw[fill=green!70!white, draw=black]
			(4,11) circle(3.5pt)
			(4,9) circle(3.5pt)
			(6,9) circle(3.5pt)
			(8,11) circle(3.5pt)
			(10,11) circle(3.5pt)
			(4,5) circle(3.5pt)
			(4,3) circle(3.5pt)
			(12,11) circle(3.5pt)
			(6,3) circle(3.5pt)
			(10,7) circle(3.5pt)
			(12,9) circle(3.5pt)
			(6,1) circle(3.5pt)
			(12,7) circle(3.5pt)
			(14,9) circle(3.5pt)
			(8,1) circle(3.5pt)
			(12,5) circle(3.5pt)
			(14,7) circle(3.5pt) 
			;
			
		\end{tikzpicture}
		\caption{The general biased two-periodic weighting. The dimers are colored according to the weight they carry. Sky blue squares contain the non-adjacent particles and the corresponding pair of dimers that can be flipped without modifying the particle configuration. Even particles are coloured in yellow and odd ones in green.}
		\label{fig:complete weighting}
	\end{figure}
	
	From the Figure~\ref{fig:complete weighting}, we can distinguish four types of edges according to the weight they carry : $ac\sqrt\lambda$ in orange, $bd\sqrt\lambda$ in violet, $ad$ in red, $bc$ in blue. In the case of adjacent particles (in the sense of \eqref{eq:ajdacency}), the weight of the corresponding dark dimer is fully determined by the four adjacency subclasses \eqref{eq:def of adjacency 2P}. Moreover, the light dimers of each species come in equal number as their dark partners since any flip preserves their relative amount. The amount of each dimer types readily follows :
	\begin{align}
		&\#\big\{(ac\sqrt{\lambda})\text{--dimers}\big\} = 2\alpha_{\text{sup}}^\text{e}(x^{(1)},\dots,x^{(n)}),
		&\#\big\{(bc)\text{--dimers}\big\} = 2\alpha_{\text{inf}}^\text{e}(x^{(1)},\dots,x^{(n)}),
		\\
		&\#\big\{(bd\sqrt{\lambda})\text{--dimers}\big\} = 2\alpha_{\text{sup}}^\text{o}(x^{(1)},\dots,x^{(n)}),
		&\#\big\{(ad)\text{--dimers}\big\} = 2\alpha_{\text{inf}}^\text{o}(x^{(1)},\dots,x^{(n)}).
	\end{align}
	For the particles that are non-adjacent (n.a.p.), they either contribute a factor $c^2(a^2\lambda+b^2)$ if they are even, or $d^2(a^2 + b^2\lambda)$ if they are odd. Hence, from the formula \eqref{eq:pdf from AD to particles}, we find the distribution describing the whole particle system
	\begin{multline}
		\rho(x^{(1)},\ldots,x^{(n+1)}) = \frac{1}{Z_n(a,b,c,d|\lambda)}(a^2c^2\lambda)^{\alpha_{\text{sup}}^\text{e}}\;
		(b^2d^2\lambda)^{\alpha_{\text{sup}}^\text{o}}\;
		(b^2c^2)^{\alpha_{\text{inf}}^\text{e}}\;
		(a^2d^2)^{\alpha_{\text{inf}}^\text{o}} \\[0.1cm]
		(a^2c^2\lambda+b^2c^2)^{\#\text{even n.a.p}}\;
		(b^2d^2\lambda+a^2d^2)^{\#\text{odd n.a.p}}\;
		\chit(x^{(1)}\prec\dots\prec x^{(n+1)}).
	\end{multline}
	where the parity of a particle on level $\ell$ and vertical position $j$ refers to the value of $\varepsilon = (-1)^{\ell+j}$. By using $\tau=a^2/b^2$ and $\sigma=c^2/d^2$ as before, one may simplify the expression of the distribution to
	\begin{align}
		\rho&(x^{(1)},\ldots,x^{(n+1)}) = \frac 1{\widetilde{Z}_n(\tau,\sigma|\lambda)}
		\;
		\frac{\lambda^{\alpha_{\text{sup}}} \tau^{\alpha_{\text{sup}}^\text{e}+\alpha_{\text{inf}}^\text{o}}}{(\tau+\lambda)^{\alpha^\text{o}}(1 + \lambda\tau)^{\alpha^\text{e}}}
		\big(\sigma (1 + \lambda\tau)\big)^{\#\,{\rm even}}\big(\tau+\lambda\big)^{\#\,{\rm odd}}
		\chit(x^{(1)}\prec\dots\prec x^{(n+1)}), \nonumber \\	
		&= \frac {\Big[\sigma(1 + \lambda\tau) + (\tau+\lambda)\Big]^{\frac {n(n+1)}2}}{\widetilde{Z}_n(\tau,\sigma|\lambda)}
		\;
		\frac{\lambda^{\alpha_{\text{sup}}} \tau^{\alpha_{\text{sup}}^\text{e}+\alpha_{\text{inf}}^\text{o}}}{(\tau+\lambda)^{\alpha^\text{o}}(1 + \lambda\tau)^{\alpha^\text{e}}}
		\:\big[\mu(+1)\big]^{\#\,{\rm even}}\:\big[\mu(-1)\big]^{\#\,{\rm odd}}
		\chit(x^{(1)}\prec\dots\prec x^{(n+1)}),
		\label{eq:density}
	\end{align}
	where $\widetilde{Z}_n(\tau,\sigma|\lambda) = Z_n(\sqrt\tau,1,\sqrt\sigma,1|\lambda) = (bd)^{-n(n+1)} Z_n(a,b,c,d|\lambda)$, the marking factors $\mu(\pm 1)$ being those in (\ref{fractions}). In this expression, the second factor is the generalization of the function $\Omega(x^{(1)},\ldots,x^{(n)})$ in \eqref{eq:Adjacency factor}. As it is neglected in the scaling limit, it does not contribute to the marking. 
	
	The powers of $\mu(\pm 1)$ are new and correspond exactly to the trace left by the marking for the complete system of particles. To complete the proof, one should now try to repeat the procedure used in Section~\ref{sec:Two-periodic Aztec diamonds} to sum over the positions of all particles but those on the first $m$ levels, and verify that, in the scaling limit, the marking appears in the expected way, leaving a continuous distribution given by
	\begin{equation}
		\rho(z^{(1)},\ldots,z^{(m)}) = \frac{1}{(2\pi)^{m/2}} \left[\prod_{\ell=1}^m \prod_{i=1}^\ell \mu\Big(\varepsilon(x_i^{(\ell)})\Big)\right] \, \Delta\big(z^{(m)}\big) \, \prod_{j=1}^m\text{e}^{-\frac{1}{2}(z_j^{(m)})^2} \, 
		\chit(z^{(1)}\prec\dots\prec z^{(m)}).
	\end{equation}

	However, and altough we believe that the method used in Section~\ref{sec:Two-periodic Aztec diamonds} for summing over the last levels should still apply, the presence of the marking factors in the expression (\ref{eq:density}) makes the procedure harder. The calculations carried out in Section~\ref{sec:Two-periodic Aztec diamonds} appear rather tedious to generalize, and even if not completely hopeless, this approach may not be the most convenient and straightforward, given the high level of complexity of the discrete distributions compared to the continuous distributions they give rise to.
	
	We emphasize that even though there are three free parameters in the model, the intensity of the marking only depends on one specific combination of them, namely $\frac{\mu(+1)}{\mu(-1)} = \frac{\sigma (1+\lambda\tau)}{\tau+\lambda}$. Whenever this ratio is equal to 1, the marking becomes invisible (for the red particles), a special case of this being $\sigma = \lambda = 1$.

	\section{Discussion}
	
	In this work, we showed that the GUE-corners distribution still appears as the limiting distribution in Aztec diamonds near the contact point between the arctic curve and the edge of the domain when the probability measure on the configuration set is modified. More precisely, we examined the biased (or one-periodic) and the two-periodic measures. The main ingredients that we used are the determinantal structure of the particle process and the asymptotic Gaussian behaviour of the determinant entries. We believe these hypotheses are sufficient to observe the convergence to the GUE-corners process in Aztec diamonds with any non-pathological probability measure, even though a deeper analysis would be required to settle this question. On the other hand, these hypotheses are certainly not necessary, since there exist discrete non-determinantal point processes (for instance, the uniformly weighted alternating sign matrices \cite{gorin2014alternating}) that also converge to GUE-corners. Under such considerations, it would not be surprising that further generalizations of the weighting system on Aztec diamonds remain associated with GUE, for instance, higher periodicity or two-periodic with an additional vertical bias, although fully explicit calculations in these cases would probably become difficult to perform with our method.
	
	Building on the findings of Berggren and Bradinoff \cite{berggren2026marked} who observed that a $2 \times \ell$ periodic measure may give rise to a marked GUE corners process, we have clarified, in the $2 \times 2$ periodic case, the conditions under which a marking occurs, and also heuristically determined, for a general two-periodic measure, what the marking factors are.
	
	It is also known that partition functions of Aztec diamonds are related to the octahedron recurrence \cite{kuo2004applications,speyer2007perfect}, and one natural question would be to know whether the distribution of particles $\rho(x^{(1)},\dots,x^{(m)})$ also share connections with this recurrence. Since the density $\rho$ gives the probability to observe a configuration while the position of a certain number of particle is prescribed, it is easily related to the notion of \textit{refined partition functions} (enumeration of weighted configurations with the same constraints). These functions have to satisfy a modified octahedron system (as it is shown in \cite{ruelle2022double} for the one-refined partition function). However, when one considers a refinement on several levels, the equations produced by the octahedron recurrence do not simplify very much, and although it effectively allows us to compute $\rho(x^{(1)},\dots,x^{(m)})$ recursively, the system of equations becomes rapidly complicated as $m$ increases and it seems that no general formula for $\rho$ can be extracted from it.

	One last peculiarity that we observed is that \textit{multi-refined partition functions} (in the sense that the particle positions on the $m$ first levels are prescribed) appear to be proportional to the determinant of one-refined partition functions. In terms of the probability distribution, it is equivalent to observe that the $m$-level correlation function for Aztec diamonds of order $n$ is proportional to a determinant of $1$-level correlation functions of ADs of smaller size,
	\begin{equation}
		\rho_n(x^{(1)},\dots,x^{(m)}) \sim \det_{1\le i,j \le m} \left[\rho_{n-m+i}(x_j^{(m)})\right] \Omega(x^{(1)},\dots,x^{(m-1)}) \chit(x^{(1)}\prec\dots\prec x^{(m)}).
	\end{equation}
	The octahedron equation provide a partial answer to this observation since it shows that $m$-refined partition functions can always be related to the $m-1$ ones, but it does not really shed any light on the underlying structure. A similar construction has already been observed in a slightly different context \cite{debin2021factorization}, even though the details of this factorization remain unclear to us. It would be interesting to investigate whether there exists a general structure behind it.

	\section*{Acknowledgements}
	
	The authors are grateful to Tomas Berggren, Nedialko Bradinoff, Christophe Charlier and Tom Claeys for valuable comments and discussions. This work was supported by the Fonds de la Recherche Scientifique - FNRS. PR is a Senior Research Associate of FSR-FNRS (Belgian Fund for Scientific Research). NR is supported by the Belgian FRIA grant FC50105.

	\begin{appendices}
		
		\section{Alternating sum of powers}
		\label{apx:Appendix sum of powers}
		
		
		In this appendix, we prove the formula 
		\begin{equation}
			\sum_{x=y_1}^{y_2} \varepsilon_\ell(x)\, x^{j-1} \,\Omega(x)
			= 
			-v\left( y_2^{j-1}- y_1^{j-1}\right) - \sum_{k=2}^{j} \frac{B_k}{j}\binom{j}{k} (2^k-1) \left( \varepsilon_{\ell+1}(y_2)\,y_2^{j-k} - \varepsilon_{\ell+1}(y_1)\,y_1^{j-k}\right).
			\label{eq:alt sum powers with epsilon}
		\end{equation}

		\begin{lemma}
			For $y$ and $j$ strictly positive integers, we have
			\begin{equation}
				\sum_{x=0}^{y} (-1)^x x^{j-1} 
				=
				(-1)^y \left[\frac{y^{j-1}}{2}+\sum_{k=2}^{j} \frac{B_k}{j}\binom{j}{k}(2^k-1) y^{j-k} \right] - \frac{B_{j}}{j}(2^{j}-1).
				\label{eq:alt sum powers}
			\end{equation}
			\label{lemma:alternate sum powers}
		\end{lemma}
		
		\begin{proof}
			We compute the exponential generating function by multiplying by $z^{j-1}/(j-1)!$ both sides of the equality. The lhs gives 
			\begin{equation}
				\sum_{j=1}^{\infty}\sum_{x=0}^{y} (-1)^x \frac{(zx)^{j-1}}{(j-1)!} 
				=
				\sum_{x=0}^{y} (-1)^x \e^{zx}
				=
				\frac{1-(-\e^{z})^{y+1}}{1+\e^{z}}.
			\end{equation}
			While we find for the right-hand side
			\begin{equation}\begin{aligned}
					&\frac{(-1)^y}{2} \sum_{j=1}^\infty \frac{(zy)^{j-1}}{(j-1)!} + (-1)^y \sum_{k=2}^\infty \frac{B_k}{k!}(2^k-1)z^{k-1} \sum_{j=0}^\infty \frac{(zy)^{j}}{j!} - \sum_{j=1}^\infty \frac{B_j^{-}}{j!}(2^k-1)z^{j-1}
					\\[0.2cm] &=
					\frac{(-1)^y}{2}\e^{zy} \big(1+\tanh(z/2)\big) + \frac{1}{2} \big(1-\tanh(z/2)\big).
			\end{aligned}\end{equation}
			Using the exponential definition of $\tanh$, it is immediate to prove the equality between the two generating functions.
		\end{proof}
		
		We recall that $\varepsilon_\ell(x)=(-1)^{\ell+x}=-\varepsilon_{\ell+1}(x)$ and 
		\begin{equation}
			\Omega(x) 
			=
			\left\{\begin{array}{ll}
				\Omega_{\text{sup}} = \frac{1}{2} +v \varepsilon_{\ell+1}(y_2) & \text{if } x=y_2 ,
				\\[0.15cm]
				\Omega_{\text{inf}} = \frac{1}{2} -v \varepsilon_{\ell+1}(y_1) & \text{if } x=y_1 ,
				\\[0.15cm]
				1 & \text{otherwise.}
			\end{array}\right.
		\end{equation}
		Finally, to evaluate the desired sum, one can separate the terms as
		\begin{equation}
			\sum_{x=y_1}^{y_2} \varepsilon_\ell(x)\, x^{j-1} \,\Omega(x) 
			=
			\sum_{x=0}^{y_2} \varepsilon_{\ell}(x)\, x^{j-1}  
			-\sum_{x=0}^{y_1} \varepsilon_{\ell}(x)\, x^{j-1}  
			+ \Omega_{\text{inf}}\,\varepsilon_{\ell}(y_1)\, y_1^{j-1}
			+ (\Omega_{\text{sup}}-1)\,\varepsilon_{\ell}(y_2)\, y_2^{j-1},
		\end{equation}
		and use the lemma above to obtain the claimed formula \eqref{eq:alt sum powers with epsilon}.

		\section{Particle position distribution on the first level}
		\label{apx:appB}
		
		We give some details about the numerical computation of the position probability distribution for the red or blue particle on the first level. 
		
		\begin{figure}[hb]
\begin{center}
\includegraphics[width=6cm]{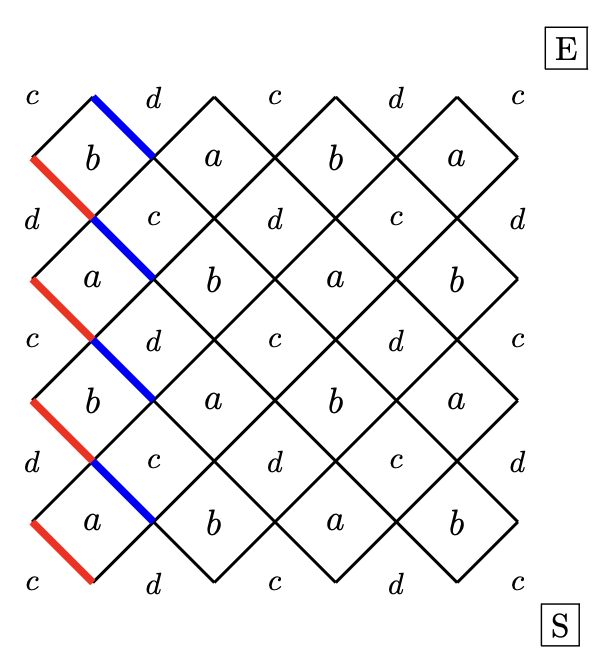}
\caption{System of weights: the latin letters refer to face weights, the vertical red and blue edges get a weight $x$ and $y$ respectively. In addition all vertical edges have a weight $\sqrt\lambda$.} 
\label{fig:octa}
\end{center}
\end{figure}

		The main tool is the octahedron equation, which we briefly recall (see \cite{ruelle2022double} for similar computations). We consider the generic two-periodic face weights shown in Figure~\ref{fig:octa}, as well as the usual vertical bias, namely a vertical edge gets a weight $\sqrt\lambda$. In addition, the vertical edges in red get a extra weight $x$, while the blue egdes get a extra weight $y$.
		
		The measure on the dimer configurations has a slightly different definition from what it is in the text. Namely, a dimer gets the weight of the edge it covers, as usually, and a face with weight $w$ contributes the weight of a configuration by the factor\footnote{When probabilities are considered, the factor $w^{1-N_d}$ can be replaced by $w^{-N_d}$. It follows that the probabilities computed with the usual measure and those computed relative to the present measure are related by an inversion of the face weights.} $w^{1-N_d}$ where $N_d$ is the number of dimers which are adjacent to that face. For the Aztec diamond of order $n$, we denote by $T_n(a,b,c,d|\lambda,x,y)$ the partition function relative to this measure. Then the following recurrence relation holds \cite{kuo2004applications,speyer2007perfect},
		\begin{align}
		T_n(a,b,c,d|\lambda,x,y) \: T_{n-2}(a,b,c,d|\lambda) = \:&T_{n-1}(a,b,c,d|\lambda,x,y) \: T_{n-1}(a,b,c,d|\lambda) \nonumber \\
		\noalign{\medskip}
		& \hspace{1cm} + x \lambda \: T_{n-1}(b,a,d,c|\lambda,x,y) \: T_{n-1}(b,a,d,c|\lambda),
		\label{eq:octarec}
		\end{align}
		with the following initial conditions, 
		\begin{equation}
		T_0(a,b,c,d|\lambda,x,y) = c, \qquad T_1(a,b,c,d|\lambda,x,y) = \frac{c^2 + \lambda x y \, d^2}a.
				\end{equation}
In the above recurrence relation, the absence of the parameters $x,y$ indicate that the corresponding partition function are computed with $x=y=1$.

The partition function is a polynomial in $x$, $T_n(a,b,c,d|\lambda,x,y) = \sum_{k=0}^n x^k \: T_{n,k}(a,b,c,d|\lambda,y)$, where $T_{n,k}(a,b,c,d|\lambda,y)$ is the partition function for those configurations which have exactly $k$ dimers on red edges, or equivalently, which have a red particle at vertical position $k$ on level 1. Extracting the coefficient of $x^k$ in (\ref{eq:octarec}), and then dividing by $T_n(a,b,c,d|\lambda) \, T_{n-2}(a,b,c,d|\lambda)$ yields a first order recurrence relation for $S_{n,k}(a,b,c,d|\lambda,y) \equiv T_{n,k}(a,b,c,d|\lambda,y)/T_{n}(a,b,c,d|\lambda)$,
\begin{equation}
S_{n,k}(a,b,c,d|\lambda,y) = R_{n-1}(a,b,c,d|\lambda) \, S_{n-1,k}(a,b,c,d|\lambda,y) + \big(1 - R_{n-1}(a,b,c,d|\lambda)\big) \, S_{n-1,k-1}(b,a,d,c,|\lambda,y),
\label{eq:Srec1}
\end{equation}
where 
\begin{align}
R_{n-1}(a,b,c,d|\lambda) &= \frac{T_{n-1}^2(a,b,c,d|\lambda)}{T_n(a,b,c,d|\lambda) \, T_{n-2}(a,b,c,d|\lambda)}, \\ \noalign{\medskip}
1 - R_{n-1}(a,b,c,d|\lambda) &= \lambda \, \frac{T_{n-1}^2(b,a,d,c|\lambda)}{T_n(a,b,c,d|\lambda) \, T_{n-2}(a,b,c,d|\lambda)}.
\end{align}
The second identity follows from the octahedron recurrence (\ref{eq:octarec}) for $x = y = 1$. The recurrence (\ref{eq:Srec1}) allows to compute the $S_{n,k}(a,b,c,d|\lambda)$ for $n \ge 2$ from the initial values
\begin{equation}
S_{1,0}(a,b,c,d|\lambda,y) = \frac{c^2}{a T_1(a,b,c,d|\lambda)} = \frac \sigma{\sigma + \lambda}, \qquad 
S_{1,1}(a,b,c,d|\lambda,y) = \frac{y\lambda d^2}{a T_1(a,b,c,d|\lambda)} = \frac {\lambda y}{\sigma + \lambda},
\label{eq:Sinit}
\end{equation}
provided the quantities $R_{n-1}(a,b,c,d|\lambda)$ are known. These could be computed from the octahedron recurrence itself, for $x=y=1$, but the analytical expressions get very quickly extremely convoluted, and their numerical evaluation leads to very large numbers. By using the condensation algorithm to compute $\lambda$-determinants, it has been shown in \cite{ruelle2023lambda} that $R_{n-1}(a,b,c,d|\lambda)$ can be directly computed from the terms of an intermediate sequence, namely,
\begin{equation}
R_{n-1}(a,b,c,d|\lambda) = \frac 1{1 + \lambda \, r_{n-1}^2}, \quad \text{with } \quad r_k = \frac{\lambda + r_{k-1}^2}{r_{k-2}\,(1 + \lambda r_{k-1}^2)}, \quad r_{-1} = \sqrt\tau, \;\: r_0 = \sqrt\sigma.
\label{eq:r_rec}
\end{equation}
Because the terms in the sequence $\{r_k\}_k$ and the initial conditions (\ref{eq:Sinit}) are functions of $\tau,\sigma,\lambda$ only, so are $R_{n-1}(a,b,c,d|\lambda)$ and $S_{n,k}(a,b,c,d|\lambda,y)$, which in addition depend on $y$. To summarize, the quantities $S_{n-1}(\tau,\sigma|\lambda,y)$, for fixed $\tau,\sigma,\lambda$, can be recursively computed as follows,
\begin{subequations}
\begin{align}
&S_{n,k}(\tau,\sigma|\lambda,y) = \frac 1{1 + \lambda \, r_{n-1}^2} \, S_{n-1,k}(\tau,\sigma|\lambda,y) + \frac {\lambda \, r_{n-1}^2}{1 + \lambda \, r_{n-1}^2} \, S_{n-1,k-1}(\tau^{-1},\sigma^{-1}|\lambda,y),\\
\noalign{\smallskip}
&S_{1,0}(\tau,\sigma|\lambda) = \frac \sigma{\sigma + \lambda}, \qquad S_{1,1} = \frac {\lambda y}{\sigma + \lambda}, \qquad S_{n,k}(\tau,\sigma|\lambda) = 0 \quad \text{for } k \not\in \{0,1,2,\ldots,n\},
\end{align}
\label{eq:S_rec}
\end{subequations}
where the terms $r_k$ obey the recurrence given in (\ref{eq:r_rec}).

The refined partition function $T_{n,k}(a,b,c,d|\lambda,y)$ is a polynomial of degree one in $y$, since the dimer supporting the red particle can have two orientations: either downwards along the blue edge, in  which case the blue particle on level 1 has vertical position $k-1$, or upwards, in which case the blue particle has vertical position $k$. Expanding in powers of $y$ yields
\begin{equation}
T_{n,k}(a,b,c,d|\lambda,y) = T_{n,k,+}(a,b,c,d|\lambda) + y \: T_{n,k,-}(a,b,c,d|\lambda).
\end{equation}

It follows that $S_{n,k}(\tau,\sigma|\lambda,y)$ itself is a polynomial of degree one in $y$,
\begin{equation}
S_{n,k}(\tau,\sigma|\lambda,y) = S_{n,k,+}(\tau,\sigma|\lambda) + y \: S_{n,k,-}(\tau,\sigma|\lambda),
\end{equation}
with $S_{n,0,-}(\tau,\sigma|\lambda) = 0$ and $S_{n,n,+}(\tau,\sigma|\lambda) = 0$.
The two coefficients readily give a direct way to compute the seeked probabilities. Recalling that under the usual two-periodic measure, a face with weight $w$ contributes a factor $w^{N_d}$ (rather than $w^{-N_d}$ like here, see the last footnote), the probability that the red particle on the first level $\ell=1$ is at a vertical position $k$ is equal to
\begin{equation}
{\mathbb P}_n[\text{red particle at position }k] = S_{n,k,+}(\tau^{-1},\sigma^{-1}|\lambda) + S_{n,k,-}(\tau^{-1},\sigma^{-1}|\lambda), \qquad 0 \le k \le n,
\end{equation}
while the probability that the blue particle on the first level $t=1$ is at a vertical position $k$ is given by
\begin{equation}
{\mathbb P}_n[\text{blue particle at position }k] = S_{n,k,+}(\tau^{-1},\sigma^{-1}|\lambda) + S_{n,k+1,-}(\tau^{-1},\sigma^{-1}|\lambda), \qquad 0 \le k \le n-1.
\end{equation}
Numerically, the recurrences (\ref{eq:r_rec}) and (\ref{eq:S_rec}) are straightforward to implement and lead to the plots of Section~\ref{sec:marking}.

	\end{appendices}

	\addcontentsline{toc}{section}{References}

\end{document}